\documentclass[final,3p,times]{elsarticle}

\journal{Theoretical Computer Science}

\usepackage{amsmath}
\usepackage{amsfonts}
\usepackage{amsthm}
\usepackage{amssymb}

\usepackage{mathtools}
\usepackage{thmtools}
\usepackage{stmaryrd}
\usepackage{etoolbox}

\usepackage{bm}
\usepackage{color}
\usepackage[dvipsnames]{xcolor}
\usepackage{subfig}
\usepackage{tabulary}
\usepackage{enumitem}

\usepackage{censor}
\usepackage{xstring}
\usepackage{calculus}
\usepackage{forloop}

\usepackage{tikz}
\usetikzlibrary{calc}
\usetikzlibrary{shapes}
\usetikzlibrary{patterns}
\usetikzlibrary{backgrounds}

\makeatletter
\patchcmd\thmt@mklistcmd
  {\thmt@thmname}
  {\check@optarg{\thmt@thmname}}
  {}{}
\patchcmd\thmt@mklistcmd
  {\thmt@thmname\ifx}
  {\check@optarg{\thmt@thmname}\ifx}
  {}{}
\protected\def\check@optarg#1{%
  \@ifnextchar\thmtformatoptarg\@secondoftwo{#1}%
}
\makeatother

\declaretheorem[name=Definition, shaded={rulecolor=black, rulewidth=0pt, bgcolor={rgb}{0.95,0.95,0.95}, margin=4pt}]{definition}
\declaretheorem[name=Definition, shaded={rulecolor=black, rulewidth=0pt, bgcolor={rgb}{0.95,0.95,0.95}, margin=4pt}]{hidendefinition}
\declaretheorem[name=Theorem, shaded={rulecolor=black, rulewidth=0pt, bgcolor={rgb}{0.95,0.95,0.95}, margin=4pt}]{theorem}
\declaretheorem[name=Algorithm, shaded={rulecolor=black, rulewidth=0pt, bgcolor={rgb}{0.95,0.95,0.95}, margin=4pt}]{algorithm}

\renewcommand*{\c}[1]{\protect\leavevmode\begingroup\color{blue}#1\endgroup}
\newcommand{\cn}[1]{#1}
\newcommand{\mycenter}[1]{\makebox[\textwidth][c]{#1}}

\newcommand{\ef}{e}

\newcommand{\MG}{M\hspace{-0.25ex}G}

\newcommand{\MS}{MS}

\makeatletter
\DeclareRobustCommand{\cev}[1]{\mathpalette\do@cev{#1}}
\newcommand{\do@cev}[2]{\fix@cev{#1}{+}\reflectbox{$\m@th#1\vec{\reflectbox{$\fix@cev{#1}{-}\m@th#1#2\fix@cev{#1}{+}$}}$}\fix@cev{#1}{-}}
\newcommand{\fix@cev}[2]{\ifx#1\displaystyle\mkern#23mu\else\ifx#1\textstyle\mkern#23mu\else\ifx#1\scriptstyle\mkern#22mu\else\mkern#22mu\fi\fi\fi}
\makeatother

\newcommand{\threshold}{\tau}
\renewcommand{\pm}{{\hspace{-0.05ex}\prime}}
\newcommand{\dpm}{{\hspace{-0.05ex}\prime\hspace{-0.2ex}\prime}}

\newcommand{\V}{\bm{V}}
\newcommand{\E}{\bm{E}}
\newcommand{\T}{\bm{T}}

\newcommand{\Nin}{N_{\text{in}}}
\newcommand{\Nout}{N_{\text{out}}}

\newcommand{\vv}{{(v,v^\pm)}}

\newcommand{\vvt}{{(v,v^\pm\!,t)}}

\newcommand{\cvvt}{{(v,v^\pm\!\c{,t})}}

\newcommand{\VV}{{(V,V^\pm)}}
\newcommand{\pVV}{{V,V^\pm}}

\newcommand{\VVT}{{(V,V^\pm\!,T)}}
\newcommand{\pVVT}{{V,V^\pm\!,T}}
\newcommand{\cVVT}{{(V,V^\pm\!\c{,T})}}
\newcommand{\pcVVT}{{V,V^\pm\!\c{,T}}}

\newcommand{\bVV}{{(\V,\V)}}

\newcommand{\bVVT}{{(\V,\V,\T)}}

\newcommand{\XX}{{(X,X^\pm)}}

\newcommand{\XXX}{{(X,X^\pm\!,X^\dpm)}}

\newcommand{\cXXX}{{(X,X^\pm\!\c{,X^\dpm})}}

\newcommand{\Xin}{\cev{X}}
\newcommand{\Xout}{\vec{X}}

\newcommand{\YY}{{(Y,Y^\pm)}}
\newcommand{\pYY}{{Y,Y^\pm}}

\newcommand{\YYY}{{(Y,Y^\pm\!,Y^\dpm)}}

\newcommand{\cYYY}{{(Y,Y^\pm\!\c{,Y^\dpm})}}

\newcommand{\uu}{{(u,u^\pm)}}

\newcommand{\Utm}{\bar{U}}

\newcommand{\calI}{\mathcal{I}}
\newcommand{\calH}{\mathcal{H}}
\newcommand{\calP}{\mathcal{P}}
\newcommand{\calT}{\mathcal{T}}
\newcommand{\calE}{\mathcal{E}}

\newcommand{\calV}{\mathcal{V}}
\newcommand{\calX}{\mathcal{X}}
\newcommand{\calY}{\mathcal{Y}}
\newcommand{\calZ}{\mathcal{Z}}

\newcommand{\bVxVxT}{\V{\times}\V{\times}\T}
\newcommand{\cbVxVxT}{\V{\times}\V\c{{\times}\T}}
\newcommand{\VxVxT}{V{\times}V^\pm{\times}T}
\newcommand{\cVxVxT}{V{\times}V^\pm\c{{\times}T}}
\newcommand{\bVxV}{\V{\times}\V}
\newcommand{\VxV}{V{\times}V^\pm}
\newcommand{\calVxV}{\calV{\times}\calV}
\newcommand{\calVV}{\calV\calV}
\newcommand{\calVVT}{\calV\calV\calT}
\newcommand{\ccalVVT}{\calVV\c{\calT}}
\newcommand{\calVxVxT}{\calV{\times}\calV{\times}\calT}

\newcommand{\hatcalP}{\hat{\calP}}
\newcommand{\tilcalP}{\widetilde{\calP}}

\newcommand{\frakP}{\mathfrak{P}}
\newcommand{\frakR}{\mathfrak{R}}
\newcommand{\frakC}{\mathfrak{C}}

\newcommand{\frakH}{\mathfrak{H}}
\newcommand{\tilfrakP}{\widetilde{\frakP}}
\newcommand{\hatfrakP}{\hat{\frakP}}
\newcommand{\hatfrakR}{\hat{\frakR}}
\newcommand{\hatfrakC}{\hat{\frakC}}

\newcommand{\frakI}{\mathfrak{I}}
\newcommand{\obj}{q}
\newcommand{\pop}{\bm{X}}
\newcommand{\partpop}{X}
\newcommand{\partitionpop}{\mathcal{X}}

\newcommand{\DKL}{{D_{\mathrm{KL}}}}

\newcommand\given[1][]{\:#1\vert\:}
\newcommand\ggiven[1][]{\:#1\vert\vert\:}

\newcommand{\defeq}{\stackrel{\mathclap{\normalfont\footnotesize\mbox{def}}}{=}}

\DeclareMathOperator*{\argmin}{arg\,min}
\DeclareMathOperator{\loss}{loss}
\DeclareMathOperator{\suma}{sum}

\DeclareMathOperator{\info}{info}

\newcommand{\lossf}[3]{\loss(#2)}

\newcounter{partitionSize}
\newcounter{partSize}
\newcounter{partSizeMax}
\newcounter{partCount}
\newcounter{lastPart}
\newcounter{currentPart}
\newcounter{var1}

\newlength{\lena}

\newcommand{\optStr}{}
\newcommand{\txt}{}

\newif\iftikzbold
\tikzboldtrue

\newif\iftikzempty
\tikzemptyfalse

\definecolor{opacCol}{rgb}{0.5,0.5,0.5}
\newcommand{\tikzopac}[1]{\textcolor{opacCol}{#1}}

\newcommand{\mychar}{}

\newcommand{\tikzpart}[3][-]{%
	\IfStrEq{#1}{-}
	{\renewcommand{\optStr}{#2}}
	{\renewcommand{\optStr}{#1}}
    
	\coordinate (sep) at (0.3,0);

	\StrRight{#2}{1}[\lastc]
	\IfStrEq{\lastc}{-}{\renewcommand{\txt}{#2}}{
		\IfStrEq{\lastc}{+}{\renewcommand{\txt}{#2}}{\renewcommand{\txt}{#2-}}
	}	
	
	\StrLen{\txt}[\len]
	\StrCount{\txt}{-}[\cutNb]
	\StrCount{\txt}{+}[\cutNbb]

	\setcounter{partitionSize}{\len-\cutNb-\cutNbb}
	\setcounter{partSizeMax}{\len+1}
	
	\setcounter{lastPart}{1}
	\setcounter{currentPart}{0}
	
	\setlength{\lena}{\arabic{partitionSize} cm * \real{0.575}}
	\node[minimum width=\lena,minimum height=0.9cm] (#3) at (pos) {};

	\iftikzbold \tikzset{tikzfill/.style={thick,draw=black,fill=white}}
	\else \tikzset{tikzfill/.style={thick,draw=white!55!black,fill=white}} \fi

	\StrLeft{#2}{1}[\firstc]
	\IfStrEq{\firstc}{-}{}{
		\forloop{partCount}{1}{\value{partCount} < \value{partSizeMax}}
		{
			\StrChar{\txt}{\value{partCount}}[\mychar]
			\IfStrEq{\mychar}{-}{
				\setcounter{partSize}{\value{currentPart}-\value{lastPart}+1}
				\begin{pgfonlayer}{background}
                  \setlength{\lena}{\arabic{partSize} cm * \real{0.575}}
                  \setcounter{var1}{\value{partitionSize} - \value{currentPart} - \value{lastPart} + 1}
                  \node[minimum width=\lena,minimum height=0.9cm,draw] (#3p\arabic{lastPart}) at ($(pos) - \value{var1}*(sep)$) {};
                  \setcounter{lastPart}{\value{currentPart}+1}
				\end{pgfonlayer}
              }{
				\IfStrEq{\mychar}{+}{
                  \setcounter{partSize}{\value{currentPart}-\value{lastPart}+1}
                  \begin{pgfonlayer}{background}
                    \setlength{\lena}{\arabic{partSize} cm * \real{0.575}}
                    \setcounter{var1}{\value{partitionSize} - \value{currentPart} - \value{lastPart} + 1}
                    \node[minimum width=\lena,minimum height=0.9cm,draw,densely dotted] (#3p\arabic{lastPart}) at ($(pos) - \value{var1}*(sep)$) {};
                    \setcounter{lastPart}{\value{currentPart}+1}
                  \end{pgfonlayer}
				}{
					\setcounter{currentPart}{\value{currentPart}+1}
					\setcounter{var1}{2*\value{currentPart} - \value{partitionSize} - 1}
					\iftikzempty\renewcommand{\mychar}{}\fi
					\iftikzbold \node[mypart] (#3i\arabic{currentPart}) at ($(pos) + \value{var1}*(sep)$) {$\mychar$\vphantom{\optStr}};
					\else \node[mypart] (#3i\arabic{currentPart}) at ($(pos) + \value{var1}*(sep)$) {$\tikzopac{\mychar}$\vphantom{\optStr}};\fi
				}
			}
		}
	}
}

\newcommand{\tikzpartbackup}[3][-]{%
	\IfStrEq{#1}{-}
	{\renewcommand{\optStr}{#2}}
	{\renewcommand{\optStr}{#1}}
    
	\coordinate (sep) at (0.3,0);

	\StrRight{#2}{1}[\lastc]
	\IfStrEq{\lastc}{-}{\renewcommand{\txt}{#2}}{
		\IfStrEq{\lastc}{+}{\renewcommand{\txt}{#2}}{\renewcommand{\txt}{#2-}}
	}	
	
	\StrLen{\txt}[\len]
	\StrCount{\txt}{-}[\cutNb]
	\StrCount{\txt}{+}[\cutNbb]

	\setcounter{partitionSize}{\len-\cutNb-\cutNbb}
	\setcounter{partSizeMax}{\len+1}
	
	\setcounter{lastPart}{1}
	\setcounter{currentPart}{0}
	
	\setlength{\lena}{\arabic{partitionSize} cm * \real{0.4}}
	\node[minimum width=\lena,minimum height=1cm] (#3) at (pos) {};

	\iftikzbold \tikzset{tikzfill/.style={thick,draw=black,fill=white}}
	\else \tikzset{tikzfill/.style={thick,draw=white!55!black,fill=white}} \fi

	\StrLeft{#2}{1}[\firstc]
	\IfStrEq{\firstc}{-}{}{
		\forloop{partCount}{1}{\value{partCount} < \value{partSizeMax}}
		{
			\StrChar{\txt}{\value{partCount}}[\mychar]
			\IfStrEq{\mychar}{-}{
				\setcounter{partSize}{\value{currentPart}-\value{lastPart}+1}
				\begin{pgfonlayer}{background}
					\ifnum\value{partSize}>1
						\fill[tikzfill] \convexpath{#3i\arabic{lastPart},#3i\arabic{currentPart}}{0.4cm};
						\setlength{\lena}{\arabic{partSize} cm * \real{0.4}}
						\setcounter{var1}{\value{partitionSize} - \value{currentPart} - \value{lastPart} + 1}
						\node[minimum width=\lena,minimum height=1cm] (#3p\arabic{lastPart}) at ($(pos) - \value{var1}*(sep)$) {};
					\else
						\fill[tikzfill] (#3i\arabic{lastPart}) circle (0.4cm);
					\fi
					\setcounter{lastPart}{\value{currentPart}+1}
				\end{pgfonlayer}
			}{
				\IfStrEq{\mychar}{+}{
					\setcounter{partSize}{\value{currentPart}-\value{lastPart}+1}
					\begin{pgfonlayer}{background}
						\ifnum\value{partSize}>1
							\fill[tikzfill,densely dotted] \convexpath{#3i\arabic{lastPart},#3i\arabic{currentPart}}{0.4cm};
							\setlength{\lena}{\arabic{partSize} cm * \real{0.4}}
							\setcounter{var1}{\value{partitionSize} - \value{currentPart} - \value{lastPart} + 1}
							\node[minimum width=\lena,minimum height=1cm] (#3p\arabic{lastPart}) at ($(pos) - \value{var1}*(sep)$) {};
						\else
							\fill[tikzfill,densely dotted] (#3i\arabic{lastPart}) circle (0.4cm);
						\fi
							\setcounter{lastPart}{\value{currentPart}+1}
					\end{pgfonlayer}
				}{
					\setcounter{currentPart}{\value{currentPart}+1}
					\setcounter{var1}{2*\value{currentPart} - \value{partitionSize} - 1}
					\iftikzempty\renewcommand{\mychar}{}\fi
					\iftikzbold \node[mypart] (#3i\arabic{currentPart}) at ($(pos) + \value{var1}*(sep)$) {$\mychar$\vphantom{\optStr}};
					\else \node[mypart] (#3i\arabic{currentPart}) at ($(pos) + \value{var1}*(sep)$) {$\tikzopac{\mychar}$\vphantom{\optStr}};\fi
				}
			}
		}
	}
}

\newcommand{\tikzringpart}[3][-]{%
	\IfStrEq{#1}{-}
	{\renewcommand{\optStr}{#2}}
	{\renewcommand{\optStr}{#1}}
		
	\coordinate (sep) at (0.3,0);

	\StrRight{#2}{1}[\lastc]
	\IfStrEq{\lastc}{-}{\renewcommand{\txt}{#2}}{
		\IfStrEq{\lastc}{+}{\renewcommand{\txt}{#2}}{\renewcommand{\txt}{#2-}}
	}	
	
	\StrLen{\txt}[\len]
	\StrCount{\txt}{-}[\cutNb]
	\StrCount{\txt}{+}[\cutNbb]

	\setcounter{partitionSize}{\len-\cutNb-\cutNbb}
	\setcounter{partSizeMax}{\len+1}
	
	\setcounter{lastPart}{1}
	\setcounter{currentPart}{0}
	
	\setlength{\lena}{\arabic{partitionSize} cm * \real{0.5}}
	\node[minimum width=\lena,minimum height=0.3cm] (#3) at (pos) {};

	\iftikzbold \tikzset{tikzfill/.style={thick,draw=black,fill=white}}
	\else \tikzset{tikzfill/.style={thick,draw=white!55!black,fill=white}} \fi

	\StrLeft{#2}{1}[\firstc]
	\IfStrEq{\firstc}{-}{}{
		\forloop{partCount}{1}{\value{partCount} < \value{partSizeMax}}
		{
			\StrChar{\txt}{\value{partCount}}[\mychar]
			\IfStrEq{\mychar}{+}{}{\IfStrEq{\mychar}{-}{}{\setcounter{currentPart}{\value{currentPart}+1}}}
			
			\coordinate (r) at (0,0);
			\ifnum\value{currentPart}=1\coordinate (r) at (0,1);\fi
			\ifnum\value{currentPart}=2\coordinate (r) at (0.951,0.309);\fi
			\ifnum\value{currentPart}=3\coordinate (r) at (0.588,-0.809);\fi
			\ifnum\value{currentPart}=4\coordinate (r) at (-0.588,-0.809);\fi
			\ifnum\value{currentPart}=5\coordinate (r) at (-0.951,0.309);\fi
			
			\IfStrEq{\mychar}{-}{
				\setcounter{partSize}{\value{currentPart}-\value{lastPart}+1}
				\begin{pgfonlayer}{background}
					\ifnum\value{partSize}>1
						\setlength{\lena}{\arabic{partSize} cm * \real{0.5}}
						\node[minimum width=\lena,minimum height=0.3cm] (#3p\arabic{lastPart}) at ($(pos) + 1.5*(r)$) {};
					\fi
					\setcounter{lastPart}{\value{currentPart}+1}
				\end{pgfonlayer}
			}{
				\IfStrEq{\mychar}{+}{
					\setcounter{partSize}{\value{currentPart}-\value{lastPart}+1}
					\begin{pgfonlayer}{background}
						\ifnum\value{partSize}>1
							\setlength{\lena}{\arabic{partSize} cm * \real{0.5}}
							\node[minimum width=\lena,minimum height=0.3cm] (#3p\arabic{lastPart}) at ($(pos) + 1.5*(r)$) {};
						\fi
							\setcounter{lastPart}{\value{currentPart}+1}
					\end{pgfonlayer}
				}{
					\iftikzempty\renewcommand{\mychar}{}\fi
					\iftikzbold \node[mypart] (#3i\arabic{currentPart}) at ($(pos) + 1.5*(r)$) {$\mychar$\vphantom{\optStr}};
					\else \node[mypart] (#3i\arabic{currentPart}) at ($(pos) + 1.5*(r)$) {$\tikzopac{\mychar}$\vphantom{\optStr}};\fi
				}
			}
		}
	}
}

\newcommand{\partie}[2][-]{%
	\StrRight{#2}{1}[\lastc]%
	\IfStrEq{\lastc}{-}{\renewcommand{\txt}{#2}}{%
		\IfStrEq{\lastc}{+}{\renewcommand{\txt}{#2}}{\renewcommand{\txt}{#2-}}%
	}%
	\StrLen{\txt}[\len]
	\StrCount{\txt}{-}[\cutNb]
	\StrCount{\txt}{+}[\cutNbb]
	\setcounter{currentPart}{\len-\cutNb-\cutNbb}
	\setlength{\unitlength}{0.43cm}%
	\ifnum\value{currentPart}=2\hspace{-0.5cm}\fi%
	\ifnum\value{currentPart}=4\hspace{-0.6cm}\fi%
	\begin{picture}(\value{currentPart},0)
	\put(0,0){\begin{tikzpicture}[baseline=(nodei1.base),scale=0.3]
		\coordinate (pos) at (0,0);
		\tikzpart{#2}{node}		
	\end{tikzpicture}%
	}\end{picture}%
}

\tikzset{
  transform shape,
  cell/.style = {rectangle, draw=black, minimum size=1cm, line width=0.4mm, font=\normalsize, inner sep=0},
  grid/.style = {rectangle, draw=black, fill=white, minimum size=1cm, line width=0.1mm, dotted, font=\normalsize},
  agg/.style = {draw=black!25!red, line width=0.8mm, font=\normalsize},
  sagg/.style = {draw=black, line width=0.4mm, font=\normalsize},
  cross/.style = {cross out, draw=black, minimum size=0.6cm, inner sep=0pt, outer sep=0pt},
  labels/.style = {font=\bf\large},
  bigtext/.style = {align=center, anchor=south, font=\bf\Large},
  myobject/.style = {draw=black, fill=white, font=\large},
  mynode/.style = {circle, myobject, minimum size=10mm},
  myedge/.style = {OliveGreen},
  myagg/.style = {black!25!red, line width=0.8mm},
  mypart/.style = {minimum width=1cm, minimum height=0.9cm, font=\bf\LARGE},
  value/.style = {},
  pattern0/.style = {},
  pattern1/.style = {pattern=north west lines, pattern color=blue},
  pattern2/.style = {pattern=horizontal lines, pattern color=red},
  pattern3/.style = {pattern=north east lines, pattern color=green!50!black},
  pattern4/.style = {pattern=vertical lines, pattern color=black!75!white},
  pattern5/.style = {pattern=dots, pattern color=black},
  toInLink/.style = {},
  toOutLink/.style = {->, >=latex},
  toPart/.style = {dashed, ->, >=latex},
  toPartition/.style = {thick, ->, >=latex},
  toIncluded/.style = {->, thick},
}

\newcommand{\tikzPrintSimpleCoord}{
  \coordinate (c1) at (0,2);
  \coordinate (c2) at (-1.75,0);
  \coordinate (c3) at (0,-2);
  \coordinate (c4) at (4,1.75);
  \coordinate (c5) at (4,-1.75);  
}

\newcommand{\tikzPrintLabels}{
  \tikzPrintHLabels
  \tikzPrintVLabels
}

\newcommand{\tikzPrintVLabels}{
  \node [labels] at (-1,4) {$v_1$};
  \node [labels] at (-1,3) {$v_2$};
  \node [labels] at (-1,2) {$v_3$};
  \node [labels] at (-1,1) {$v_4$};
  \node [labels] at (-1,0) {$v_5$};
}

\newcommand{\tikzPrintHLabels}{
  \node [labels] at (0,5) {$v_1$};
  \node [labels] at (1,5) {$v_2$};
  \node [labels] at (2,5) {$v_3$};
  \node [labels] at (3,5) {$v_4$};
  \node [labels] at (4,5) {$v_5$};
}

\newcommand{\tikzPrintEdge}[3]{
  \MULTIPLY{#3}{0.5}{\temp}
  \draw [myedge, ->, line width=\temp] (#1) to [bend left=10] (#2);  
}

\newcommand{\tikzPrintLoop}[3]{
  \MULTIPLY{#2}{0.333}{\temp}
  \ifthenelse{\equal{#3}{1}} { \draw [myedge, <->, line width=\temp] (#1) to [out=150,in=210,looseness=5] (#1); }
  {\ifthenelse{\equal{#3}{2}} { \draw [myedge, <->, line width=\temp] (#1) to [out=-60,in=-120,looseness=5] (#1); }
    {\ifthenelse{\equal{#3}{3}} { \draw [myedge, <->, line width=\temp] (#1) to [out=-30,in=30,looseness=5] (#1); }
      {\ifthenelse{\equal{#3}{4}} { \draw [myedge, <->, line width=\temp] (#1) to [out=60,in=120,looseness=5] (#1); }}}}
}

\newcommand{\tikzPrintSEdge}[3]{
  \MULTIPLY{#3}{1}{\temp}
  \draw [myedge, ->, line width=\temp] (#1) to [bend left=10] (#2);  
}

\newcommand{\tikzPrintSLoop}[3]{
  \MULTIPLY{#2}{1}{\temp}
  \ifthenelse{\equal{#3}{1}} { \draw [myedge, <->, line width=\temp] (#1) to [out=150,in=210,looseness=5] (#1); }
  {\ifthenelse{\equal{#3}{2}} { \draw [myedge, <->, line width=\temp] (#1) to [out=-60,in=-120,looseness=5] (#1); }
    {\ifthenelse{\equal{#3}{3}} { \draw [myedge, <->, line width=\temp] (#1) to [out=-30,in=30,looseness=5] (#1); }
      {\ifthenelse{\equal{#3}{4}} { \draw [myedge, <->, line width=\temp] (#1) to [out=60,in=120,looseness=5] (#1); }}}}
}

\newcommand{\tikzPrintCell}[4]{
  \ifthenelse{\equal{#4}{NA}}%
  {%
    \node[cell, fill=white] at (#1,#2) {#3};
    \node[cross] at (#1,#2) {};
  }%
  {%
    \MULTIPLY{#4}{1000}{\temp}
    \ROUND[0]{\temp}{\temp}
    \node[cell, fill=OliveGreen!\temp!white] at (#1,#2) {#3};
  }%
}

\newcommand{\tikzPrintWhiteCell}[2]{
    \node[cell, fill=white, text=black] at (#1,#2) {};
}

\newcommand{\tikzPrintBlackCell}[2]{
    \node[cell, fill=OliveGreen, text=white] at (#1,#2) {};
}

\newcommand{\tikzPrintRectangle}[6]{
  \ifthenelse{\equal{#6}{NA}}%
  {%
    \begin{scope}[shift={(-0.5,-0.5)}]
      \draw [cell, fill=white] (#1,#2) rectangle (#3,#4);
    \end{scope}        
  }%
  {%
    \MULTIPLY{#6}{1000}{\temp}
    \ROUND[0]{\temp}{\temp}
    \begin{scope}[shift={(-0.5,-0.5)}]
      \draw [cell, fill=OliveGreen!\temp!white] (#1,#2) rectangle (#3,#4);
      \node [value] at ($0.5*(#1,#2)+0.5*(#3,#4)$) {#5};
    \end{scope}        
  }%
}

\newcommand{\tikzPrintWhiteRectangle}[4]{
  \begin{scope}[shift={(-0.5,-0.5)}]
    \draw [cell, fill=white] (#1,#2) rectangle (#3,#4);
    \node [value, text=black] at ($0.5*(#1,#2)+0.5*(#3,#4)$) {};
  \end{scope}        
}

\newcommand{\tikzPrintBlackRectangle}[4]{
  \begin{scope}[shift={(-0.5,-0.5)}]
    \draw [cell, fill=OliveGreen, text=white] (#1,#2) rectangle (#3,#4);
    \node [value, text=white] at ($0.5*(#1,#2)+0.5*(#3,#4)$) {};
  \end{scope}        
}

\newcommand{\tikzPrintAgg}[5][pattern0]{
  \begin{scope}[shift={(-0.5,-0.5)}]
    \draw [agg, #1] (#2,#3) rectangle (#4,#5);
  \end{scope}
}

\newcommand{\tikzPrintSAgg}[4]{
  \begin{scope}[shift={(-0.5,-0.5)}]
    \draw [sagg] (#1,#2) rectangle (#3,#4);
  \end{scope}
}

\newcommand{\tikzPrintTAgg}[4]{
  \begin{scope}[shift={(-0.5,-0.5)}]
    \draw [agg, fill=white] (#1,#2) rectangle (#3,#4);
  \end{scope}
}

\newcommand{\tikzPrintGrid}{
  \node[grid] at (0,4) {};
  \node[grid] at (1,4) {};
  \node[grid] at (2,4) {};
  \node[grid] at (3,4) {};
  \node[grid] at (4,4) {};

  \node[grid] at (0,3) {};
  \node[grid] at (1,3) {};
  \node[grid] at (2,3) {};
  \node[grid] at (3,3) {};
  \node[grid] at (4,3) {};

  \node[grid] at (0,2) {};
  \node[grid] at (1,2) {};
  \node[grid] at (2,2) {};
  \node[grid] at (3,2) {};
  \node[grid] at (4,2) {};

  \node[grid] at (0,1) {};
  \node[grid] at (1,1) {};
  \node[grid] at (2,1) {};
  \node[grid] at (3,1) {};
  \node[grid] at (4,1) {};

  \node[grid] at (0,0) {};
  \node[grid] at (1,0) {};
  \node[grid] at (2,0) {};
  \node[grid] at (3,0) {};
  \node[grid] at (4,0) {};
}

\newcommand{\convexpath}[2]{
  [   
  create hullnodes/.code={
    \global\edef\namelist{#1}
    \foreach [count=\counter] \nodename in \namelist {
      \global\edef\numberofnodes{\counter}
      \node at (\nodename) [draw=none,name=hullnode\counter] {};
    }
    \node at (hullnode\numberofnodes) [name=hullnode0,draw=none] {};
    \pgfmathtruncatemacro\lastnumber{\numberofnodes+1}
    \node at (hullnode1) [name=hullnode\lastnumber,draw=none] {};
  },
  create hullnodes
  ]
  ($(hullnode1)!#2!-90:(hullnode0)$)
  \foreach [
  evaluate=\currentnode as \previousnode using \currentnode-1,
  evaluate=\currentnode as \nextnode using \currentnode+1
  ] \currentnode in {1,...,\numberofnodes} {
    -- ($(hullnode\currentnode)!#2!-90:(hullnode\previousnode)$)
    let \p1 = ($(hullnode\currentnode)!#2!-90:(hullnode\previousnode) - (hullnode\currentnode)$),
    \n1 = {atan2(\y1,\x1)},
    \p2 = ($(hullnode\currentnode)!#2!90:(hullnode\nextnode) - (hullnode\currentnode)$),
    \n2 = {atan2(\y2,\x2)},
    \n{delta} = {-Mod(\n1-\n2,360)}
    in 
    {arc [start angle=\n1, delta angle=\n{delta}, radius=#2]}
  }
  -- cycle
}

\newcommand{\drawblackcell}[7]{
  \begin{scope}[shift={(#1,#2,#3)}]
    \draw [cell, fill=OliveGreen!#7!white] (#4,#5,#6) -- (0,#5,#6) -- (0,0,#6) -- (#4,0,#6) -- cycle;
    \draw [cell, fill=OliveGreen!#7!white] (#4,#5,#6) -- (0,#5,#6) -- (0,#5,0) -- (#4,#5,0) -- cycle;
    \draw [cell, fill=OliveGreen!#7!white] (#4,#5,#6) -- (#4,#5,0) -- (#4,0,0) -- (#4,0,#6) -- cycle;
  \end{scope}
}

\begin{document}

\begin{frontmatter}

\title{An Information-theoretic Framework\\ for the Lossy Compression of Link Streams}

\author{Robin Lamarche-Perrin}
\ead{Robin.Lamarche-Perrin@lip6.fr}
\ead[url]{https://www-complexnetworks.lip6.fr/~lamarche/}

\address{Centre national de la recherche scientifique\\ Institut des syst\`emes complexes de Paris \^Ile-de-France\\ Laboratoire d'informatique de Paris 6} 

\begin{abstract}
  Graph compression is a data analysis technique that consists in the replacement of parts of a graph by more general structural patterns in order to reduce its description length.
  It notably provides interesting exploration tools for the study of real, large-scale, and complex graphs which cannot be grasped at first glance.
  This article proposes a framework for the compression of temporal graphs, that is for the compression of graphs that evolve with time.
  This framework first builds on a simple and limited scheme, exploiting structural equivalence for the lossless compression of static graphs, then generalises it to the lossy compression of link streams, a recent formalism for the study of temporal graphs.
  Such generalisation relies on the natural extension of (bidimensional) relational data by the addition of a third temporal dimension.
  Moreover, we introduce an information-theoretic measure to quantify and to control the information that is lost during compression, as well as an algebraic characterisation of the space of possible compression patterns to enhance the expressiveness of the initial compression scheme.
  These contributions lead to the definition of a combinatorial optimisation problem, that is the Lossy Multistream Compression Problem, for which we provide an exact algorithm.
\end{abstract}


\end{frontmatter}

\renewcommand{\listtheoremname}{Table of Definitions}
\listoftheorems[ignoreall,show={definition}]

\clearpage

\tableofcontents

\clearpage

\section*{Table of Notations}
\addcontentsline{toc}{section}{Table of Notations}

\makeatletter
\def\z#14#2!!{\def\TY@classz{#17#2}}
\expandafter\z\TY@classz!!
\makeatother

{\normalsize
  \begin{tabulary}{\textwidth}{lL}
    \makebox[1.7cm][l]{$v \in \V$} $/ \quad t \in \T$ & a vertex / a time instance\\[1ex]
    \makebox[1.7cm][l]{$V \in \calP(\V)$} $/ \quad T \in \calP(\T)$ & a vertex subset / a time subset\\[1ex]
    \makebox[1.7cm][l]{$\calV \in \frakP(\V)$} $/ \quad \calT \in \frakP(\T)$ & a vertex partition / a time partition\\[1ex]
    \makebox[1.7cm][l]{$\calV(v) \in \calV$} $/ \quad \calT(t) \in \calT$ & the vertex subset in $\calV$ that contains $v$ / the time instance in $\calT$ that contains $t$\\[2em]
    


    $\vvt \in \bVxVxT$ & a multiedge\\[1ex]
    $\VxVxT \in \calP(\bVxVxT)$ & a Cartesian multiedge subset\\[1ex]
    $\calVxVxT \in \frakP(\bVxVxT)$ & a grid multiedge partition\\[1ex]
    $\calVVT \in \frakP^\times(\bVxVxT)$ & a Cartesian multiedge partition\\[1ex]
    $\calVVT\vvt \in \calVVT$ & the multiedge subset in partition $\calVVT$ that contains $\vvt$\\[2em]

    $e : \bVxVxT \rightarrow \mathbb{N}$ & the edge function of a multistream\\[1ex]
    $e : \calP(\V){\times}\calP(\V){\times}\calP(\T) \rightarrow \mathbb{N}$ & the additive extension of the edge function\\[1ex]
    $e\bVVT$ & the total number of edges\\[2em]
  \end{tabulary}

  \begin{tabulary}{\textwidth}{lL}
    $\XXX \in \bVxVxT$ & the observed variable associated with the empirical distribution of edges in a multistream\\[1ex]
    $\calVxVxT\XXX \in \calVxVxT$ & the compressed variable resulting from the compression of the observed variable $\XXX$ by a given multiedge partition $\calVxVxT$\\[1ex]
    $\YYY \in \bVxVxT$ & the external variable which distribution is used to decompress the compressed variable $\calVxVxT\XXX$\\[1ex]
    $\calVxVxT_\YYY\XXX \in \bVxVxT$ & the decompressed variable obtained from the decompression of the compressed variable $\calVxVxT\XXX$ according to the external variable $\YYY$\\[1ex]
    $\lossf{\XXX}{\calVxVxT}{\YYY}$ & the information loss induced from the compression of the observed variable $\XXX$ by a given multiedge partition $\calVxVxT$ and its decompression according to the external variable $\YYY$\\[2em]
  \end{tabulary}

  \begin{tabulary}{\textwidth}{lL}
    \makebox[1.7cm][l]{$V \in \hatcalP(\V)$} $/ \quad T \in \hatcalP(\T)$ & a feasible vertex subset / a feasible time subset\\[1ex]
    \makebox[1.7cm][l]{$\calH(\V)$} $/ \quad \calI(\T)$ & a vertex hierarchy / a set of time intervals\\[1ex]
    $\VxVxT \in \hatcalP(\bVxVxT)$ & a feasible Cartesian multiedge subset\\[1ex]
    $\calVVT \in \hatfrakP^\times(\bVxVxT)$ & a feasible Cartesian multiedge partition\\[2em]

    $\hatfrakR(\calVVT) \subset \hatfrakP(\bVxVxT)$ & the set of feasible multiedge partitions that refine $\calVVT$\\[1ex]
    $\hatfrakC(\calVVT) \subset \hatfrakP(\bVxVxT)$ & the set of feasible multiedge partitions that are covered by $\calVVT$
  \end{tabulary}

}


\clearpage

\section{Introduction}

\emph{Graph abstraction} is a data analysis technique aiming at the extraction of salient features from relational data to provide a simpler, and hence more useful representation of the graph under study.
Such a process generally relies on a controlled information reduction suppressing redundancies or irrelevant parts of the data~\cite{Zhou09}.
Abstraction techniques are hence crucial to the study of real, large-scale, and complex graphs which cannot be grasped at first glance.
First, they provide tools for an optimised storage and data treatment by reducing memory requirements and running times of analysis algorithms.
Second, and more importantly, they constitute valuable exploration tools for domain experts who are looking for preliminary macroscopic insights about their graphs' topology or, even better, a multiscale representation of their data.

Among abstraction techniques, \emph{graph compression}~\cite{Toivonen11,Serafino13}, also known as \emph{graph simplification} or \emph{graph summarisation}~\cite{Navlakha08,Zhang10,LeFevre10}, consists in replacing parts of the graph by more general structural patterns in order to reduce its description length.
For example, one ``can replace a dense cluster by a single node, so the overall structure of the network becomes clearer''~\cite{Zhou09},
or more generally replace any frequent subgraph pattern (\emph{e.g.}, cliques, stars, loops) by a label of that pattern.
Such techniques hence range from those building on collections of domain-specific patterns, such as \emph{graph rewriting} techniques in which patterns of interest are specified according to expert knowledge~\cite{Pinaud12}, to those relying on more generic patterns, such as \emph{power graph} techniques in which any group of vertices with identical interaction profiles is a candidate for summarisation~\cite{Dwyer13,Ahnert14}.
Because they provide more general approaches to graph analysis, we will focus on the latter.

In this article, we are more particularly interested in the compression of \emph{temporal graphs}, that is the compression of graphs that evolve with time.
Many research studies are indeed interested in the dynamics of relations, as for example the evolution of friendship relations in social sciences, or even in the dynamics of interaction events~\cite{Holme13}, as for example contact or communication networks, such as mail exchanges, financial transactions, physical meetings, and so on.
Having to deal with an additional dimension -- that is the temporal dimension -- challenges compression techniques that have initially been developed for the study of static graphs.
A traditional approach to generalise such techniques preliminary consists in the construction of a sequence of static graphs, by slicing the temporal dimension into distinct periods of interest, then in independently applying classical compression schemes to each graph of this sequence.
However, such a process introduces an asymmetry in the way structural and temporal information is handled, the latter being compressed prior to -- and independently from -- the former.

To this extent, recent work on the \emph{link stream} formalism proposes to deal with time as a simple addition to the graph's structural dimensions~\cite{Viard16, Latapy17}.
Considering temporal graphs and interaction networks as genuine tridimensional data, the arbitrary separation of structure and time is therefore prohibited.
Following this line of thinking, the compression scheme we present in this article aims at the natural generalisation of the bidimensional compression of static graphs to the tridimensional compression of link streams, thus participating in the development of this emerging framework.
Similar generalisation objectives have been addressed in previous work on graph compression, as for example the application of bidimensional \emph{block models} to multidimensional matrices~\cite{Borgatti92} or the application of \emph{biclustering} to triplets of variables~\cite{Narmadha16}, which has then been exploited for the statistical analysis of temporal graphs~\cite{Guigoures12}.
The particular interest of such approaches also consists in the fact that they provide a unified compression scheme in which structural and temporal information is simultaneously taken into account.\\

In order to present our compression framework, this article starts canonical and specific, then increase in generality and in sophistication.
Section~\ref{sec:GCP} introduces the \emph{Graph Compression Problem} (GCP), a first compression scheme that relies on a most classical combinatorial problem in graph theory:
Finding classes of structurally-equivalent vertices~\cite{Lorrain71} to summarise the adjacency-list and the adjacency-matrix representations of a given graph.
This approach to graph compression is canonical in the sense that it only builds on the primary, first-order information that is contained in relational data, that is the information encoded in vertex adjacency.
It is also specific in the sense that it only applies to simple graphs (that is graphs for which at most one edge is allowed between two vertices) with no temporal dimension (that is static graphs).
Moreover, this first scheme is lossless (it does not allow for any information loss during compression) and its solution space is both strongly constrained (only vertex partitions are considered, whereas edge partitions would allow for much more compression choices) and weakly expressive (any vertex subset is feasible, whereas interesting structural properties preliminarily defined by the expert domain might need to be preserved during compression).

In order to address such limitations, Section~\ref{sec:generalisation} consists in a step-by-step generalisation of the GCP to make it suitable for the lossy compression of temporal graphs.
First, we show how to deal with the compression of \emph{multigraphs} (that is graphs for which multiple edges are allowed between two vertices) by generalising the notion of structural equivalence to the case of multiple edges (\ref{ssec:unweighted}).
Second, we allow for a \emph{lossy} compression scheme by formalising a proper measure of information loss building on the entropy of the adjacency information contained in the compressed graph relative to the one contained in the initial graph (\ref{ssec:lossy}).
Third, we allow for a less constrained compression scheme by generalising from vertex partitions to edge partitions (\ref{ssec:powergraph}).
Fourth, we allow for a more expressive scheme by driving compression according to a predefined set of feasible aggregates (\ref{ssec:constrained}).
Fifth and last, we generalise the resulting framework to the compression of temporal multigraphs, that is what we later call \emph{multistreams}, by adding a temporal dimension to the compression scheme (\ref{ssec:dynamic}).
These five contributions finally define a general and flexible scheme for link stream compression, that we call the \emph{Multistream Compression Problem} (MSCP).

Section~\ref{sec:MSCP} then presents a combinatorial optimisation algorithm to solve the MSCP.
It relies on the reduction of the problem to the better-known \emph{Set Partitioning Problem} (SPP) arising as soon as one wants to organise a set of objects into covering and pairwise disjoint subsets such that an additive objective is minimised~\cite{Balas76}.
Building on a generic algorithmic framework proposed in previous work to solve special versions of the SPP~\cite{RLP_ICTAI14, RLP_CLUSTER14}, this article derives an algorithm to the particular case of the MSCP.
This algorithm relies on the acknowledgement of a principle of optimality, showing that the problem's solution space has an optimal substructure allowing for the recursive combination of locally-optimal solutions.
Applying classical methods of dynamic programming and providing a proper data structure for the MSCP, we finally derive an exact algorithm which is exponential in the worst case, but polynomial when the set of feasible vertex aggregates is assumed to have some particular structure (\emph{e.g.}, hierarchies of vertices and sets of intervals).

Section~\ref{sec:conclusion} discusses the outcomes of this new compression scheme and provides some research perspectives, notably to propose in the future tractable approximation algorithms for the lossy compression of large-scale temporal graphs.


\section{Starting Point: The Lossless Graph Compression Problem}
\label{sec:GCP}

The starting point to build our compression scheme is a well-known combinatorial problem:
Find the quotient set of the structural equivalence relation applying to the vertices of a graph.
As the resulting equivalence classes form a partition of the vertex set by grouping together vertices with an identical (first-order) structure -- that is with identical neighbourhoods -- one can exploit such classes to compress the graph representation, as illustrated in Figure~\ref{fig:simple_example}.
Structural equivalence can thus be used for the lossless compression of static graphs, and we later list the improvements one needs in order to generalise this first simple scheme to the lossy compression of link streams.

\subsection{Preliminary Notations}

Given a set of vertices $\V = \{v_1,\ldots,v_n\}$, we mark:
\begin{itemize}
\item $\calP(\V)$ the set of all vertex subsets: $\calP(\V) = \{V \subseteq \V\}$;
\item $\frakP(\V)$ the set of all vertex partitions:
  $$\frakP(\V) \; = \; \{\{V_1,\ldots,V_m\} \subseteq \calP(\V) \; : \: \cup_i V_i = \V \; \wedge \; \forall i \neq j, \; V_i \cap V_j = \emptyset\} \text{;}$$
  \item Given a vertex $v \in \V$ and a vertex partition $\calV \in \frakP(\V)$, we mark $\calV(v)$ the unique vertex subset in $\calV$ that contains $v$.
\end{itemize}

More generally, this article uses a consistent system of capitalization and typefaces to properly formalise the compression problem and its solution space:
\begin{itemize}
	\item {Vertices} are designated by {lowercase} letters: $v$, $v^\pm$, $u$, $u^\pm$;
	\item Vertex sets and vertex subsets by {uppercase} letters: $\V$, $V$, $V^\pm$;
	\item {Vertex partitions} and sets of vertex subsets by {calligraphic} letters: $\calV$, $\calV^\pm$, $\calP(\V)$, $\calH(\V)$, $\calI(\V)$;
	\item Sets of vertex partitions by {Gothic} letters: $\frakP(\V)$, $\frakH(\V)$, $\frakI(\V)$. 
\end{itemize}

\begin{figure}
  \centering
  \mycenter{\begin{tikzpicture}[scale=0.65]

  \node [align=center] (graph) at (0,-0.25) {
    \begin{tikzpicture}
      \tikzPrintSimpleCoord
      
      \node [mynode] (v1) at (c1) {$v_1$};
      \node [mynode] (v2) at (c2) {$v_2$};
      \node [mynode] (v3) at (c3) {$v_3$};
      \node [mynode] (v4) at (c4) {$v_4$};
      \node [mynode] (v5) at (c5) {$v_5$};

      \tikzPrintSEdge{v1}{v4}{2}
      \tikzPrintSEdge{v2}{v4}{2}
      \tikzPrintSEdge{v3}{v4}{2}
      \tikzPrintSEdge{v1}{v5}{2}
      \tikzPrintSEdge{v2}{v5}{2}
      \tikzPrintSEdge{v3}{v5}{2}

      \tikzPrintSEdge{v4}{v1}{2}
      \tikzPrintSEdge{v4}{v2}{2}
      \tikzPrintSEdge{v4}{v3}{2}
      \tikzPrintSEdge{v5}{v1}{2}
      \tikzPrintSEdge{v5}{v2}{2}
      \tikzPrintSEdge{v5}{v3}{2}

      \tikzPrintSEdge{v4}{v5}{2}
      \tikzPrintSEdge{v5}{v4}{2}
      \tikzPrintSLoop{v4}{2}{3}
      \tikzPrintSLoop{v5}{2}{3}
    \end{tikzpicture}
  };

  \node [align=center] (matrix) at (7,0) {
    \begin{tikzpicture}[scale=0.8]

      \tikzPrintLabels

      \tikzPrintWhiteCell{0}{4}
      \tikzPrintWhiteCell{1}{4}
      \tikzPrintWhiteCell{2}{4}
      \tikzPrintBlackCell{3}{4}
      \tikzPrintBlackCell{4}{4}

      \tikzPrintWhiteCell{0}{3}
      \tikzPrintWhiteCell{1}{3}
      \tikzPrintWhiteCell{2}{3}
      \tikzPrintBlackCell{3}{3}
      \tikzPrintBlackCell{4}{3}

      \tikzPrintWhiteCell{0}{2}
      \tikzPrintWhiteCell{1}{2}
      \tikzPrintWhiteCell{2}{2}
      \tikzPrintBlackCell{3}{2}
      \tikzPrintBlackCell{4}{2}

      \tikzPrintBlackCell{0}{1}
      \tikzPrintBlackCell{1}{1}
      \tikzPrintBlackCell{2}{1}
      \tikzPrintBlackCell{3}{1}
      \tikzPrintBlackCell{4}{1}

      \tikzPrintBlackCell{0}{0}
      \tikzPrintBlackCell{1}{0}
      \tikzPrintBlackCell{2}{0}
      \tikzPrintBlackCell{3}{0}
      \tikzPrintBlackCell{4}{0}
    \end{tikzpicture}
  };

  \node [align=center] (aggGraph) at (0,-6.25) {
    \begin{tikzpicture}
      \tikzPrintSimpleCoord

      \node (v1) at (c1) {};
      \node (v2) at (c2) {};
      \node (v3) at (c3) {};
      \node (v4) at (c4) {};
      \node (v5) at (c5) {};
      


      \node [mynode] (v123) at ($0.5*(v1)+0.5*(v3)-0.25*(v2)$) {$\{v_1, v_2, v_3\}$};
      \node [mynode] (v45) at ($0.5*(v4)+0.5*(v5)$) {$\{v_4, v_5\}$};

      \tikzPrintSEdge{v123}{v45}{2}
      \tikzPrintSEdge{v45}{v123}{2}
      \tikzPrintSLoop{v45}{2}{3}

    \end{tikzpicture}
  };

  \node [align=center] (aggMatrix) at (7,-6) {
    \begin{tikzpicture}[scale=0.8]
      
      \node [labels] at (1,5) {$\{v_1, v_2, v_3\}$};
      \node [labels] at (3.5,5) {$\{v_4, v_5\}$};

      \node [labels] at (-1,0.5) {\rotatebox{90}{$\{v_5, v_4\}$}};
      \node [labels] at (-1,3) {\rotatebox{90}{$\{v_3, v_2, v_1\}$}};

      \tikzPrintWhiteRectangle{0}{5}{3}{2}
      \tikzPrintBlackRectangle{3}{5}{5}{2}
      \tikzPrintBlackRectangle{0}{2}{3}{0}
      \tikzPrintBlackRectangle{3}{2}{5}{0}
    \end{tikzpicture}
  };

\end{tikzpicture}

}
  \caption{
    Lossless compression of a 5-vertex, 16-edge graph (above) into a 2-vertex, 3-edge graph (below).
    The \emph{adjacency-list} representation is given on the left and the \emph{adjacency-matrix} representation on the right.
  }
  \label{fig:simple_example}
\end{figure}
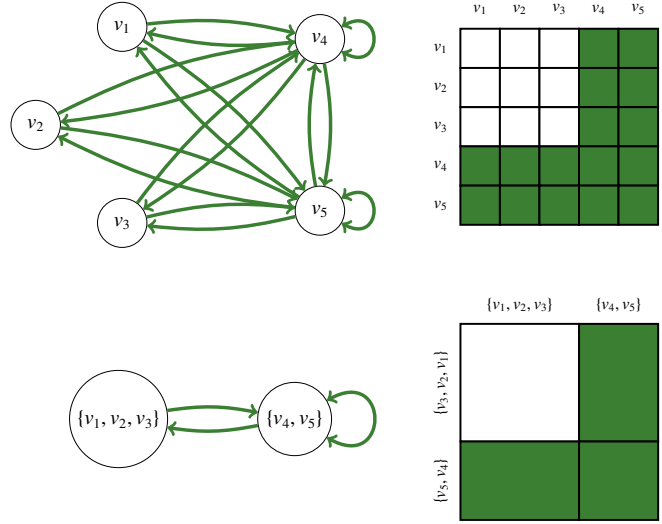

\subsection{The Lossless GCP}

To begin with, we consider a simple case: Directed static graphs, with possible self-loops on the vertices. 
\begin{definition}[Directed Graph]~\\
  \label{def:graph}
  A \emph{directed graph} $G = (\V,\E)$ is characterised by:
  \begin{itemize}
  \item A set of \emph{vertices} $\V$;
  \item A set of \emph{directed edges} $\E \subseteq \bVxV$.
  \end{itemize}
  For all vertex $v \in \V$, we respectively mark $\Nin(v) = \{v^\pm \in \V : (v^\pm,v) \in \E \}$ and $\Nout(v) = \{v^\pm \in \V : \vv \in \E \}$ the \emph{in-coming} and the \emph{out-going neighbourhoods} of $v$.
\end{definition}
The upper part of Figure~\ref{fig:simple_example} gives an example of directed graph made of $|\V| = 5$ vertices and $|\E| = 16$ edges.
It is represented in the form of \emph{adjacency lists} (on the left), where each edge is represented as an arrow going from a source vertex $v \in \V$ to a target vertex $v^\pm \in \V$, as well as in the form of an \emph{adjacency matrix} (on the right), where edges are represented within a binary matrix of size $|\V| \times |\V|$.\\

The combinatorial problem we now formalise builds on the classical relation of \emph{structural equivalence} applying to the vertex set of a graph~\cite{Lorrain71}.
\begin{definition}[Structural Equivalence]~\\
  \label{def:SE_graph}
  The \emph{structural equivalence relation} $\sim \; \subseteq \V^2$ is defined on directed graphs by the equality of neighbourhoods:
  Two vertices $\vv \in \V^2$ are \emph{structurally equivalent} if and only if they are connected to the same vertices.
  Formally:
  $$v \sim v^\pm \quad \Leftrightarrow \quad \Nin(v) = \Nin(v^\pm) \quad \text{and} \quad \Nout(v) = \Nout(v^\pm) \text{.}$$
  A vertex subset $V \in \calP(\V)$ is \emph{structurally consistent} if and only if all its vertices are structurally equivalent with each others, and
  a vertex partition $\calV \in \frakP(\V)$ is \emph{structurally consistent} if and only if all its vertex subsets are structurally consistent.
  We respectively mark $\tilcalP(\V)$ and $\tilfrakP(\V)$ the sets of structurally-consistent vertex subsets and vertex partitions:
  $$V \in \tilcalP(\V) \quad \Leftrightarrow \quad \forall \vv \in V^2, \quad v \sim v^\pm \text{.}$$
  $$\calV \in \tilfrakP(\V) \quad \Leftrightarrow \quad \forall V \in \calV, \quad V \in \tilcalP(\V) \text{.}$$
\end{definition}
The lower part of Figure~\ref{fig:simple_example} uses the fact that $v_1 \sim v_2 \sim v_3$ and that $v_4 \sim v_5$ to define a structurally-consistent vertex partition $\calV = \{V_1,V_2\}$ made of two structurally-consistent vertex subsets $V_1 = \{v_1,v_2,v_3\}$ and $V_2=\{v_4,v_5\}$.\\

Because all vertices belonging to a structurally-consistent vertex subset have the exact same neighbourhoods, one can use this structural redundancy to simplify the graph representation.
Such a compression first consists in aggregating all vertices in structurally-consistent subsets to form \emph{compressed vertices}, then in aggregating all edges between couples of structurally-consistent subsets to form \emph{compressed edges}.
The resulting \emph{compressed graph} provides a smaller, yet complete description of the initial one.
\begin{definition}[Compressed Directed Graph]~\\
  Given a directed graph $G = (\V,\E)$ and a structurally-consistent vertex partition $\calV \in \tilfrakP(\V)$, the \emph{compressed directed graph} $\calV(G) = (\calV,\calE)$ is the graph such that:
  \begin{itemize}
  \item $\calV$ is the set of \emph{(compressed) vertices};
  \item $\calE \subseteq \calVxV$ is the set of \emph{(compressed) directed edges} such that:
  \end{itemize}
    $$\forall \VV \in \calVxV, \quad \VV \in \calE \quad \begin{aligned}
      \Leftrightarrow \quad \forall v \in V, \; \forall v^\pm \in V^\pm, \; & \vv \in \E \\                                                     
      \Leftrightarrow \quad \exists v \in V, \; \exists v^\pm \in V^\pm, \; & \vv \in \E \text{.}
    \end{aligned}$$
    Note that both conditions are equivalent since $V$ and $V^\pm$ are structurally consistent.
\end{definition}
The lower part of Figure~\ref{fig:simple_example} shows the effect of such a compression on the graph's representations.
Regarding adjacency lists,
no more than one compressed edge is encoded between two given compressed vertices.
Regarding the adjacency matrix,
cells are merged into ``rectangular tiles'' containing only one binary value for each couple of compressed vertices.\\

These definitions lead to a well-known combinatorial problem that we call here the \emph{Lossless Graph Compression Problem} (Lossless GCP).
It simply consists in finding the quotient set of the structural equivalence relation, that is the smallest struc\-tu\-ral\-ly-consistent partition of $\V$.
\begin{definition}[The Lossless Graph Compression Problem]~\\
  Given a directed graph $G = (\V,\E)$, find a structurally-consistent vertex partition $\calV^* \in \tilfrakP(\V)$ with minimal size $|\calV^*|$:
  $$\calV^* \; = \; \argmin_{\calV \in \tilfrakP (\V)} \; |\calV| \text{.}$$
\end{definition}
In Figure~\ref{fig:simple_example}, the represented structurally-consistent vertex partition is the smallest:
One cannot find such another partition that contains fewer vertex subsets.
This is hence the most optimal lossless compression of the graph.

\subsection{Related Problems}


Note that \emph{structural} equivalence is the stricter form of vertex equivalence one might consider for graph analysis~\cite{Batagelj92, Hanneman05}.
Yet, other equivalence relations are traditionally used in the literature in social sciences for the detection of other kinds of structural patterns, such as
\emph{automorphic} equivalence (two vertices are equivalent if there is an isomorphic graph such that these vertices are interchanged)
and \emph{regular} equivalence~\cite{Borgatti92} (two vertices are equivalent if they are equally related to other equivalent classes).
Because these two latter equivalence relations are less strict, they induce smaller vertex partitions with bigger classes.
But more importantly, and contrary to structural equivalence, the resulting compression scheme is not reversible in the sense that one cannot find back the initial graph from the equivalent classes and their compressed edges.

Structural equivalence, and so the GCP, is also related to \emph{community detection}~\cite{Fortunato10}, also known as \emph{graph clustering}, that is a classical problem for graph analysis which consists in finding groups of vertices that are strongly connected with each others while being loosely connected to other groups.
However, dense and isolated clusters are only particular examples of structurally-consistent classes.
They correspond to dense diagonal blocks in the adjacency matrix.
The notion of structural equivalence is more generally interested in groups of vertices with similar relational patterns, that is in any block of equal-density within the adjacency matrix (not necessarily dense and not necessarily on the diagonal, as in other work focusing on \emph{block compression}~\cite{LeFevre10, Toivonen11, Serafino13}).
Hence, the GCP is more strongly related to the family of \emph{edge compression techniques}~\cite{Dwyer13} such as \emph{modular decomposition}\footnote{Not to be confused with modularity-based clustering, which is a form of community detection.}~\cite{Serafino13}, \emph{matching neighbours}, and \emph{power graph analysis}~\cite{Ahnert14}.
In the latter, one is searching for groups of vertices that have similar relation patterns of any sort.
Because it is more generally interested in the compression of equal-density blocks, the GCP can lastly be seen as a strict instance of \emph{block modelling}~\cite{Lorrain71, Batagelj92, Borgatti92, Guigoures12}, another classical method of network analysis that relies on structural equivalence to discover roles and positions in social networks.

\subsection{Possible Generalisations}

This first formulation of the GCP is restricted to static simple graphs.
Moreover, it only allows lossless compression, that is compression of vertices with \emph{identical} neighbourhoods, which is a quite stringent and unrealistic condition for empirical research.
In what follows, we list the requirements to formulate a more general and more flexible optimisation problem allowing for the lossy compression of temporal graphs.

\begin{description}
\item [From simple graphs to multigraphs (see \ref{ssec:unweighted}).]
  This first version of the GCP is restricted to \emph{simple graphs} (no more than one edge between two given vertices).
  Yet, it is easily generalisable to \emph{multigraphs} (multiple edges are allowed between two given vertices).
  Such a generalisation has two advantages.
  First, multigraphs are strictly more general than simple graphs since simple graphs can be considered as a particular cases of multigraphs. 
  Second, multigraphs are more consistent with the lossy compression scheme later presented since the end result of lossy compression is not necessarily a simple graph (as edges are aggregated into multiedges during compression).
  
\item [From lossless to lossy compression (see \ref{ssec:lossy}).]
  This first version of the GCP is \emph{lossless} in the sense that the result of compression contains all the information that is required to errorlessly build back the initial graph.
  However, such a lossless compression -- relying on \emph{exact} equivalence -- is quite inefficient in the case of real graphs within which \emph{identical} neighbourhoods are quite unlikely.
  One hence needs a measure of \emph{information loss} to allow for a more flexible compression scheme.

\item [From vertex to edge partitions (see \ref{ssec:powergraph}).]
  This first version of the GCP consists in finding an interesting \emph{vertex partition} to compress the graph, thus inducing a partition of its edges.
  This relates to classical approaches such as \emph{modular decomposition} where subsets of vertices (modules) that have similar neighbourhoods are exploited to compressed the graph's structure.
  However, this can be generalised to the direct search for \emph{edge partitions}, that is the search for interesting edge subsets that do not all necessarily rest on similar vertex subsets.
  This relates to less known approaches such as \emph{power-graph decomposition} that allows for a more subtle analysis of the graph's structure.

\item [Adding constraints to the set of feasible vertex subsets (see \ref{ssec:constrained}).]
  In this first version of the GCP, one considers any possible vertex subset as a potential candidate for compression, thus leading to an \emph{unconstrained} compression scheme.
  However, in order to represent and to preserve \emph{additional constraints} that might apply on the vertex structure, one might want to only consider ``feasible'' vertex subsets when searching for an optimal partition.
  This requires to integrate such additional constraints within the compression scheme.
  
\item [From static graphs to link streams (see \ref{ssec:dynamic}).]
  Our last generalisation step consists in integrating a temporal dimension within the optimisation problem in order to deal with the compression of \emph{link streams}.
  The structural equivalence relation hence needs to be redefined with respect to this additional dimension and equivalent classes will then be only valid on given time intervals.
  In this context, one is hence searching for aggregates that partition the Cartesian product of the vertex set and of the temporal dimension.

\end{description}


\pagebreak
\section{Generalisation: From Lossless Static Graphs to Lossy Mutlistreams}
\label{sec:generalisation}


\subsection{From Graphs to Multigraphs}
\label{ssec:unweighted}

Most approaches in the domain of graph theory focus on the analysis of \emph{simple graphs}, that is graphs for which at most one edge is allowed between two vertices, thus represented as binary adjacency matrices.
This is also the case when it comes to the field of graph compression (see for example~\cite{Borgatti92, Navlakha08, Zhang10, LeFevre10, Dwyer13}). 
Yet, in the scope of this article, we aim at the compression of \emph{multigraphs}, that is graphs for which multiple edges are allowed between two vertices, thus represented as integer adjacency matrices.
As simple graphs are special cases of multigraphs, the resulting approach is necessarily more general.

In some articles on graph compression, the generalisation to multigraphs would be quite straightforward as the result of compression -- that is the compressed graph -- already is, in fact, a multigraph (see for example~\cite{LeFevre10}).
Even if not explicitly formalised, research perspectives in that direction are sometimes provided~\cite{Ahnert14}.
Yet, other approaches natively deals with multigraph compression by directly taking into account, within the compression scheme, the presence of multiple edges~\cite{Toivonen11, Serafino13}.
Statistical methods for \emph{variable co-clustering} also offers compression frameworks that are designed for numerical (non-binary) matrices~\cite{Dhillon03, Guigoures12, Narmadha16}.
This is the approach we choose here by directly working with multigraphs.

\begin{definition}[Directed Multigraph]~\\
  \label{def:multigraph}
  A \emph{directed multigraph} $\MG = (\V,e)$ is characterised by:
  \begin{itemize}
  \item A set of \emph{vertices} $\V$;
  \item A multiset of \emph{directed edges} $(\bVxV, e)$\\
    where $e : \bVxV \rightarrow \mathbb{N}$ is the \emph{edge function}, that is the multiplicity function counting the number of edges $e\vv \in \mathbb{N}$ going from a given source vertex $v \in \V$ to a given target vertex $v^\pm \in \V$.\\[1ex]
    We also define the additive extension of the edge function on couples of vertex subsets:
    $$e : \calP(\V){\times}\calP(\V) \rightarrow \mathbb{N} \quad \text{such that} \quad \ef\VV \; = \; \sum_{\vv \in \VxV} {\ef\vv} \text{.}$$
  It simply counts the number of edges going from any vertex of a given source subset $V \in \calP(\V)$ to any vertex of a given target subset $V^\pm \in \calP(\V)$.
  In particular, $e\bVV$ is the total number of edges in the multigraph, $e(\V,v)$ is the in-coming degree of $v$ and $e(v,\V)$ its out-going degree.
  \end{itemize}
\end{definition}
The upper part of Figure~\ref{fig:lossy_graph} gives an example of directed multigraph made of $|\V| = 5$ vertices, that is $|\bVxV| = 25$ multiedges, and $e\bVV = 40$ edges distributed within $\bVxV$.
It is represented in the form of \emph{adjacency lists} (on the left), where each multiedge is represented as an arrow which width is proportional to the number of edges $e\vv$ going from a source vertex $v$ to a target vertex $v^\pm$, as well as in the form of an \emph{adjacency matrix} (on the right), where the edge function is represented as an integer matrix of size $|\V|{\times}|\V|$.

Structural equivalence could then be generalised to multigraphs in order to define, as done previously for simple graphs, a Lossless Multigraph Compression Problem (MGCP).
In few words, the structural equivalence relation would be defined on directed multigraphs by the equality of the edge function:
Two vertices are hence structurally equivalent if and only if they are each connected the same number of times to the different graph's vertices.
However, as we are interested in this article in lossy compression, we directly consider an alternative version of the MGCP that relies on a stochastic relaxation of the structural equivalence relation and on an appropriate measure of information loss.

\begin{figure}
  \centering
  \mycenter{\begin{tikzpicture}[scale=0.65]

  \node [align=center] (graph) at (0,-0.25) {
    \begin{tikzpicture}
      \tikzPrintSimpleCoord
      
      \node [mynode] (v1) at (c1) {$v_1$};
      \node [mynode] (v2) at (c2) {$v_2$};
      \node [mynode] (v3) at (c3) {$v_3$};
      \node [mynode] (v4) at (c4) {$v_4$};
      \node [mynode] (v5) at (c5) {$v_5$};

      \tikzPrintLoop{v1}{1}{1}
      \tikzPrintEdge{v1}{v3}{2}
      \tikzPrintEdge{v1}{v4}{9}
      \tikzPrintEdge{v1}{v5}{2}

      \tikzPrintLoop{v2}{1}{1}
      \tikzPrintEdge{v2}{v4}{10}
      \tikzPrintEdge{v2}{v5}{6}

      \tikzPrintEdge{v3}{v1}{1}
      \tikzPrintEdge{v3}{v2}{5}
      \tikzPrintLoop{v3}{1}{1}
      \tikzPrintEdge{v3}{v4}{9}
      \tikzPrintEdge{v3}{v5}{5}

      \tikzPrintEdge{v4}{v1}{4}
      \tikzPrintEdge{v4}{v2}{4}
      \tikzPrintEdge{v4}{v3}{5}
      \tikzPrintLoop{v4}{11}{3}
      \tikzPrintEdge{v4}{v5}{10}

      \tikzPrintEdge{v5}{v1}{5}
      \tikzPrintEdge{v5}{v2}{5}
      \tikzPrintEdge{v5}{v3}{4}
      \tikzPrintEdge{v5}{v4}{10}
      \tikzPrintLoop{v5}{10}{3}
    \end{tikzpicture}
  };

  \node [align=center] (matrix) at (7,0) {
    \begin{tikzpicture}[scale=0.8]

      \tikzPrintLabels

      \tikzPrintCell{0}{4}{$1$}{0.008}
      \tikzPrintCell{1}{4}{$0$}{0.000}
      \tikzPrintCell{2}{4}{$2$}{0.017}
      \tikzPrintCell{3}{4}{$9$}{0.075}
      \tikzPrintCell{4}{4}{$2$}{0.017}

      \tikzPrintCell{0}{3}{$0$}{0.000}
      \tikzPrintCell{1}{3}{$1$}{0.008}
      \tikzPrintCell{2}{3}{$0$}{0.000}
      \tikzPrintCell{3}{3}{$10$}{0.083}
      \tikzPrintCell{4}{3}{$6$}{0.050}

      \tikzPrintCell{0}{2}{$1$}{0.008}
      \tikzPrintCell{1}{2}{$5$}{0.042}
      \tikzPrintCell{2}{2}{$1$}{0.008}
      \tikzPrintCell{3}{2}{$9$}{0.075}
      \tikzPrintCell{4}{2}{$5$}{0.042}

      \tikzPrintCell{0}{1}{$4$}{0.033}
      \tikzPrintCell{1}{1}{$4$}{0.033}
      \tikzPrintCell{2}{1}{$5$}{0.042}
      \tikzPrintCell{3}{1}{$11$}{0.092}
      \tikzPrintCell{4}{1}{$10$}{0.083}

      \tikzPrintCell{0}{0}{$5$}{0.042}
      \tikzPrintCell{1}{0}{$5$}{0.042}
      \tikzPrintCell{2}{0}{$4$}{0.033}
      \tikzPrintCell{3}{0}{$10$}{0.083}
      \tikzPrintCell{4}{0}{$10$}{0.083}
    \end{tikzpicture}
  };

  \node [align=center] (aggGraph) at (0,-6.25) {
    \begin{tikzpicture}
      \tikzPrintSimpleCoord

      \node (v1) at (c1) {};
      \node (v2) at (c2) {};
      \node (v3) at (c3) {};
      \node (v4) at (c4) {};
      \node (v5) at (c5) {};



      \node [mynode] (v123) at ($0.5*(v1)+0.5*(v3)-0.25*(v2)$) {$\{v_1, v_2, v_3\}$};
      \node [mynode] (v45) at ($0.5*(v4)+0.5*(v5)$) {$\{v_4, v_5\}$};

      \tikzPrintLoop{v123}{1.2}{1}
      \tikzPrintEdge{v123}{v45}{6.8}
      \tikzPrintEdge{v45}{v123}{4.5}
      \tikzPrintLoop{v45}{10.25}{3}

    \end{tikzpicture}
  };

  \node [align=center] (aggMatrix) at (7,-6) {
    \begin{tikzpicture}[scale=0.8]
      
      \node [labels] at (1,5) {$\{v_1, v_2, v_3\}$};
      \node [labels] at (3.5,5) {$\{v_4, v_5\}$};

      \node [labels] at (-1,0.5) {\rotatebox{90}{$\{v_5, v_4\}$}};
      \node [labels] at (-1,3) {\rotatebox{90}{$\{v_3, v_2, v_1\}$}};

      \tikzPrintRectangle{0}{5}{3}{2}{$11$}{0.010}
      \tikzPrintRectangle{3}{5}{5}{2}{$41$}{0.057}
      \tikzPrintRectangle{0}{2}{3}{0}{$27$}{0.038}
      \tikzPrintRectangle{3}{2}{5}{0}{$41$}{0.085}
    \end{tikzpicture}
  };

\end{tikzpicture}

}
  \caption{
    Lossy compression of a 5-vertex, 40-edge multigraph (above) into a 2-vertex, 40-edge multigraph (below).
    In the \emph{adjacency-list representation} (on the left), the width of arrows is proportional to the edge function, that is to the number of edges going from a source vertex to a target vertex.
  }
  \label{fig:lossy_graph}
\end{figure}
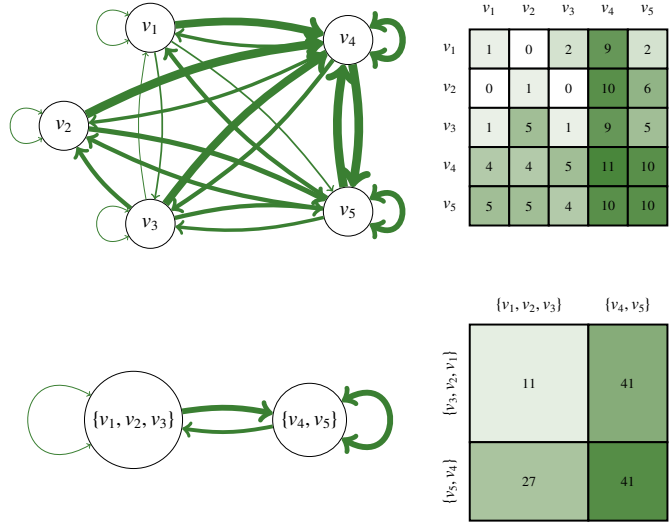

\begin{figure}
  \centering
  \mycenter{\input{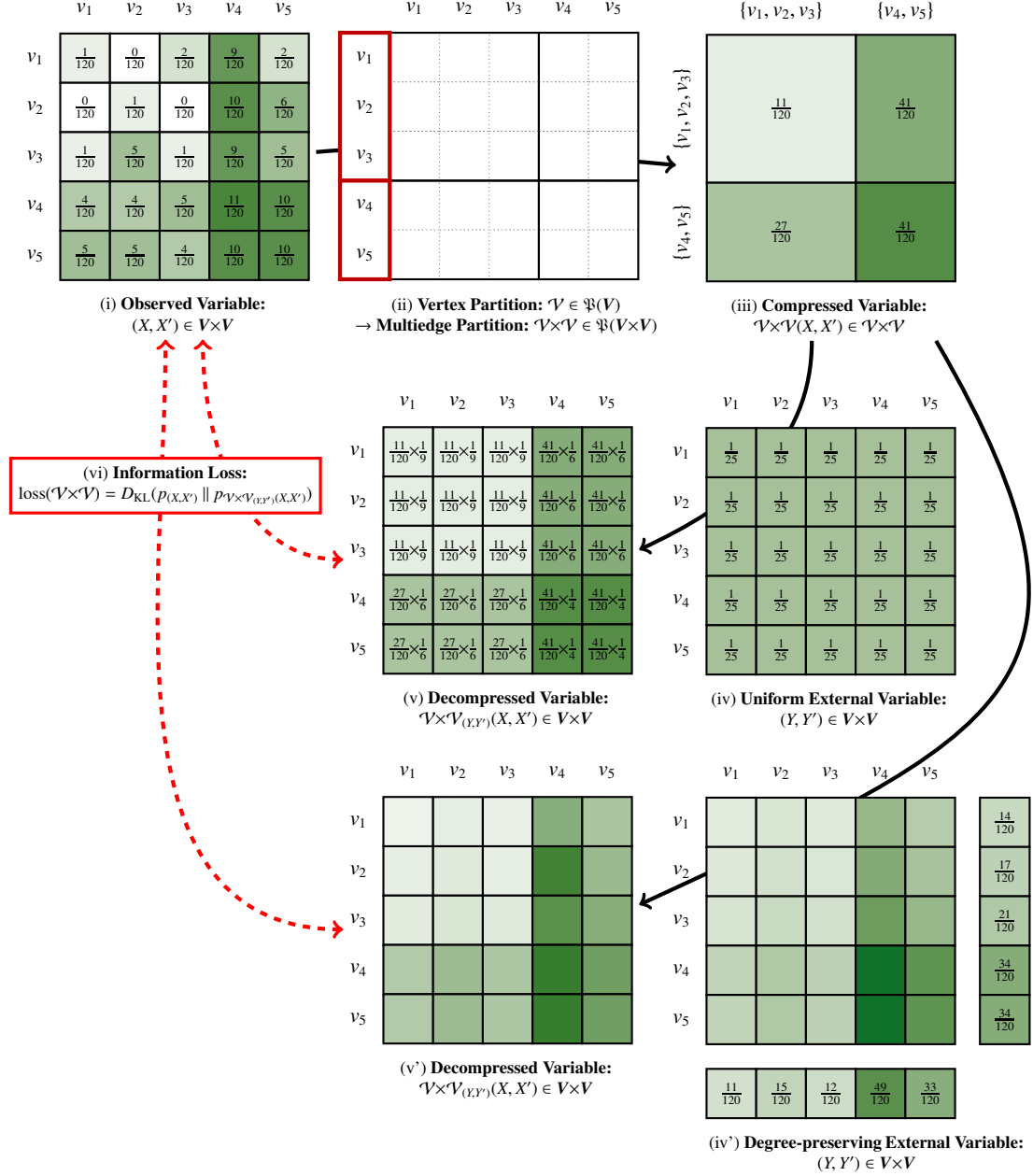}}
  \caption{
    Lossy compression consists in
    (i)~modelling the multigraph as a random variable $\XX$ having the empirical distribution of edges in $\bVxV$,
    (ii)~choosing a vertex partition $\calV$, and so a multiedge partition $\calVxV$ that is used to compress $\XX$,
    (iii)~computing the distribution of the resulting compressed variable $\calVxV\XX$ by applying partition $\calVxV$ onto $\XX$,
    (iv)~taking an external variable $\YY$ (for example (iv)~uniformly distributed or (iv')~preserving the degree profile of vertices) to project back the distribution of $\calVxV\XX$ into $\bVxV$,
    (v)~computing the distribution of the resulting decompressed variable $\calVxV_\YY\XX$ by first choosing $\XX$, then choosing $\YY$ conditioned on $\calVxV(\pYY) = \calVxV\XX$,
    and (vi)~comparing the distribution of $\XX$ with the distribution of $\calVxV_\YY\XX$ by using information-theoretic measures such as the entropy of $\XX$ relative to $\calVxV_\YY\XX$.
  }
  \label{fig:lossy_compression}
\end{figure}

\subsection{From Lossless to Lossy Compression}
\label{ssec:lossy}

Information-theoretic compression first requires a stochastic model of the data, that is a model of the multigraph to be compressed.
The measure of information loss that we present in this subsection has been previously introduced for the aggregation of geographical data~\cite{RLP_TCCI14} and for the summarisation of execution traces of distributed systems~\cite{RLP_CLUSTER14}.
The first contribution of this article in this regard is the application of this measure to graph compression.
Moreover, the underlying stochastic models were not made explicit in previous work.
Our second contribution is hence the thorough formalisation of the graph model we use, in order to properly justify and interpret the resulting measure.
In few words, we model a multigraph as a set of edges that are stochastically distributed within the two-dimensional set of multiedges $\bVxV$:
Each edge has a particular location within this space, characterised by its source vertex and its target vertex.
The {edge function} $e : \bVxV \rightarrow \mathbb{N}$ hence characterises the \emph{empirical distribution} of the edges within the multigraph, thus allowing to model the data as a discrete random variable $\XX$ taking on $\bVxV$.
\begin{definition}[Observed Variable]~\\
  \label{def:empirical}
  The \emph{observed variable} $\XX \in \bVxV$ associated with a multigraph $\MG = (\V,e)$ is a couple of discrete random variables having the empirical distribution of edges in $\MG$:
  $$\Pr(\XX = \vv) \quad = \quad \frac{\ef\vv}{e\bVV} \quad \defeq \quad p_\XX\vv \text{.}$$
\end{definition}
In other terms, $p_\XX\vv$ represents the probability that, if one chooses an edge at random among the $e\bVV$ edges of the multigraph, it will go from the source vertex $v$ to the target vertex $v^\pm$.
For example, matrix~(i) in Figure~\ref{fig:lossy_compression} represents the distribution of the observed variable associated with the multigraph of Figure~\ref{fig:lossy_graph}.\\

We then define the edge distribution of a multigraph that have been compressed according to a vertex partition $\calV \in \frakP(\V)$ by defining a second random variable taking on the multiedge partition $\calVxV \in \frakP(\bVxV)$.
\begin{definition}[Compressed Variable]~\\
  \label{def:compressed}
  The \emph{compressed variable} $\calVxV\XX \in \calVxV$ associated with an observed variable $\XX \in \bVxV$ and a vertex partition $\calV \in \frakP(\V)$ is the unique couple of vertex subsets in $\calVxV$ that contains $\XX$:
  $$\calVxV\XX \quad = \quad (\calV(X),\calV(X^\pm)) \quad \in \quad \calVxV \text{.}$$ 
  It hence has the following distribution\footnotemark{}:
  \begin{align*}
    \Pr(\calVxV\XX = \VV) \quad = \quad \frac{\ef\VV}{\ef\bVV} \quad \defeq \quad p_{\calVxV\XX} \VV \text{.}
  \end{align*}
\end{definition}
\footnotetext{
  By applying the law of total probability:
  \begin{align*}
    \Pr(\calVxV\XX = \VV) \; & = \; \sum_{\vv \in \bVxV} {\underbrace{\Pr(\calVxV\XX = \VV  \given  \XX = \vv)}_{= 1  \text{\ if\ } \vv \in \VxV, \text{\ and\ } 0 \text{\ else}} \; \Pr(\XX = \vv)}\\
                             & = \; \sum_{\vv \in \VxV} {\Pr(\XX = \vv)} \; = \; \sum_{\vv \in \VxV} {\frac{e\vv}{e\bVV}} \; = \; \frac{\ef\VV}{\ef\bVV}
  \end{align*}
  Note that, more generally, compression could be defined for any (possibly stochastic) function of the observed variable $\XX$, thus modelling what is sometimes called a ``soft partitioning'' of the vertices.
  The information-theoretic framework presented in this article, along with all the measures it contains, are straightforwardly generalisable to such setting.
  However, because it is often much easier to interpret the result of compression when it is based on ``hard partitioning'', especially in the case of vertex partitioning, we focus in this article on this simpler setting.
}
In other terms, $p_{\calVxV\XX} \VV$ represents the probability that, if one chooses an edge at random among the $e\bVV$ edges of the multigraph, it will go from a vertex of the source subset $V$ to a vertex of the target subset $V^\pm$.
For example, matrix~(ii) in Figure~\ref{fig:lossy_compression} represents the multiedge partition $\calVxV$ induced by the vertex partition $\calV = \{\{v_1,v_2,v_3\},\{v_4,v_5\}\}$ and matrix~(iii) then represents the distribution of the resulting compressed variable $\calVxV\XX$.\\

In order to quantify the information that has been lost during this compression step, we propose to compare the information that is contained in the initial multigraph (that is in the observed variable $\XX$) with the information that is contained in the compressed multigraph (that is in the compressed variable $\calVxV\XX$).
To do so, we project back the compressed distribution onto the initial value space $\bVxV$ by defining a third random variable, the \emph{external variable} $\YY \in \bVxV$, that models additional information that one might have at his or her disposal when trying to decompress the multigraph.
It is hence assumed to have a distribution that is somehow ``informative'' of the initial distribution.
It then induces a fourth variable, the \emph{decompressed variable} $\calVxV_\YY\XX$, that models an approximation of the initial multigraph inferred from the combined knowledge of the compressed variable $\calVxV\XX$ and of the external variable $\YY$.
This last variable is 
hence defined according to the distribution of the external variable within the multiedge subsets. 
\begin{definition}[Decompressed Variable]~\\
  \label{def:decompressed}
  The \emph{decompressed variable} $\calVxV_\YY\XX \in \bVxV$ associated with an observed variable $\XX \in \bVxV$, a vertex partition $\calV \in \frakP(\V)$, and an external variable $\YY \in \bVxV$, is the result of this external variable $\YY$ conditioned by its compression $\calVxV(\pYY)$ being equal to the result of the compressed variable $\calVxV\XX$:
  $$\calVxV_\YY\XX \quad = \quad \YY  \given  \calVxV(\pYY) = \calVxV\XX \quad \in \quad \bVxV \text{.}$$
  It hence has the following distribution\footnotemark{}:
  $$\Pr (\calVxV_\YY\XX = \vv) \quad = \quad \frac{\ef(\calVxV\vv)}{\ef\bVV} \; \frac{p_\YY \vv}{p_{\calVxV\YY} (\calVxV\vv)} \quad \defeq \quad p_{\calVxV_\YY\XX} \vv \text{.}$$
\end{definition}
\footnotetext{
  By applying the law of total probability:
  \begin{align*}
    \Pr(\calVxV_\YY\XX = \vv) \quad & = \quad \Pr (\YY = \vv  \given  \calVxV(\pYY) = \calVxV\XX)\\
                                    & = \quad \sum_{\VV \in \calVxV} \underbrace{\Pr (\YY = \vv  \given  \calVxV(\pYY) = \VV)}_{
                                      = \; 0 \text{ if } \VV \; \neq \; \calVxV\vv
                                      } \; \Pr (\calVxV\XX = \VV) \\
                                    & = \quad \Pr (\YY = \vv  \given  \calVxV(\pYY) = \calVxV\vv) \; \Pr (\calVxV\XX = \calVxV\vv) \\
                                    & = \quad \frac{\Pr (\YY = \vv)}{\Pr (\calVxV(\pYY) = \calVxV\vv)} \; \frac{\ef(\calVxV\vv)}{\ef\bVV} \text{.}
  \end{align*}    
}
In other terms, $p_{\calVxV_\YY\XX}\vv$ represents the probability that, if one
(i)~chooses an edge $\uu$ at random among the $e\bVV$ edges of the multigraph,
(ii)~considers its compressed multiedge subset $\calVxV\uu = \VV \in \calVxV$,
and (iii)~chooses a source vertex within $V$ and a target vertex within $V^\pm$ according to the distribution $p_\YY$ of the external variable within this multiedge subset,
then one will result with an edge going from the source vertex $v$ to the target vertex $v^\pm$.
For example, matrix~(iv) in Figure~\ref{fig:lossy_compression} represents a uniformly-distributed external variable $\YY \in \bVxV$ that is used to decompress $\calVxV\XX$ (see \emph{Blind Decompression} below).
Matrix~(iv') represents another such external variable (see \emph{Degree-preserving Decompression} below).
Matrices~(v) and~(v') then represent the distribution of the resulting uniformly decompressed variable $\calVxV_\YY\XX$.

\paragraph{Blind Decompression}
When the external variable $\YY$ is uniformly distributed on $\bVxV$, the decompression step is done without any additional information about the initial edge distribution: 
$$p_{\YY}\vv \; = \; \frac{1}{|\bVxV|} \quad \Rightarrow \quad p_{\calVxV_\YY\XX} \vv \; = \; \frac{\ef(\calVxV\vv)}{\ef\bVV} \frac{1}{|\calVxV\vv|} \text{.}$$
In this case, only the knowledge of the compressed variable hence is exploited.
The decompressed variable is hence the result of a uniform trial among the multiedges contained in $\calVxV\XX$, that is a ``maximum-entropy sampling'' guarantying that no additional information has been injected during decompression.

\paragraph{Reversible Decompression}
To the contrary, when the external variable $\YY$ has the same distribution than the observed variable $\XX$, then the decompression step is done with a full knowledge of the initial edge distribution:
$$p_\YY\vv \; = \; p_\XX\vv \quad \Rightarrow \quad {\calVxV_\YY\XX} \vv \; = \; \frac{\ef\vv}{\ef\bVV} \; = \; p_\XX\vv \text{.}$$
In this case, the decompressed variable also has the same distribution than the observed variable, meaning that one fully restores the initial multigraph when decompressing.

\paragraph{Degree-preserving Decompression}
An intermediary example of external information can be derived from the knowledge of the vertex degrees in the initial multigraph:
\begin{multline*}
  p_{\YY}\vv \; = \; p_X(v) \; p_{X^\pm}(v^\pm) \; = \; \displaystyle\frac{\ef(v,\V)}{\ef\bVV} \frac{\ef(\V,v^\pm)}{\ef\bVV} \\
  \Rightarrow \quad p_{\calVxV_\YY\XX} \vv \; = \; \displaystyle\frac {\ef(\calVxV\vv)} {\ef\bVV} \displaystyle\frac {\ef(v,\V) \; \ef(\V,v^\pm)} {\ef(\calV(v),\V) \; \ef(\V,\calV(v^\pm))} \text{.}
\end{multline*}
In this case, the decompression step takes into account the initial vertex degrees and the resulting multigraph hence has the same degree profile than the initial one.
The corresponding generative model is hence similar to the one of a \emph{configuration model}~\cite{Hofstad16}:
Multigraphs are sampled according to the compressed variable, while also preserving the initial degree profile.\\

Now that we have defined compression and decompression as sequential operations on stochastic variables, we exploit a classical measure of information theory to quantify the information that is lost during such a process.
Intuitively, it consists in comparing the initial edge distribution (the one of the observed variable $\XX$) with the approximated edge distribution  (the one of the decompressed variable $\calVxV_\YY\XX$). 
In this article, we propose to do so by using the \emph{relative entropy} of these two distributions -- also known as the \emph{Kullback-Leibler divergence}~\cite{Kullback51, Cover91} -- as it is the most canonical measure of dissimilarity provided by information theory to compare an approximated probability distribution to a real one.
\begin{definition}[Information Loss]~\\
  \label{def:information_loss}
  The \emph{information loss} induced by a vertex partition $\calV \in \frakP(\V)$ on an observed variable $\XX \in \bVxV$, and according to an external variable $\YY \in \bVxV$, is given by the entropy of $\XX$ relative to ${\calVxV_\YY\XX}$: 
\begin{align*}
  \lossf{\XX}{\calVxV}{\YY} \quad 
                               & = \quad \sum_{\vv \in \bVxV} p_\XX\vv \log_2 \left( \frac{p_\XX\vv}{p_{\calVxV_\YY\XX} \vv} \right)\\
                               & = \quad \sum_{\vv \in \bVxV} \frac{\ef\vv}{\ef\bVV} \log_2 \left( \frac{\ef\vv}{\ef(\calVxV\vv)} \bigg/ \frac{p_\YY \vv}{p_{\calVxV\YY} (\calVxV\vv)} \right)
\end{align*}
Note that information loss is \emph{additively decomposable}.
It can be expressed as a sum of information losses defined at the subset level instead of at the partition level:
$$\lossf{\XX}{\calVxV}{\YY} \; = \; \sum_{\VV \in \calVxV} {\lossf{\XX}{\pVV}{\YY}} \quad \text{with\ } \lossf{\XX}{\pVV}{\YY} \; = \; \sum_{\vv \in \VxV} \frac{\ef\vv}{\ef\bVV} \log_2 \left( \frac{\ef\vv}{\ef\VV} \bigg/ \frac{p_\YY \vv}{p_{\calVxV\YY} \VV} \right) \text{.}$$

\end{definition}
Intuitively, this measure considers the following reconstruction task:
Imagine that all the edges of a multigraph have been ``detached'' from their vertices and put into a bag.
An observer is now taking one edge out of the bag and tries to guess its initial location, that is its source and target vertices.
We then compare two peoples trying to do so:
One having a perfect knowledge of the distribution $p_\XX$ of the edges in the initial multigraph (\emph{e.g.}, matrix~(i) in Figure~\ref{fig:lossy_compression});
The second only having an approximation $p_{\calVxV_\YY\XX}$ of this distribution, obtained through the compression, then the decompression of the initial distribution (\emph{e.g.}, matrices~(v) and~(v') in Figure~\ref{fig:lossy_compression}).
Relative entropy then measures the average quantity of \emph{additional} information (in bits per edge) that the second observer needs in order to make a guess that is as informed as the guess of the first observer.
In other words, relative entropy quantifies the information that has been lost during compression,
that is no longer contained in the compressed graph, and that cannot be retrieved from the knowledge of the external variable.

\paragraph{Blind Decompression}
When the external variable is uniformly distributed on $\bVxV$, relative entropy simply quantifies the information that has been lost during compression, without the help of any additional information:
$$p_{\YY}\vv \; = \; \frac{1}{|\bVxV|} \quad \Rightarrow \quad \lossf{\XX}{\pVV}{\YY} \; = \; \sum_{vv \in \VxV} {\frac{\ef\vv}{\ef\bVV} \log_2 \left(\frac{\ef\vv}{\ef\VV} |\VxV| \right)} \text{.}$$

\paragraph{Reversible Decompression}
When the external variable $\YY$ has the same distribution than the observed variable $\XX$, compression then induces no information loss -- whatever the chosen vertex partition $\calV$ -- since all the information required to reconstruct the initial multigraph is reinjected during the decompression step:
$$p_\YY\vv \; = \; p_\XX\vv \quad \Rightarrow \quad \lossf{\XX}{\pVV}{\YY} \; = \; 0\text{.}$$

\paragraph{Degree-preserving Decompression}
In this intermediary context, relative entropy quantifies the information that has been lost during compression, and that cannot be retrieved from the additional knowledge of the vertex degrees in the initial multigraph:
\begin{multline*}
  p_{\YY}\vv \; = \; p_X(v) \; p_{X^\pm}(v^\pm) \; = \; \displaystyle\frac{\ef(v,\V)}{\ef\bVV} \frac{\ef(\V,v^\pm)}{\ef\bVV} \\
  \Rightarrow \quad \lossf{\XX}{\pVV}{\YY} \; = \; \displaystyle\sum_{\vv \in \VxV} { \frac{\ef\vv}{\ef\bVV} \log_2 \left( \frac {\ef\vv} {\ef\VV} \bigg/ \frac {\ef(v,\V) \; \ef(\V,v^\pm)} {\ef(V,\V) \; \ef(\V,V^\pm)} \right)} \text{.}
\end{multline*}

\paragraph{Related Measures}

Relative entropy is one among many measures that can be found in the literature to quantify information loss in graph compression.
Given an initial multigraph and a decompressed one, which is described by the approximated edge distribution, any measure of weighted graph similarity may be relevant to quantify the impact of compression~\cite{Zhou09}:
\emph{E.g.},
the percentage of edges in common,
the size of the maximum common subgraph or of the minimum common supergraph,
the {edit distance}, that is the insertion and removal of vertices and edges needed to go from one graph to another~\cite{He06},
or any measure aggregating the similarities of vertices from graph to graph (\emph{e.g.}, Jaccard index, Pearson coefficient, cosine similarity on vertex neighbourhoods).

More sophisticated graph-theoretical measures go beyond the mere level of edges by taking into account paths within the two compared graphs~\cite{Toivonen11}:
``[T]he best path between any two nodes should be approximately equally good in the compressed graph as in the original graph, but the path does not have to be the same.''
More generally, query-based measures aim at quantifying the impact of compression on the results of goal-oriented queries regarding the graph structure:
\emph{E.g.}, queries about shortest paths, about degrees and adjacency, about centrality and community structures (see for example~\cite{Feder95, Henandez12, Dhabu13}).
The expected difference between the results of queries on the initial graph and the results of queries on the decompressed graph thus provides a reconstruction error that serves as a goal-oriented information loss~\cite{LeFevre10}.
To some extent, we are interested in this article in adjacency-oriented queries, that is the most canonical ones, taking the perspective of the weighted adjacency matrices of the two compared graphs:
\emph{What is the weight of the multiedge located between two given vertices of the initial multigraph?}
Relative entropy measures the expected error when answering this query from the only knowledge of the compressed graph.

More generally, this perspective is related to the density profile of edges within the graph. 
Hence, density-based measures~\cite{Zhang10} seems more relevant than other traditional connectivity measures:
\emph{E.g.}, the Euclidean distance or the mean squared error between the two density matrices~\cite{Toivonen11, Serafino13},
the average variance within matrix blocks~\cite{Borgatti92}, and many other measures inherited from traditional block modelling methods~\cite{Batagelj92}.
Note that, when it comes to the latter, the stochastic model underlying our compression framework is similar, but not equivalent to the one of block modelling.
The compressed matrix describes in our case the parameters of a multinomial distribution from which the graph's edges are sampled, whereas it describes in the case of block modelling the parameters of $|\bVxV|$ independent Bernouilli distributions.
In other words, our model gives the probability that an edge -- taken at random -- is located between two given vertices, and not the probability that an edge exists between two given vertices (see for example block model compression in~\cite{LeFevre10}).

Because of this particular stochastic model, we chose in this article an information-theoretic approach to measure information loss.
This allows to derive a measure that is clearly in line with the defined model and which can be easily interpreted within the realm of information theory.
Among tools provided by this theory, other approaches use the principles of \emph{minimum description length}~\cite{Navlakha08, LeFevre10} to compress a graph using the density-based model in an optimal fashion.
In the same line of thinking, traditional tools for {Bayesian inference} propose to interpret the compressed graph as a generative model and the initial graph as observed data, then computes the likelihood of the data given the model as a measure of information loss~\cite{Guigoures12}.
A similar Bayesian interpretation of relative entropy could be given, as it measures the difference of likelihood between two generative models of the multigraph:
One corresponding to $e\bVV$ independent trials with the empirical distribution of the multigraph's edges;
The second corresponding to $e\bVV$ independent trials with the distribution obtained through compression, and then decompression.
Similarly, \emph{co-clustering}~\cite{Dhillon03} interprets the graph's adjacency matrix as the joint probability distribution of two random variables, then finds two vertex partitions that minimise the loss in mutual information between these two variables~\cite{Cover91} from the initial graph to the compressed graph.
This is shown to be equivalent to minimising the relative entropy between the initial distribution and a decompressed distribution that preserves the marginal values.
It is hence equivalent to our measure of information loss in the particular case of a decompression scheme that takes into account additional information regarding the vertex degrees.

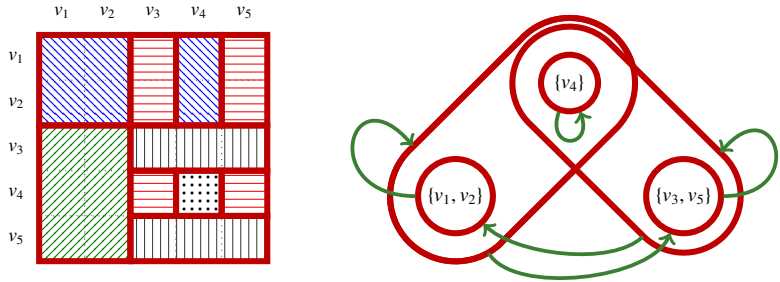
\begin{figure}[t!]
  \centering
  \begin{tikzpicture}[scale=0.75]

  \node [bigtext, align=center, anchor=south] at (0,0) {Grid multiedge partition $\calVxV \in \frakP(\bVxV)$\\
    \normalfont{that is the Cartesian product of a vertex partition $\calV \in \frakP(\V)$}};

  \begin{scope}[shift={(-5,-4.5)}, scale=0.8]
    \tikzPrintLabels
    \tikzPrintGrid

    \tikzPrintSAgg{0}{5}{1}{4}
    \tikzPrintSAgg{0}{4}{1}{2}
    \tikzPrintSAgg{0}{2}{1}{0}

    \tikzPrintSAgg{1}{5}{3}{4}
    \tikzPrintSAgg{1}{4}{3}{2}
    \tikzPrintSAgg{1}{2}{3}{0}

    \tikzPrintSAgg{3}{5}{5}{4}
    \tikzPrintSAgg{3}{4}{5}{2}
    \tikzPrintSAgg{3}{2}{5}{0}
    
    \tikzPrintAgg{-1}{5}{0}{4}
    \tikzPrintAgg{-1}{4}{0}{2}
    \tikzPrintAgg{-1}{2}{0}{0}
  \end{scope}

  \begin{scope}[shift={(3.75,-1.75)}]

    \node [mynode, agg] (v1) at (0,0) {$\{v_1\}$};
    \node [mynode, agg] (v23) at (-2,-2) {$\{v_2, v_3\}$};
    \node [mynode, agg] (v45) at (2,-2) {$\{v_4, v_5\}$};




    \draw [myedge, ->, line width=0.5mm] (v23) to [bend left=15] (v45);
    \draw [myedge, ->, line width=0.5mm] (v45) to [bend left=15] (v23);

    \draw [myedge, ->, line width=0.5mm] (v1) to [bend left=15] (v23);
    \draw [myedge, ->, line width=0.5mm] (v23) to [bend left=15] (v1);

    \draw [myedge, ->, line width=0.5mm] (v1) to [bend left=15] (v45);
    \draw [myedge, ->, line width=0.5mm] (v45) to [bend left=15] (v1);

    \draw [myedge, ->, line width=0.5mm] (v1) to [out=120,in=60,looseness=5] (v1);
    \draw [myedge, ->, line width=0.5mm] (v23) to [out=210,in=150,looseness=5] (v23);
    \draw [myedge, ->, line width=0.5mm] (v45) to [out=30,in=-30,looseness=5] (v45);  
  \end{scope}
\end{tikzpicture}

\vspace{2em}

\hspace{5.25ex} \begin{tikzpicture}[scale=0.75]
  \node [bigtext, align=center, anchor=south] at (2,0) {Cartesian multiedge partition $\calVV \in \frakP^\times(\bVxV)$\\
    \normalfont{consisting in Cartesian multiedge subsets $\VxV \in \calP(\bVxV)$}};

  \begin{scope}[shift={(-3.9,-4.5)}, scale=0.8]

    \tikzPrintLabels
    \tikzPrintGrid
    
    \tikzPrintAgg[pattern1]{0}{5}{2}{3}
    \tikzPrintAgg[pattern2]{2}{5}{3}{3}
    \tikzPrintAgg[pattern1]{3}{5}{4}{3}
    \tikzPrintAgg[pattern2]{4}{5}{5}{3}
    
    \tikzPrintAgg[pattern3]{0}{3}{2}{0}

    \tikzPrintAgg[pattern4]{2}{3}{5}{2}

    \tikzPrintAgg[pattern2]{2}{2}{3}{1}
    \tikzPrintAgg[pattern5]{3}{2}{4}{1}
    \tikzPrintAgg[pattern2]{4}{2}{5}{1}

    \tikzPrintAgg[pattern4]{2}{1}{5}{0}
  \end{scope}

  \begin{scope}[shift={(5,-1.75)}]

    \node [mynode, agg] (p4) at (0,0) {};
    \node [mynode, agg] (p12) at (-2,-2) {};
    \node [mynode, agg] (p35) at (2,-2) {};

    \draw [myobject, agg] \convexpath {p12,p4} {11.5mm};
    \draw [myobject, agg] \convexpath {p35,p4} {10mm};
    \draw [agg] \convexpath {p12,p4} {11.5mm};

    \node [mynode, agg] (v4) at (p4) {$\{v_4\}$};
    \node [mynode, agg] (v12) at (p12) {$\{v_1, v_2\}$};
    \node [mynode, agg] (v35) at (p35) {$\{v_3, v_5\}$};

    \node [circle, inner sep=0, minimum size=23mm] (av12) at (v12) {};
    \node [circle, inner sep=0, minimum size=20mm] (av35) at (v35) {};

    \draw [myedge, ->, line width=0.5mm] (v12) to [out=180,in=130,looseness=5] (av12);
    \draw [myedge, ->, line width=0.5mm] (v35) to [out=0,in=50,looseness=5] (av35); 
    \draw [myedge, ->, line width=0.5mm] (av12) to [out=-60,in=-110,looseness=0.7] (v35);
    \draw [myedge, ->, line width=0.5mm] (av35) to [out=-135,in=-45,looseness=0.7] (v12);
    \draw [myedge, ->, line width=0.5mm] (v4) to [out=-110,in=-70,looseness=5] (v4);


    



  \end{scope}

\end{tikzpicture}

  \caption{
    Two partitioning schemes that one might consider to define the solution space of the GCP.
    Multiedge subsets are represented with hatching when they are not ``compact tiles''.
  }
  \label{fig:schemes}
\end{figure}

\subsection{From Vertex to Edge Partitions}
\label{ssec:powergraph}

By providing a measure of information loss, previous subsection focuses on the \emph{objective function} of the GCP, that is on the quality measure to be optimised.
We now focus on the \emph{solution space} of this problem, that is the set of partitions that one actually consider for compression.
The original GCP presented in Section~\ref{sec:GCP} consists in using a vertex partition $\calV \in \frakP(\V)$ to then determine a multiedge ``grid'' partition (see top-left matrix of Figure~\ref{fig:schemes}):
$$\calVxV \quad = \quad \{ \VxV \; :\;  V \in \calV \; \wedge \; V^\pm \in \calV \} \quad \in \quad \frakP(\bVxV) \text{,}$$
such that two multiedges $\vv$ and $\uu$ are in the same multiedge subset $\VxV$ if and only if their source vertices are both in $V$ and their target vertices are both in $V^\pm$:
$$\calVxV \vv = \calVxV \uu \qquad \Leftrightarrow \qquad \calV(v) = \calV(u) \quad \text{and} \quad \calV(v^\pm) = \calV(u^\pm) \text{.}$$
In other terms, the induced two-dimensional partitioning of the multiedge set $\bVxV$ is the Cartesian product of a one-dimensional partitioning of the vertex set $\V$.
This first compression scheme is classically used in macroscopic graph models such as block models and community-based representations.
One of the reason is that the result of compression can still be represented as a graph, which vertices are the subsets of the vertex partition and which multiedges are the subsets of the multiedge ``grid'' partition (see top-right graph of Figure~\ref{fig:schemes})

However, this results in a quite constrained solution space for the optimisation problem:
Only a small number of partitions of $\bVxV$ are actually feasible.
One might instead consider a less constrained solution space by allowing a larger number of multiedge partitions to be used for compression, by for example considering partitions of the multiedge set $\bVxV$ that are made of Cartesian products of two vertex subsets $\VxV \subseteq \calP(\bVxV)$.
\begin{definition}[Cartesian Multiedge Partitions]~\\
  \label{def:cartesian_partition}
  Given a vertex set $\V$, the set of \emph{Cartesian multiedge partitions} $\frakP^{\cn{\times}} (\bVxV) \subset \frakP(\bVxV)$ is the set of multiedge partitions that are made of Cartesian products of vertex subsets\footnotemark:
  \begin{align*}
    \frakP^{\cn{\times}} (\bVxV) \quad = \quad \{ & \{(V_1 \cn{\times} V^\pm_1),\ldots,(V_m \cn{\times} V^\pm_m)\} \subseteq \calP(\V){\times}\calP(\V)\\
                                                  & : \; \cup_i (V_i \cn{\times} V^\pm_i) = \bVxV \; \wedge \; \forall i \neq j, \; (V_i \cn{\times} V^\pm_i) \cap (V_j \cn{\times} V^\pm_j) = \emptyset\} \text{.}
  \end{align*}
\end{definition}
\footnotetext{
  There is an abuse of notation in this definition when writing that ``$\VxV \in \calP(\V){\times}\calP(\V)$''.
  We should instead write that ``$\VxV \in \calP(\bVxV)$ such that $\VV \in \calP(\V){\times}\calP(\V)$'', but we prefer the former for notation conciseness.
}
Such Cartesian partitions of $\bVxV$ hence contains ``rectangular'' multiedge subsets\footnote{
  Note that each of these multiedge subsets is indeed a rectangle in the adjacency matrix, modulo a reordering of its rows and columns.
  It might happen however that no reordering can make all multiedge subsets look rectangular \emph{at the same time}.
}
(see bottom-left matrix of Figure~\ref{fig:schemes}).
Hence, although the result of compression can no longer be represented as a graph, it can be represented as a \emph{directed hypergraph} -- or as a \emph{power graph}~\cite{Dwyer13, Ahnert14} -- that is a graph where edges can join couples of vertex subsets instead of couples of vertices (see bottom-right graph of Figure~\ref{fig:schemes}).

\paragraph{Related Work}

Most classical approaches for graph compression are based on the discovery of \emph{vertex partitions}.
This is for example the case of most work on graph summarisation~\cite{Zhang10, LeFevre10}, block modelling~\cite{Borgatti92}, community detection~\cite{Fortunato10}, and modular decomposition~\cite{Serafino13}.
As previously said, one of the interests of such vertex partitioning is that the result of compression is still a graph (a set of compressed vertices connected by a set of compressed edges) that can hence be represented, analysed, and visualised with traditional tools~\cite{Navlakha08, Toivonen11}.
However, the number and diversity of feasible solutions -- that is the set of vertex partitions -- is quite limited when compated to the full space of matrix partitions.

As proposed above, some approaches hence focus on \emph{edge partitions}, instead of vertex partitions, in order to provide more expressive compression schemes, but only on \emph{particular} edge partitions, in order to preserve some fundamental structures of graph data.
The most generic framework is formalised by \emph{power graph analysis}~\cite{Dwyer13, Ahnert14} and \emph{Mondrian processes}~\cite{Roy09}, where edge subsets are only required to be the Cartesian product of vertex subsets.
Such edge subsets can hence be represented as compressed edges between couples of compressed vertices.
But since different edge subsets can lead to overlapping vertex subsets, the resulting compressed graph is no longer a graph, but an \emph{hypergraph} (or a \emph{power graph}).
While still interpretable within the broad realm of graph theory, power graphs are more expressive~\cite{Toivonen11} than classical approaches -- such as community partitions -- since a given vertex might belong to different similarity groups:
It might be similar to a given group of vertices with respect to a given part of the graph, and similar to another group of vertices with respect to another part (see examples in~\cite{Ahnert14}).


\subsection{Adding Constraints to the Set of Feasible Vertex Subsets}
\label{ssec:constrained}

A wide variety of approaches span from generic edge partitions to more constrained schemes that ``impose restrictions on the nature of overlap'' for example by ``fixing the topology of connectivity between overlapping nodes''~\cite{Ahnert14} such as in clique percolation, spin models, mixed-membership block models, latent attribute models, spectral clustering, and link communities.
This aims at expressing particular topological properties within the compression scheme in order to produce meaningful graphs with respect to some particular analysis framework.
``Frequent patterns possibly reflect some semantic structures of the domain and therefore are useful candidates for replacement''~\cite{Zhou09}.

In the partitioning scheme hereabove presented, the Cartesian product of \emph{any} two vertex subsets $\VV \in \calP(\V){\times}\calP(\V)$ gives a feasible multiedge subset $\VxV \in \calP(\bVxV)$.
In other terms, the solution space is only limited by this Cartesian principle and by the partitioning constraints (covering and non-overlapping subsets).
However, one can also control the shape of feasible partitions by directly specifying the feasible vertex subsets that can be combined to build the ``rectangular'' multiedge subsets of the Cartesian partition.
Such additional constraints might prove to be useful to incorporate domain knowledge about particular vertex structures within the compression scheme:
By driving the process, feasibility constraints might indeed facilitate the interpretation of the compression's result according to the expert domain.
\begin{definition}[Feasible Multiedge Partitions]~\\
  \label{def:feasible_partition}
  Given a set of feasible vertex subsets $\hatcalP(\V) \subseteq \calP(\V)$, the set of \emph{feasible Cartesian multiedge partitions} $\cn{\hatfrakP^\times} (\bVxV) \subseteq \frakP^\times (\bVxV)$ is the set of Cartesian multiedge partitions which multiedge subsets are only made of feasible vertex subsets:
\begin{align*}
  \cn{\hatfrakP^\times} (\bVxV) \; = \; \{ & \{(V_1{\times}V^\pm_1),\ldots,(V_m{\times}V^\pm_m)\} \in \cn{\hatcalP}(\V){\times}\cn{\hatcalP}(\V)\\
                                                   & : \; \cup_i (V_i{\times}V^\pm_i) = \bVxV \; \wedge \; \forall i \neq j, \; (V_i{\times}V^\pm_i) \cap (V_j{\times}V^\pm_j) = \emptyset\} \text{.}
\end{align*}
\end{definition}
In the most general case, the set of feasible vertex subsets $\hatcalP(\V)$ can be composed of all vertex subset in $\calP(\V)$.
However, one can use more constrained -- and hence more meaningful vertex structures -- to drive the compression scheme and its possible applications.
A survey about the types of constraints that have been used in many different domains can be found in~\cite{RLP_MPI14}.
We present below some of these structures and their basic combinatorics.

\paragraph{The Complete Case}

When no additional constraint applies to vertex subsets, then:
$$\hatcalP(\V) \; = \; \calP(\V) \text{,} \quad \text{and\ so} \quad \hatfrakP(\V) = \frakP(\V) \text{.}$$
Such scheme can be used when no additional useful structure is known about the vertex set, for example to model \emph{coalition structures} in multi-agent systems when assuming that every possible group of agents is an adequate candidate to constitute a coalition~\cite{Sandholm99,Rahwan08}. 
It has been shown that, in this case, the number of feasible partitions grows considerably faster than the number of feasible subsets~\cite{Sandholm02}:
$$|\calP(\V)| = 2^{n} \quad \text{and} \quad |\frakP(\V)| = \omega(n^{n/2}) \text{,}$$ 
where $n = |\V|$ is the number of vertices in the graph.

\paragraph{Hierarchies of Vertex Subsets}

A \emph{hierarchy} $\hatcalP(\V) = \calH(\V)$ is such that any two feasible vertex subsets are either disjoint, or one is included in the other:
$$\forall (V_1,V_2) \in \calH(\V)^2, \quad V_1 \cap V_2 = \emptyset \; \vee \; V_1 \subseteq V_2 \; \vee \; V_2 \subseteq V_1 \text{.}$$
We mark $\hatfrakP(\V) = \frakH(\V)$ the resulting set of feasible vertex partitions.
Such vertex hierarchies can be used to model graphs that are known to have a multilevel nested structure that one wants to preserve during compression.
This might include, among others, a \emph{community structure} that have been preliminarily identified through a hierarchical community detection algorithm~\cite{Pons11}, a sequence of \emph{geographical nested partitions} of the world's territorial units~\cite{RLP_TCCI14};
a \emph{hierarchical communication network} in distributed computers~\cite{RLP_CLUSTER14}.

It has been shown that the number of feasible subsets in a hierarchy is asymptotically bounded from above by the number of objects it contains~\cite{RLP_MPI14}:
$|\calH(\V)| = \mathcal{O}(n)$.
The resulting number of feasible partitions however depends on the number of levels and branches in the hierarchy.
For a complete binary tree, it is asymptotically bounded by an exponential function: $|\frakH(\V)| = \Theta(\alpha^n)$, with $\alpha \approx 1.226$~\cite{Cloitre02}.
Similar results have been found for complete ternary trees (with $\alpha \approx 1.084$~\cite{McGarvey07}), complete quaternary trees, and so on.
Henceforth, for any bounded number of children per node in the hierarchy, 
the number of feasible partitions {exponentially} grows with the number of objects. 

\paragraph{Sets of Vertex Intervals}

Any total order $<$ on the vertex set $\V$ induces a \emph{set of vertex intervals} $\hatcalP(\V) = \calI(\V)$ defined as follows:
$$\calI(\V) \; = \; \{ [v_1,v_2] \; : \; (v_1,v_2) \in \V^2 \} \quad \text{where} \quad [v_1,v_2] \; = \; \{v \in \V : v_1 \leq v \leq v_2\} \text{.}$$
We mark $\hatfrakP(\V) = \frakI(\V)$ the resulting set of feasible vertex partitions.
It can be represented as a ``pyramid of intervals''~\cite{RLP_MPI14} and such partitions are sometimes called \emph{consecutive partitions}~\cite{Anily91}.
Such sets of intervals naturally apply to vertices having a temporal feature (\emph{e.g.}, events, dates, or time periods are naturally ordered by the ``arrow of time'') and have hence been exploited for the aggregation of time series~\cite{Jackson04,Pagano13,RLP_TCCI14}.
They might also model unidimensional spacial features, such as the geographical ordering of cities on a coast, or on a transport route~\cite{Rothkopf98}.

The number of intervals of an ordered set of size $n$ is $|\calI(\V)| = \frac{n(n+1)}{2}$.
The resulting number of feasible partitions is $|\frakI(\V)| = 2^{n-1}$~\cite{RLP_MPI14}.

\subsection{From Multigraphs to Multistreams}
\label{ssec:dynamic}

The last step of our generalisation work regards the integration of a temporal dimension within our framework in order to finally deal with the compression of temporal graphs~\cite{Holme13}.
As argued in the introduction of this article, we build on the \emph{link stream} representation of such graphs~\cite{Viard16, Latapy17}.
However, because we also want to deal with cases for which multiple edges are allowed between two vertices at a given time instance, we actually generalise from streams to multistreams, as we did in Subsection~\ref{ssec:unweighted} to generalise from graphs to multigraphs.
\begin{definition}[Directed Multistream]~\\
  A \emph{directed multistream} $\MS = (\V,\T,e)$ is characterised by:
  \begin{itemize}
  \item A set of \emph{vertices} $\V$;
  \item A set of \emph{time instances} $\T \subseteq \mathbb{R}$; 
  \item A multiset of \emph{directed edges} $(\bVxV{\times}\T, e)$\\
    where $e : \bVxV{\times}\T \rightarrow \mathbb{N}$ is the \emph{edge function}, that is the multiplicity function counting the number of edges $e(v,v^\pm,t) \in \mathbb{N}$ going from a given source vertex $v \in \V$ to a given target vertex $v^\pm \in \V$ at a given time instance $t \in \T$.
  \end{itemize}
\end{definition}
In this context, the edge function counts the number of interactions happening between vertices at a given time:
$\ef(v,v^\pm,t)$ is the number of edges going from source vertex $v$ to target vertex $v^\pm$ at time $t$.\\

\begin{figure}
  \centering
  \mycenter{\begin{tikzpicture} [scale=0.8, cell/.style={black!25!red, very thick}]

  \node [align=center] at ($0.5*(0,0,4.5) + 0.25*(0,2,4.5)$) {$v_1$};
  \node [align=center] at ($0.5*(0,0,4.5) + 0.25*(0,4,4.5)$) {$v_2$};
  \node [align=center] at ($0.5*(0,0,4.5) + 0.25*(0,6,4.5)$) {$v_3$};
  \node [align=center] at ($0.5*(0,0,4.5) + 0.25*(0,8,4.5)$) {$v_4$};
  \node [align=center] at ($0.5*(0,0,4.5) + 0.25*(0,10,4.5)$) {$v_5$};
  \node [align=center, rotate=45] at ($0.5*(0,5,0) + 0.25*(0,0,5)$) {$v_1  v_2  v_3  v_4  v_5$\\[1ex]};

  \begin{scope}[scale=0.5]
    \draw [cell] (0,0,0) -- (21,0,0) -- (21,5,0) -- (0,5,0) -- cycle;
    \draw [cell] (0,0,0) -- (21,0,0) -- (21,0,5) -- (0,0,5) -- cycle;
    \draw [cell] (0,0,0) -- (0,0,5) -- (0,5,5) -- (0,5,0) -- cycle;

    \drawblackcell{0}{0}{2}{2}{3}{3}{10};
    \drawblackcell{2}{0}{3}{19}{2}{2}{90};
    \drawblackcell{19}{0}{0}{2}{2}{5}{30};

    \drawblackcell{2}{2}{3}{3}{1}{2}{50};
    \drawblackcell{5}{2}{3}{13}{1}{2}{80};

    \drawblackcell{0}{3}{0}{9}{2}{2}{100};
    \drawblackcell{0}{3}{2}{3}{2}{3}{80};

    \drawblackcell{9}{3}{0}{9}{2}{2}{20};
    \drawblackcell{3}{3}{2}{6}{2}{1}{40};

    \drawblackcell{9}{3}{2}{9}{2}{1}{70};
    \drawblackcell{3}{3}{3}{15}{2}{2}{10};
    
    \drawblackcell{18}{2}{0}{3}{3}{3}{5};
    \drawblackcell{18}{2}{3}{3}{3}{2}{70};

    \draw [cell] (21,5,5) -- (0,5,5) -- (0,0,5) -- (21,0,5) -- cycle;
    \draw [cell] (21,5,5) -- (0,5,5) -- (0,5,0) -- (21,5,0) -- cycle;
    \draw [cell] (21,5,5) -- (21,5,0) -- (21,0,0) -- (21,0,5) -- cycle;
  \end{scope}
    
\end{tikzpicture}

}
  \caption{Multistream compression, that is the Cartesian partitioning of $\bVxVxT$, consists in the natural tridimensional generalisation of multigraph compression, that is the Cartesian partitioning of $\bVxV$.}
  \label{fig:stream_compression}
\end{figure}
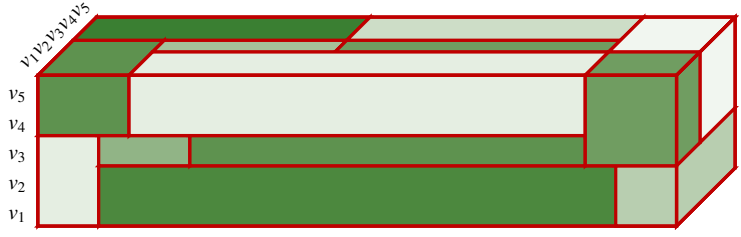

As illustrated in Figure~\ref{fig:stream_compression}, this formalism constitutes an elegant solution to generalise our compression scheme since it simply adds a third dimension $\T \subseteq \mathbb{R}$ to the multigraph's definition.
The GCP is then generalised to this three-dimensional formalism by (i)~defining a set of feasible time subsets that preserves the ordering of time instances, that is a set of intervals $\hatcalP(\T) = \calI(\T)$ (see previous subsection), (ii)~considering three-dimensional Cartesian multiedge subsets $\VxVxT \in \hatcalP(\V){\times}\hatcalP(\V){\times}\calI(\T)$, and (iii)~computing the information loss on this generalised space in a similar fashion than for the static version of the compression problem.
Here is the resulting generalisation in details.
\[
  \rotatebox[origin=c]{90}{\text{see Definition~\ref{def:multigraph}}} \left\{\parbox{0.95\textwidth}{%
      \begin{itemize}[leftmargin=1.5em]

      \item \makebox[3cm][l]{Time instances:} $\c{t \in \T \subseteq \mathbb{R}}$;

      \item \makebox[3cm][l]{Time subsets:} $\c{T \in \calP(\T)}$; 

      \item \makebox[3cm][l]{Multiedges:} $\cvvt \in \cbVxVxT$;

      \item \makebox[3cm][l]{Edge function:} $e : \cbVxVxT \rightarrow \mathbb{N}$;

      \item Compressed edge function:
      \end{itemize}
      $$\ef : \calP(\V){\times}\calP(\V)\c{{\times}\calP(\T)} \rightarrow \mathbb{N} \quad \text{with} \quad \ef (V,V^\pm \c{,T}) = \sum_{\vv \in \VxV} \c{\int_{t \in T}} {\ef \cvvt \; \c{dt}} \text{;}$$
    }\right.
\]

\vspace{-1.5em}

\[
  \rotatebox[origin=c]{90}{\text{see Def.~\ref{def:feasible_partition}}} \left\{\parbox{0.95\textwidth}{%
      \begin{itemize}[leftmargin=1.5em]
      \item Feasible time subsets, that is time intervals: $\c{T = [t_1,t_2] \; \in \; \hatcalP(\T) = \calI(\T) \; \subset \; \calP(\T)}$;

      \item Feasible multiedge subsets: $\cVxVxT \; \in \; \hatcalP(\V){\times}\hatcalP(\V)\c{{\times}\calI(\T)} \; \subset \; \calP(\cbVxVxT)$;
        
      \end{itemize}
    }\right.
\]

\vspace{-1.5em}

\[
  \rotatebox[origin=c]{90}{\text{see Def.~\ref{def:cartesian_partition}}} \left\{\parbox{0.95\textwidth}{%
      \begin{itemize}[leftmargin=1.5em]
      \item Feasible Cartesian multiedge partitions:
      \end{itemize}\vspace{-1.5em}
      \begin{align*}
        \ccalVVT \in \hatfrakP^\times (\bVxV\c{{\times}\T}) = \{ & \{(V_1{\times}V^\pm_1\c{{\times}T_1}),\ldots,(V_m{\times}V^\pm_m\c{{\times}T_m})\} \in \hatcalP(\V){\times}\hatcalP(\V)\c{{\times}\calI(\T)}\\
                                                                             & : \cup_i (V_i{\times}V^\pm_i\c{{\times}T_i}) = \bVxV\c{{\times}\T} \wedge (V_i{\times}V^\pm_i\c{{\times}T_i}) \cap (V_j{\times}V^\pm_j\c{{\times}T_j}) = \emptyset \} \text{;}
      \end{align*}
    }\right.
\]

\vspace{-1.5em}

\[
  \rotatebox[origin=c]{90}{\text{see Definition~\ref{def:empirical}}} \left\{\parbox{0.95\textwidth}{%
      \begin{itemize}[leftmargin=1.5em]
      \item Observed variable:
        $$\cXXX \in \bVxV\c{{\times}\T} \quad \text{with} \quad \c{f}_\cXXX \cvvt = \frac{\ef\cvvt }{\ef (\V,\V, \c{\T}) } \text{;}$$
      \end{itemize}
      }\right.
\]

\vspace{-1.5em}

\[
  \rotatebox[origin=c]{90}{\text{see Definition~\ref{def:compressed}}} \left\{\parbox{0.95\textwidth}{%
      \begin{itemize}[leftmargin=1.5em]
      \item Partition function:
        $$\ccalVVT(v,v^\pm\c{,t}) \text{ is the only multiedge subset } \cVxVxT \text{ in } \ccalVVT \text{ that contains } (v,v^\pm\c{,t}) \text{;}$$

      \item Compressed variable:
        $$\ccalVVT\cXXX \in \ccalVVT \quad \text{with} \quad \c{f}_{\ccalVVT\cXXX} (V,V^\pm \c{,T}) = \frac{\ef (V,V^\pm \c{,T}) }{\ef (\V,\V \c{,\T}) } \text{;}$$
      \end{itemize}
    }\right.
\]

\vspace{-1.5em}

\[
  \rotatebox[origin=c]{90}{\text{see Definition~\ref{def:decompressed}}} \left\{\parbox{0.95\textwidth}{%
      \begin{itemize}[leftmargin=1.5em]
      \item External variable, in the case of a degree-preserving compression:
        $$\cYYY \in \ccalVVT  \quad \text{with} \quad \c{f}_\cYYY (v,v^\pm\c{,t}) = \displaystyle\frac{\ef(v,\V\c{,\T})}{\ef(\V,\V\c{,\T})} \frac{\ef(\V,v^\pm\c{,\T})}{\ef(\V,\V\c{,\T})} \c{\frac{\ef(\V,\V,t)}{\ef(\V,\V,\T)}} \text{;}$$

      \item Decompressed variable:
      \end{itemize}\vspace{-1.5em}
      \begin{multline*}
        \ccalVVT_\cYYY \cXXX \in \bVxV\c{{\times}\T}\\
        \text{with} \quad \c{f}_{\ccalVVT_\cYYY\cXXX} \cvvt = \displaystyle\frac {\ef\cVVT} {\ef(\V,\V\c{,\T})} \displaystyle\frac {\ef(v,\V\c{,\T}) \; \ef(\V,v^\pm\c{,\T}) \; \c{\ef(\V,\V,t)}} {\ef(V,\V\c{,\T}) \; \ef(\V,V^\pm\c{,\T}) \; \c{\ef(\V,\V\c{,T})}}\\
        \text{where\ } \cVxVxT = \ccalVVT\cvvt \text{;}
      \end{multline*}
    }\right.
\]

\vspace{-1.5em}

\[
  \rotatebox[origin=c]{90}{\text{see Definition~\ref{def:information_loss}}} \left\{\parbox{0.95\textwidth}{%
      \begin{itemize}[leftmargin=1.5em]
      \item Information loss, decomposed at the level of multiedge subsets:
      \end{itemize}
      $$\lossf{\cXXX}{\pcVVT}{\cYYY} = \hspace{-3ex} \sum_{\vv \in \VxV} \c{\int_{t \in T}} {\frac{\ef \cvvt}{\ef (\V,\V \c{,\T}) } \log_2 \left( \frac {\ef(v,v^\pm\c{,t})} {\ef(V,V^\pm\c{,T})} \bigg/ \frac {\ef(v,\V\c{,\T}) \; \ef(\V,v^\pm\c{,\T}) \; \c{\ef(\V,\V,t)}} {\ef(V,\V\c{,\T}) \; \ef(\V,V^\pm\c{,\T}) \; \c{\ef(\V,\V,T)}} \right) \c{dt}} \text{.}$$
    }\right.
\]

The use of \emph{integrals} instead of discrete sums to define the edge function, as well as the use of probability \emph{density function} instead of discrete probability distribution to define the random variables, is due to the fact that the added temporal dimension $\T$ is continuous, contrary to the vertex set $\V$.
A simpler setting for temporal graphs would consist in assuming a discrete representation of time, that would hence only require discrete operators.
In this regard, the optimisation algorithm that is introduced in next section to solve the generalised GCP is only provided for the discrete case, for simplicity reasons.

\paragraph{Related Work}

As discussed in the introduction of this article, and as illustrated in this subsection, the recent work on the \emph{link stream} formalism~\cite{Viard16, Latapy17} proposes to deal with time as a simple addition to the graph's structural dimension.
Considering temporal graphs as genuine tridimensional data, the arbitrary separation of structure and time is avoided, thus making the generalisation quite natural.
Note that similar generalisation objectives have been addressed in previous work on graph compression, as for example the application of bidimensional \emph{block models} to multidimensional matrices~\cite{Borgatti92} or the application of \emph{biclustering} to triplets of variables~\cite{Narmadha16}, which has then been exploited for the statistical analysis of temporal graphs~\cite{Guigoures12}.
The particular interest of such approaches also consists in the fact that they provide a unified compression scheme in which structural and temporal information is simultaneously taken into account.


\section{Result: The Lossy Multistream Compression Problem}
\label{sec:MSCP}

This section integrates all generalisations that have been proposed in previous section to the GCP in order to properly formalise the Lossy Multistream Compression Problem (MSCP).
It then proposes an algorithmic solution to this problem by reducing it to the \emph{Set Partitioning Problem}, a well-known combinatorial optimisation problem allowing to exploit state-of-the-art approaches.

\subsection{The Lossy MSCP}

We first formalise the Lossy MSCP by exploiting all notions introduced in previous section.
Note that this instance of the problem has a discrete temporal dimension (for an easier algorithmic solution) and is described for a particular external variable (corresponding to a degree-preserving decompression, see Subsection~\ref{ssec:lossy}).

\begin{definition}[The Lossy Multistream Compression Problem]~\\
  Given:
  \begin{itemize}
  \item A directed multistream $\MS = (\V,\T,e)$,
    where $\T = \llbracket t_\alpha, t_\omega \rrbracket \subseteq \mathbb{N}$ and where the edge function $e : \bVxVxT \rightarrow \mathbb{N}$ determines the observed variable $\XXX \in \bVxVxT$; 
  \item A set of feasible vertex subsets $\hatcalP(\V) \subseteq \calP(\V)$, the set of time intervals $\calI(\T) \subset \calP(\T)$, and the set of feasible Cartesian multiedge partitions $\hatfrakP^\times (\bVxVxT)$;
  \item The external variable $\YYY \in \bVxVxT$ corresponding to a degree-preserving decompression, which determines the information function $\loss : \hatfrakP^\times(\bVxVxT) \rightarrow \mathbb{R}^+$, and an information threshold $\threshold \in \mathbb{R}^+$;
  \end{itemize}
  Find a feasible Cartesian multiedge partition $\calVVT^* \in \hatfrakP^\times (\bVxVxT)$ with minimal size $|\calVVT^*|$ such that $\lossf{\XX}{\calVVT^*}{(\Xout,\Xin,\Utm)} \leq \threshold$:
  $$\calVVT^* \; = \; \argmin_{\substack{\calVVT \in \hatfrakP^\times (\bVxVxT) \\ \lossf{\XX}{\calVVT}{(\Xout,\Xin,\Utm)} \leq \threshold}} \; |\calVVT| \text{.}$$
\end{definition}

\paragraph{Lagrangian Relaxation}

The MSCP is a \emph{constrained} optimisation problem:
One wants to minimise the size of the multiedge partition (objective) while guarantying that the information loss it induces stays below a given threshold (constraint).
Lagrangian relaxation is a classical relaxation method which approximates the constrained problem by a simpler unconstrained one~\cite{Ahuja93}.
It consists in injecting the constraints within the objective, penalising eventual violations of these constraints by using a Lagrange multiplier $\lambda \in \mathbb{R}^+$ that is interpreted as the cost that such violations pose to the objective.
It hence results in an unconstrained optimisation problem which objective is parametrised by the Lagrange multiplier.
In the case of the MSCP, the Lagrange function that mixes the information-theoretic objective and the size constraint is
$$\obj_\lambda(\calVVT) \quad = \quad |\calVVT| \; + \; \lossf{\XX}{\calVVT}{(\Xout,\Xin,\Utm)} \text{,}$$
where $\lambda \in \mathbb{R}^+$ is the Lagrange multiplier.
The relaxed MSCP then simply consists in minimising this parametrised objective:
$$\calVVT^* \; = \; \argmin_{\calVVT \in \hatfrakP^\times (\bVxVxT)} \; \obj_\lambda(\calVVT) \text{.}$$
Intuitively, when $\lambda = 0$, the relaxed MSCP consists in minimising the partition size without any constraint regarding the information loss it induces.
This hence corresponds to a virtually infinite information threshold $\threshold$ for which the \emph{maximal partition} (all multiedges in one subset) is optimal.
When $\lambda \rightarrow \infty$, the objective becomes negligible in comparison to the cost of constraint violation, so the relaxed MSCP simply consists in minimising the information loss.
This hence corresponds to the highest level of constraint ($\tau = 0$) for which the \emph{minimal partition} (each multiedge in its own subset) is optimal.
By varying the Lagrange multiplier $\lambda$ between $0$ and $\infty$, one can then express different levels of constraints and obtain a multiscale sequence of compression results ranging from the minimal to the maximal partition.

\subsection{Reducing the Lossy MSCP to the Set Partitioning Problem}
\label{ssec:reduction}

The \emph{Set Partitioning Problem} (SPP) is a deeply-studied combinatorial optimisation problem that naturally arises as soon as one wants to organise a set of objects into covering and pairwise disjoint subsets such that an additive objective is minimised~\cite{Balas76}.
To that matter, the Lossy MSCP can be expressed as a particular case of the SPP, hence allowing the use of state-of-the-art algorithmic results to address it. 
\begin{definition}[The Set Partitioning Problem]~\\
  Given a set of objects $\pop$, a set of feasible subsets $\hatcalP(\pop) \subseteq \calP(\pop)$, and a cost function $c : \hatcalP(\pop) \rightarrow \mathbb{R}^+$, the \emph{Set Partitioning Problem} consists in finding a minimal feasible partition of $\pop$, that is a set of feasible subsets $\partitionpop \subseteq \hatcalP(\pop)$ that partitions $\pop$ and that minimises the sum of the costs:
  $$\partitionpop^* \quad = \quad \argmin_{\partitionpop \in \hatfrakP(\pop)} \; \sum_{\partpop \in \partitionpop} {c(\partpop)} \text{,}$$
  where $\hatfrakP(\pop) = \{ \{\partpop_1,\ldots,\partpop_m\} \subseteq \hatcalP(\pop) \; : \; \cup_i \partpop_i = \pop \; \wedge \; \forall i \neq j, \; \partpop_i \cap \partpop_j = \emptyset \}$.
\end{definition}
The Lossy MSCP can be expressed as special cases of the SPP:
\begin{itemize}
\item The objects are the triples consisting in two vertices and one time instance: $\pop \; = \; \bVxVxT$;
\item The feasible subsets are the Cartesian products of any two feasible vertex subsets and a time interval\footnote{
    There is again an abuse of notation in the following definition when we write that ``$\hatcalP (\bVxVxT) \; = \; \hatcalP(\V) \times \hatcalP(\V) \times \calI(\T)$''.
    In principle, we should write that ``$\hatcalP (\bVxVxT) \; = \; \{ \VxVxT : \VVT \in \hatcalP(\V) \times \hatcalP(\V) \times \calI (\T)\}$'', but we prefer the former for notation conciseness.
  }:
  $$\hatcalP (\bVxVxT) \; = \; \hatcalP(\V) \times \hatcalP(\V) \times \calI(\T) \text{;}$$
\item The cost function is given by the information loss function defined at the subset level:
  $$c : \hatcalP(\bVxVxT) \rightarrow \mathbb{R}^+ \quad \text{with} \quad c\VVT \; = \; q_\lambda\VVT \; = \; 1 \; +  \; \lambda \lossf{\XX}{\pVVT}{(\Xout,\Xin,\Utm)} \text{.}$$
\end{itemize}

\paragraph{Solving the Set Partitioning Problem}

The SPP is NP-complete in the {general} case~\cite{Chakravarty82}. 
Hence, one cannot hope for a general-purpose algorithm that can efficiently solve every instance of the SPP (except if P = NP).
However, it as been shown that some restrictive, yet meaningful constraints on the set of feasible subsets induce tractable versions of the SPP.
Hence, among the many algorithmic strategies that have been proposed in the literature to tackle the problem (see~\cite{RLP_MPI14} for an extensive review of such work), we are interested in this article in approaches that (i)~take advantage of the algebraic properties of partition lattices (see \ref{sec:spp}) to cleverly search the solution space for an optimal solution, and that moreover (ii)~take advantage of the structure of constraints exerting on the set of feasible subsets (see Subsection~\ref{ssec:constrained}) to reduce the size of the solution space and to speed up the search.

Regarding~(i), approaches based on dynamic programming have been proposed in the past to solve the SPP in its most general formulation, leading to a $\Omega(2^n)$ and $O(3^n)$ optimisation algorithm~\cite{Rothkopf98,Sandholm02}, where $n = |\pop|$ is the number of objects.
However, regarding~(ii), it has been shown that more efficient algorithms are possible when dealing with special versions of the problem, that is when the set of feasible subsets has a more specific structure that can also be exploited by dynamic programming. 
For example, when the set of feasible subsets forms a hierarchy $\calH(\pop)$ (see Subsection~\ref{ssec:constrained}), then an appropriate depth-first search of the tree representing the hierarchy allows for a $O(n)$ optimisation algorithm~\cite{Pons11,RLP_TCCI14}.
When feasible subsets forms a set of intervals $\calI(\pop)$, then solving the SPP is equivalent to solving the \emph{shortest path problem}~\cite{Chakravarty82} and several $O(n^2)$ dynamic programming algorithms have been proposed in the past for this particular case~\cite{Anily91,Vidal93,Rothkopf98,Jackson04}.

In \ref{sec:spp}, we present a general algorithmic framework that has been proposed in previous work to solve such special versions of the SPP. 
The key principle of this framework, and of dynamic programming in general, is the following:
(i)~The solution space is first broken down into smaller covering subspaces;
(ii)~Then, thanks to a \emph{principle of optimality} that fits with the algebraic structure of the partition set, these subproblems are recursively solved;
(iii)~Locally-optimal solutions are then compared to globally solve the initial problem.
More precisely, the principle of optimality builds on the following observation:
For each feasible subset $X \in \hatcalP(\pop)$, by first computing a feasible optimal partition $\calX^* \in \frakP(X)$ on this subset, one can use this optimal partition to efficiently evaluate all the partitions that are coarser than $X$, that is all partitions that contain a subset that contains $X$.
This allows to identify and suppress numerous redundant computations when searching the solution space (see detailed definitions and proofs in \ref{sec:spp}, in particular for the definition of refinements $\frakR$ and covering partitions $\frakC$).\\ 

\pagebreak

Here is the resulting high-level algorithm:

\begin{algorithm}[Recursive Algorithm to Solve the SPP~\cite{RLP_ICTAI14, RLP_MPI14}]~\\[1ex]
  \textbf{Global:}\vspace{-1ex}
  \begin{itemize}[itemsep=0pt]
  \item A set of objects $\pop$;
  \item A set of feasible subsets $\hatcalP(\pop) \subseteq \calP(\pop)$;
  \item A cost function $c : \hatcalP(\pop) \rightarrow \mathbb{R}^+$.
  \end{itemize}
  \textbf{Input:} A feasible subset $X \in \hatcalP(\pop)$.\\[1ex]
  \textbf{Output:} A locally-optimal feasible partition $\calX^* \in \hatfrakP^*(X)$.\\[1ex]
  \textbf{Algorithm:}\vspace{-1ex}
  \begin{itemize}[itemindent=0.5cm]
  \item[(step 1)] If the algorithm has already been called on $X$, return the recorded locally-optimal partition and its cost.
  \item[(step 2)] Compute the set $\hatfrakC(\{X\})$ of feasible partitions that are covered by $\{X\}$.
  \item[(step 3)] For each such partition $\calX \in \hatfrakC(\{X\})$, do the following:
	\begin{itemize}[itemsep=0.05cm,topsep=0.05cm,itemindent=0.7cm]
    \item[(step 3.a)] For each subset ${X^\pm} \in \calX$, recursively compute a locally-optimal feasible partition $\calX_{X^\pm}^* \in \hatfrakP^*({X^\pm})$;
    \item[(step 3.b)] Compute the union $\calX^* = \bigcup_{{X^\pm} \in \calX} {\calX_{X^\pm}^*}$ of these partitions.
      The principle of optimality ensures that $\calX^* \in \hatfrakR^*(\calX)$.
	\end{itemize}
  \item[(step 4)] Record and return a partition that minimises the cost function $c$ among the maximal partition $\{X\}$ and all the partitions $\calX^*$ that have been computed. 
  \end{itemize}
  \label{algo:generic_algo}
\end{algorithm}

\subsection{Solving the Lossy MSCP}

We now precise and apply the high-level SPP algorithm introduced in previous subsection to solve the Lossy MSCP, by giving more details about its specialisation.
In particular, the algorithm above does neither precise how the set of feasible subsets $\hatcalP(\bVxVxT)$ should be represented in memory, nor how the computation of feasible covered partitions $\hatfrakC(\{\VxVxT\})$ and the computation of the cost function $c(\VxVxT)$ should be implemented in practice.
In what follows, we provide an efficient data structure to do so and a detailed implementation of the resulting combinatorial optimisation algorithm.

\begin{figure}[b!]
  \centering
  \begin{minipage}{0.35\textwidth}
    \centering
    \begin{tikzpicture}[scale=0.5]

  \coordinate (pos) at (0,0);
  \tikzpart{1234}{p1234}

  \coordinate (pos) at (-1.25,-2);
  \tikzpart{123}{p123}

  \coordinate (pos) at (2.5,-4);
  \tikzpart{34}{p34}

  \coordinate (pos) at (-3.75,-6);
  \tikzpart{1}{p1}

  \coordinate (pos) at (-1.25,-6);
  \tikzpart{2}{p2}

  \coordinate (pos) at (1.25,-6);
  \tikzpart{3}{p3}

  \coordinate (pos) at (3.75,-6);
  \tikzpart{4}{p4}

  \node [draw, circle, thick, minimum width=2pt, inner sep=0pt] (l1234i1) at ($(p1234i1) - (0,0.75)$) {};
  \draw [toInLink] (p1234i1) -- (l1234i1);
  \draw [toOutLink] (l1234i1) -- (p123);
  \draw [toOutLink] (l1234i1) to [bend left=45] (p4);

  \node [draw, circle, thick, minimum width=2pt, inner sep=0pt] (l1234i4) at ($(p1234i4) - (0,2.25)$) {};
  \draw [toInLink] (p1234i4) -- (l1234i4);
  \draw [toOutLink] (l1234i4) to [bend right=20] (p1);
  \draw [toOutLink] (l1234i4) to [bend right=20] (p2);
  \draw [toOutLink] (l1234i4) -- (p34);

  \node [draw, circle, thick, minimum width=2pt, inner sep=0pt] (l123) at ($(p123) - (0,1.5)$) {};
  \draw [toInLink] (p123) -- (l123);
  \draw [toOutLink] (l123) -- (p1);
  \draw [toOutLink] (l123) -- (p2);
  \draw [toOutLink] (l123) -- (p3);
  
  \node [draw, circle, thick, minimum width=2pt, inner sep=0pt] (l34) at ($(p34p1) - (0,0.75)$) {};
  \draw [toInLink] (p34p1) -- (l34);
  \draw [toOutLink] (l34) -- (p3);
  \draw [toOutLink] (l34) -- (p4);
\end{tikzpicture}

    \vskip2ex
    a.~Arbitrary set
  \end{minipage} \begin{minipage}{0.43\textwidth}
    \centering
    \begin{tikzpicture}[scale=0.5]

  \coordinate (pos) at (0,0);
  \tikzpart{1234}{p1234}

  \coordinate (pos) at (-1.5,-2);
  \tikzpart{123}{p123}

  \coordinate (pos) at (1.5,-2);
  \tikzpart{234}{p234}

  \coordinate (pos) at (-2.5,-4);
  \tikzpart{12}{p12}

  \coordinate (pos) at (0,-4);
  \tikzpart{23}{p23}

  \coordinate (pos) at (2.5,-4);
  \tikzpart{34}{p34}

  \coordinate (pos) at (-3.75,-6);
  \tikzpart{1}{p1}

  \coordinate (pos) at (-1.25,-6);
  \tikzpart{2}{p2}

  \coordinate (pos) at (1.25,-6);
  \tikzpart{3}{p3}

  \coordinate (pos) at (3.75,-6);
  \tikzpart{4}{p4}

  \node [draw, circle, thick, minimum width=2pt, inner sep=0pt] (l1234i1) at ($(p1234i1) - (0,0.75)$) {};
  \draw [toInLink] (p1234i1) -- (l1234i1);
  \draw [toOutLink] (l1234i1) -- (p234);
  \draw [toOutLink] (l1234i1) to [bend right=45] (p1);

  \node [draw, circle, thick, minimum width=2pt, inner sep=0pt] (l1234) at ($(p1234) - (0,2)$) {};
  \draw [toInLink] (p1234) -- (l1234);
  \draw [toOutLink] (l1234) to [bend left=10] (p12);
  \draw [toOutLink] (l1234) to [bend right=10] (p34);

  \node [draw, circle, thick, minimum width=2pt, inner sep=0pt] (l1234i4) at ($(p1234i4) - (0,0.75)$) {};
  \draw [toInLink] (p1234i4) -- (l1234i4);
  \draw [toOutLink] (l1234i4) -- (p123);
  \draw [toOutLink] (l1234i4) to [bend left=45] (p4);

  \node [draw, circle, thick, minimum width=2pt, inner sep=0pt] (l123i1) at ($(p123i1) - (0,0.75)$) {};
  \draw [toInLink] (p123i1) -- (l123i1);
  \draw [toOutLink] (l123i1) to [bend right=45] (p1);
  \draw [toOutLink] (l123i1) -- (p23);
  
  \node [draw, circle, thick, minimum width=2pt, inner sep=0pt] (l123i3) at ($(p123i3) - (0,0.75)$) {};
  \draw [toInLink] (p123i3) -- (l123i3);
  \draw [toOutLink] (l123i3) -- (p12);
  \draw [toOutLink] (l123i3) to [bend left=45] (p3);

  \node [draw, circle, thick, minimum width=2pt, inner sep=0pt] (l234i1) at ($(p234i1) - (0,0.75)$) {};
  \draw [toInLink] (p234i1) -- (l234i1);
  \draw [toOutLink] (l234i1) to [bend right=45] (p2);
  \draw [toOutLink] (l234i1) -- (p34);
  
  \node [draw, circle, thick, minimum width=2pt, inner sep=0pt] (l234i3) at ($(p234i3) - (0,0.75)$) {};
  \draw [toInLink] (p234i3) -- (l234i3);
  \draw [toOutLink] (l234i3) -- (p23);
  \draw [toOutLink] (l234i3) to [bend left=45] (p4);

  \node [draw, circle, thick, minimum width=2pt, inner sep=0pt] (l12) at ($(p12p1) - (0,0.75)$) {};
  \draw [toInLink] (p12p1) -- (l12);
  \draw [toOutLink] (l12) -- (p1);
  \draw [toOutLink] (l12) -- (p2);

  \node [draw, circle, thick, minimum width=2pt, inner sep=0pt] (l23) at ($(p23p1) - (0,0.75)$) {};
  \draw [toInLink] (p23p1) -- (l23);
  \draw [toOutLink] (l23) -- (p2);
  \draw [toOutLink] (l23) -- (p3);

  \node [draw, circle, thick, minimum width=2pt, inner sep=0pt] (l34) at ($(p34p1) - (0,0.75)$) {};
  \draw [toInLink] (p34p1) -- (l34);
  \draw [toOutLink] (l34) -- (p3);
  \draw [toOutLink] (l34) -- (p4);
\end{tikzpicture}

    \vskip2ex
    b.~Set of intervals
  \end{minipage} \begin{minipage}{0.20\textwidth}
    \centering
    \begin{tikzpicture}[scale=0.5]

  \coordinate (pos) at (0,0);
  \tikzpart{1234}{p1234}

  \coordinate (pos) at (-1.5,-4);
  \tikzpart{12}{p12}

  \coordinate (pos) at (1.5,-4);
  \tikzpart{34}{p34}

  \coordinate (pos) at (-2.25,-6);
  \tikzpart{1}{p1}

  \coordinate (pos) at (-0.75,-6);
  \tikzpart{2}{p2}

  \coordinate (pos) at (0.75,-6);
  \tikzpart{3}{p3}

  \coordinate (pos) at (2.25,-6);
  \tikzpart{4}{p4}

  \node [draw, circle, thick, minimum width=2pt, inner sep=0pt] (l1234) at ($(p1234) - (0,1.5)$) {};
  \draw [toInLink] (p1234) -- (l1234);
  \draw [toOutLink] (l1234) -- (p12);
  \draw [toOutLink] (l1234) -- (p34);

  \node [draw, circle, thick, minimum width=2pt, inner sep=0pt] (l12) at ($(p12p1) - (0,0.75)$) {};
  \draw [toInLink] (p12p1) -- (l12);
  \draw [toOutLink] (l12) -- (p1);
  \draw [toOutLink] (l12) -- (p2);

  \node [draw, circle, thick, minimum width=2pt, inner sep=0pt] (l34) at ($(p34p1) - (0,0.75)$) {};
  \draw [toInLink] (p34p1) -- (l34);
  \draw [toOutLink] (l34) -- (p3);
  \draw [toOutLink] (l34) -- (p4);
\end{tikzpicture}

    \vskip2ex
    c.~Hierarchy
  \end{minipage}
  \caption{
    Poset structures representing different sets of feasible vertex subsets $\hatcalP(\V)$.
    Each node represents a feasible vertex subset $V \in \hatcalP(\V)$ and each one-to-many link represents a feasible covered partition $\{V_1,\ldots,V_k\} \in \hatfrakC(\{V\})$.
  }
  \label{fig:data_structure}
\end{figure}
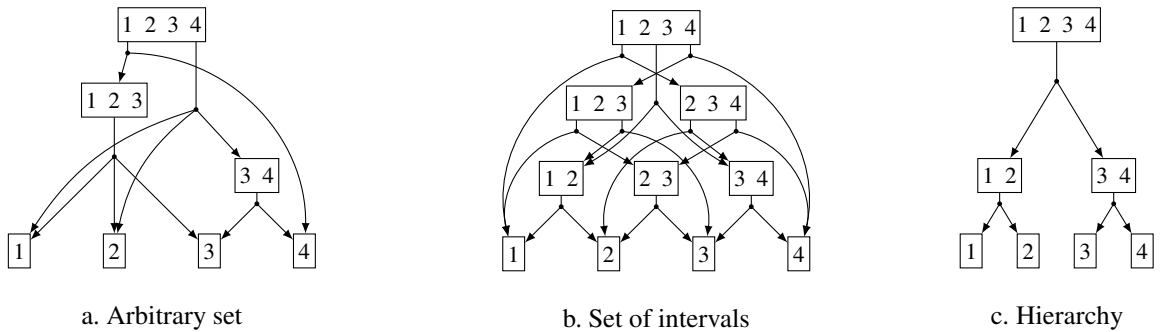

\paragraph{Data structure to represent the set of feasible vertex subsets}

In general, one can use a poset structure with one-to-many links to represent a set of feasible vertex subset $\hatcalP(\V)$ along with its covering relation~$\sqsubset$:
Each node of this data structure represents a feasible vertex subset $V \in \hatcalP(\V)$ and each link represents a covering relation $\{V\} \sqsupset \{V_1,\ldots,V_k\}$ going from one node $V$ to multiple nodes $\{V_1,\ldots,V_k\} \in \hatfrakC(\{V\})$, thus representing a feasible vertex partition that is covered by $\{V\}$.
Figure~\ref{fig:data_structure} gives an example of such poset structures in the case of (a)~an arbitrary set of feasible vertex subsets, (b)~a set of vertex intervals, and (c)~a hierarchy of vertex subsets (see Subsection~\ref{ssec:constrained}).

In what follows, we mark $|\hatfrakC(\V)|$ the sum of the numbers of vertex subsets for all covered partitions encoded in the poset structure, that is the total number of in-coming links to be encoded:
$$|\hatfrakC(\V)| \quad = \quad \sum_{\substack{V \in \hatcalP(\V) \\ \calV \in \hatfrakC(\{V\})}} {|\calV|} \text{.}$$
Note that, in particular cases such as (b) and (c), more straightforward data structures have been proposed in the literature:
For example, using a triangular matrix in the case of a set of intervals and using a simple poset in the case of a hierarchy (with one-to-one links).
This would allow to avoid the explicit representation of the covering relation through the use of one-to-many links:
\emph{E.g.}, memory requirement in $\Theta(n^2)$ instead of $\Theta(n^3)$ to represent the covering relations within a set of intervals by using a triangular matrix instead of a poset structure~\cite{RLP_MPI14}.
However, as the objective is to be here as general as possible, we do not enter in the details of such specialised implementations.

\begin{figure}[b!]
  \centering
  \begin{tikzpicture}[scale=0.5]

  \coordinate (pos) at (0,4.5);
  \tikzpart{12}{p12}

  \coordinate (pos) at (-2,3.25);
  \tikzpart{1}{p1}

  \coordinate (pos) at (2,3.25);
  \tikzpart{2}{p2}

  \node [draw, circle, thick, minimum width=2pt, inner sep=0pt] (l12) at ($(p12p1) - (0,0.75)$) {};
  \draw [toInLink] (p12p1) -- (l12);
  \draw [toOutLink] (l12) -- (p1);
  \draw [toOutLink] (l12) -- (p2);

  \coordinate (pos) at (-4.5,0);
  \tikzpart{12}{p34}

  \coordinate (pos) at (-3.25,2);
  \tikzpart{1}{p3}

  \coordinate (pos) at (-3.25,-2);
  \tikzpart{2}{p4}

  \node [draw, circle, thick, minimum width=2pt, inner sep=0pt] (l34) at ($(p34p1) + (0.75,0)$) {};
  \draw [toInLink] (p34p1) -- (l34);
  \draw [toOutLink] (l34) -- (p3);
  \draw [toOutLink] (l34) -- (p4);

  \coordinate (pos11) at (-2,2);
  \draw [blue, thick, densely dashed] ($(pos11) + (-0.5,-0.5)$) rectangle ($(pos11) + (0.5,0.5)$);
  \node [blue] (p11i11) at ($(pos11) + (0,0)$) {$(1,1)$};

  \coordinate (pos112) at (-2,0);
  \draw [blue, thick, densely dashed] ($(pos112) + (-0.5,-1)$) rectangle ($(pos112) + (0.5,1)$);
  \node [blue] (p112i21) at ($(pos112) + (0,0.5)$) {$(1,1)$};
  \node [blue] (p112i11) at ($(pos112) + (0,-0.5)$) {$(2,1)$};

  \coordinate (pos12) at (-2,-2);
  \draw [blue, thick, densely dashed] ($(pos12) + (-0.5,-0.5)$) rectangle ($(pos12) + (0.5,0.5)$);
  \node [blue] (p12i12) at ($(pos12) + (0,0)$) {$(2,1)$};

  \coordinate (pos121) at (0,2);
  \draw [thick, densely dashed] ($(pos121) + (-1,-0.5)$) rectangle ($(pos121) + (1,0.5)$);
  \node (p121i11) at ($(pos121) + (-0.5,0)$) {$(1,1)$};
  \node (p121i21) at ($(pos121) + (0.5,0)$) {$(1,2)$};

  \coordinate (pos1212) at (0,0);
  \draw [blue, thick, densely dashed] ($(pos1212) + (-1,-1)$) rectangle ($(pos1212) + (1,1)$);
  \node [blue] (p1212i11) at ($(pos1212) + (-0.5,0.5)$) {$(1,1)$};
  \node [blue] (p1212i21) at ($(pos1212) + (0.5,0.5)$) {$(1,2)$};
  \node [blue] (p1212i12) at ($(pos1212) + (-0.5,-0.5)$) {$(2,1)$};
  \node [blue] (p1212i22) at ($(pos1212) + (0.5,-0.5)$) {$(2,2)$};

  \coordinate (pos122) at (0,-2);
  \draw [thick, densely dashed] ($(pos122) + (-1,-0.5)$) rectangle ($(pos122) + (1,0.5)$);
  \node (p122i11) at ($(pos122) + (-0.5,0)$) {$(2,1)$};
  \node (p122i21) at ($(pos122) + (0.5,0)$) {$(2,2)$};

  \coordinate (pos21) at (2,2);
  \draw [thick, densely dashed] ($(pos21) + (-0.5,-0.5)$) rectangle ($(pos21) + (0.5,0.5)$);
  \node (p21i21) at ($(pos21) + (0,0)$) {$(1,2)$};

  \coordinate (pos212) at (2,0);
  \draw [blue, thick, densely dashed] ($(pos212) + (-0.5,-1)$) rectangle ($(pos212) + (0.5,1)$);
  \node [blue] (p212i21) at ($(pos212) + (0,0.5)$) {$(1,2)$};
  \node [blue] (p212i11) at ($(pos212) + (0,-0.5)$) {$(2,2)$};

  \coordinate (pos22) at (2,-2);
  \draw [thick, densely dashed] ($(pos22) + (-0.5,-0.5)$) rectangle ($(pos22) + (0.5,0.5)$);
  \node (p22i22) at ($(pos22) + (0,0)$) {$(2,2)$};

  \node [draw, circle, thick, minimum width=2pt, inner sep=0pt] (l121) at ($(pos121) + (0,0.75)$) {};
  \draw [toInLink] ($(pos121) + (0,0.5)$) -- (l121);
  \draw [toOutLink] (l121) to [bend right=10] ($(pos11) + (0.5,0.5)$);
  \draw [toOutLink] (l121) to [bend left=10] ($(pos21) + (-0.5,0.5)$);

  \node [blue, draw, circle, thick, minimum width=2pt, inner sep=0pt] (l1212a) at ($(pos1212) + (0,1.25)$) {};
  \draw [blue, toInLink] ($(pos1212) + (0,1)$) -- (l1212a);
  \draw [blue, toOutLink] (l1212a) to [bend right=10] ($(pos112) + (0.5,1)$);
  \draw [blue, toOutLink] (l1212a) to [bend left=10] ($(pos212) + (-0.5,1)$);


  \node [draw, circle, thick, minimum width=2pt, inner sep=0pt] (l122) at ($(pos122) + (0,0.75)$) {};
  \draw [toInLink] ($(pos122) + (0,0.5)$) -- (l122);
  \draw [toOutLink] (l122) to [bend right=10] ($(pos12) + (0.5,0.5)$);
  \draw [toOutLink] (l122) to [bend left=10] ($(pos22) + (-0.5,0.5)$);

  \node [blue, draw, circle, thick, minimum width=2pt, inner sep=0pt] (l112) at ($(pos112) - (0.75,0)$) {};
  \draw [blue, toInLink] ($(pos112) - (0.5,0)$) -- (l112);
  \draw [blue, toOutLink] (l112) to [bend left=10] ($(pos11) - (0.5,0.5)$);
  \draw [blue, toOutLink] (l112) to [bend right=10] ($(pos12) - (0.5,-0.5)$);

  \node [draw, circle, thick, minimum width=2pt, inner sep=0pt] (l1212b) at ($(pos1212) - (1.25,0)$) {};
  \draw [toInLink] ($(pos1212) - (1,0)$) -- (l1212b);
  \draw [toOutLink] (l1212b) to [bend left=10] ($(pos121) - (1,0.5)$);
  \draw [toOutLink] (l1212b) to [bend right=10] ($(pos122) - (1,-0.5)$);
  

  \node [draw, circle, thick, minimum width=2pt, inner sep=0pt] (l212) at ($(pos212) - (0.75,0)$) {};
  \draw [toInLink] ($(pos212) - (0.5,0)$) -- (l212);
  \draw [toOutLink] (l212) to [bend left=10] ($(pos21) - (0.5,0.5)$);
  \draw [toOutLink] (l212) to [bend right=10] ($(pos22) - (0.5,-0.5)$);

  \node at (4.25,0) {\LARGE $\longrightarrow$};

  \coordinate (posbis) at (7,0);
  \draw [blue, thick, densely dashed] ($(posbis) + (-1,-1)$) rectangle ($(posbis) + (1,1)$);
  \draw [blue, thick, densely dashed] ($(posbis) + (-1,0)$) -- ($(posbis) + (0,0)$);
  \draw [blue, thick, densely dashed] ($(posbis) + (0,-1)$) -- ($(posbis) + (0,1)$);
  \node [blue] (pbi11) at ($(posbis) + (-0.5,0.5)$) {$(1,1)$};
  \node [blue] (pbi21) at ($(posbis) + (0.5,0.5)$) {$(1,2)$};
  \node [blue] (pbi12) at ($(posbis) + (-0.5,-0.5)$) {$(2,1)$};
  \node [blue] (pbi22) at ($(posbis) + (0.5,-0.5)$) {$(2,2)$};

\end{tikzpicture}

  \caption{
    Bidimensional poset structure representing a set of feasible multiedge subsets $\hatcalP(\bVxV)$.
    Solid rectangles on the left and on the top represent the two unidimensional poset structures $\calP(\V) \times \calP(\V)$.
    Each dashed rectangle represents a feasible multiedge subset $\VV \in \hatcalP(\bVxV)$ and each one-to-many link represents a feasible covered partition $\{V_1{\times}V^\pm_1, \ldots, V_k{\times}V^\pm_k\} \in \hatfrakC(\{\VV\})$.
    The nodes and the links in blue form a rooted subposet that represents a feasible multiedge partition (represented on the right).
  }
  \label{fig:data_bistructure}
\end{figure}

\paragraph{Data structure to represent the set of feasible multiedge subsets}

In the case of a Cartesian multiedge partition $\calVVT \in \hatfrakP^\times(\bVxVxT)$ (see Subsection~\ref{ssec:powergraph}), given the data structure hereabove mentioned to represent the set of feasible \emph{vertex subsets} $\hatcalP(\V)$ and the set of feasible \emph{time intervals} $\calI(\T)$, one can build a similar data structure to represent the set of feasible \emph{multiedge subsets} $\hatcalP(\bVxVxT)$ by generalising to a tridimensional poset structure (see a bidimensional example in Figure~\ref{fig:data_bistructure}).
Each node in this tridimensional poset structure represents the Cartesian product of three nodes in the unidimensional poset structures:
$$\hatcalP(\bVxVxT) \quad = \quad \hatcalP(\V) \times \hatcalP(\V) \times \calI(\T) \text{.}$$
The number of nodes in the tridimensional poset structure is hence the product of the number of nodes in the unidimensional poset structures:
\begin{equation}
  \label{eq:node_number}
  |\hatcalP(\bVxVxT)| \quad = \quad |\hatcalP(\V)| \times |\hatcalP(\V)| \times |\calI(\T)| \text{.}
\end{equation}
The one-to-many links representing the covering relations are then obtained by copying, for each node in the tridimensional poset structure, the corresponding links of the unidimensional poset structures \emph{along all dimensions}, that is the links associated with the three nodes of the corresponding Cartesian product:
\begin{align*}
  \hatfrakC(\{\VxVxT\}) \quad = \quad & \{\{\VxVxT : V \in \calV\} : \calV \in \hatfrakC(\{V\})\}\\
  \cup \; & \{\{\VxVxT : V^\pm \in \calV^\pm\} : \calV^\pm \in \hatfrakC(\{V^\pm\})\}\\
  \cup \; & \{\{\VxVxT : T \in \calT\} : \calT \in \hatfrakC(\{T\})\} \text{.}
\end{align*}
The number of in-coming links in the tridimensional poset structure is hence given by the following formula:
\begin{align}
  \label{eq:link_number}
  \begin{split}
    |\hatfrakC(\bVxVxT)| \quad = \quad & |\hatfrakC(\V)| \; |\hatcalP(\V)| \; |\calI(\T)|\\
    + \; & |\hatcalP(\V)| \; |\hatfrakC(\V)| \; |\calI(\T)|\\
    + \; & |\hatcalP(\V)| \; |\hatcalP(\V)| \; |\hatfrakC(\T)| \text{.}
  \end{split}
\end{align}
Table~\ref{tab:data_structure} then exploits Equations~\ref{eq:node_number} and~\ref{eq:link_number}, as well as the basic combinatorial analysis presented in Subsection~\ref{ssec:constrained}, to present the size of the resulting unidimensional and tridimensional poset structures (in terms of nodes and in terms of links) depending on the types of feasibility constraints applying on the set of vertex subset.

\paragraph{Representing a feasible multiedge partition}

Within such a poset structure, a feasible multiedge partition can be simply represented as a rooted subposet, that is a poset which nodes and links are contained in the original poset and such that its maximal element is the same as the one of the original poset.
This amount in expressing the partition as a sequence of ``division steps'' going from the root of the poset structure (its maximal element) down to the nodes that represent the multiedge subsets belonging to the represented partition.
Hence, to do so, one only needs to record, for each node in this sequence, the next link to follow, or no link at all if one arrived at a subset that belongs to the encoded partition.
An example of such representation is given in blue in Figure~\ref{fig:data_bistructure} (in the case of a bidimensional poset structure).

\paragraph{Recursive computation of the cost function} 

In addition of being decomposable from partitions to their subsets (see Definition~\ref{def:information_loss}), the information loss function defined in Section~\ref{ssec:lossy} can also be decomposed into aggregative measures on these subsets, hence allowing for a recursive computation:
$$\lossf{\XX}{\pVVT}{(\Xout,\Xin,\Utm)} \; = \; \info\; - \; \suma\VVT \log_2 \left( \frac{\suma\VVT} {\suma_1\VVT \; \suma_2\VVT \; \suma_3\VVT} \right) \text{,}$$
where
$$\info \; = \; \displaystyle\sum_{(v,v^\pm,t) \in \VxVxT} {\frac{e(v,v^\pm,t)}{e(\V,\V,\T)} \log_2 \left( \frac{e(v,v^\pm,t)} {e(\V,\V,\T)} \bigg/ \frac{\ef(v,\V,\T) \; \ef(\V,v^\pm,\T) \; \ef(\V,\V,t)} {e(\V,\V,\T) \; e(\V,\V,\T) \; e(\V,\V,\T)} \right)}$$
is a constant term that does not depend on $\VVT$ and 
\begin{align*}
  \suma\VVT \quad = \quad \frac{e\VVT}{e(\V,\V,\T)} \text{,} \qquad & \suma_1\VVT \quad = \quad \frac{e(V,\V,\T)}{e(\V,\V,\T)} \text{,}\\
  \suma_2\VVT \quad = \quad \frac{e(\V,V^\pm,\T)}{e(\V,\V,\T)} \text{,} \qquad & \suma_3\VVT \quad = \quad \frac{e(\V,\V,T)}{e(\V,\V,\T)} \text{,}
\end{align*}
are all additive, in the sense that,
when applied to a multiedge subset $\VxVxT$, they can be expressed as the sum of measures applied to a set of multiedge subsets that forms a partition of $\VxVxT$:
$$\{ V_1{\times}V^\pm_1{\times}T_1, \ldots, V_k{\times}V^\pm_k{\times}T_k \} \; \in \; \frakP(\VxVxT) \quad \Rightarrow \quad \suma_k\VVT \; = \; \sum_i \suma_k (V_i,V^\pm_i,T_i) \text{.}$$
Hence, by first computing these measures for minimal elements in the aforementioned tridimensional poset structure, 
then browsing the poset in a bottom-up fashion, one can efficiently compute the values of the cost function for all feasible multiedge subsets it contains.
The information loss induced by a multiedge partition is then given by the sum of these values, to which is added the constant term (info) which is pre-computed once and for all.

\paragraph{Resulting optimisation algorithm}

The resulting recursive algorithm solving the Lossy MSCP is detailed below.
It is obtained by applying high-level Algorithm~1 to the poset structure we just described.
Hence, to solve the Lossy MSCP, one first needs (i)~to build this data structure, then (ii)~to compute the cost of each node it contains by applying the recursive approach above, and finally (iii)~to run Algorithm~2 on its maximal element.

\begin{table}
    \centering\small
    \begin{tabular}{|p{3cm}|p{1.5cm}p{2cm}|p{1.5cm}p{2cm}|}
      \hline
      ~ \newline For $|\V| = n$ vertices, \newline $\hatcalP(\T) = \calI(\T)$, & \multicolumn{2}{p{3.5cm}|}{\textbf{Number of nodes} \newline \small Number of \emph{vertex} subsets \newline or \emph{multiedge} subsets} & \multicolumn{2}{p{3.5cm}|}{\textbf{Number of links} \newline \small Number of \emph{vertex} coverings \newline or \emph{multiedge} coverings} \\[3ex]
      and $\hatcalP(\V) = \ldots$ & $|\hatcalP(\V)|$ & $|\hatcalP(\bVxVxT)|$ & $|\hatfrakC(\V)|$ & $|\hatfrakC(\bVxVxT)|$\\[1ex] \hline
      $\calP(\V)$ (complete set) \rule{0pt}{1\normalbaselineskip} & $\Theta(2^n)$ & $\Theta(4^n)$ & $\Theta(3^n)$ & $\Theta(6^n)$\\[1ex]
      $\calI(\V)$ (set of intervals) & $\Theta(n^2)$ & $\Theta(n^6)$ & $\Theta(n^3)$ & $\Theta(n^7)$\\[1ex]
      $\calH(\V)$ (hierarchy) & $\Theta(n)$ & $\Theta(n^4)$ & $\Theta(n)$ & $\Theta(n^5)$\\[1ex] \hline
    \end{tabular}
    \caption{Size of the poset structure in terms of nodes (columns~1 and~2) and in terms of in-coming links (columns~3 and~4) for a unidimensional poset structure (columns~1 and~3) and for a tridimensional poset structure (columns~2 and~4) when different types of feasibility constraints apply to the set of vertex subsets (rows).
      The set of feasible time subsets is here assumed to be the set of time intervals: $\hatcalP(\T) = \calI(\T) \; \Rightarrow \; |\hatcalP(\T)| = \Theta(n) \; \text{and} \; |\hatfrakC(\T)| = \Theta(n^3)$.}
    \label{tab:data_structure}
  \end{table}

\begin{figure}
  \begin{algorithm}[Recursive Algorithm to Solve the Lossy MSCP]~\\[1ex]
    \textbf{Global:} A tridimensional poset structure with one-to-many links such that, given a {\tt node} in this structure:\vspace{-1ex}
    \begin{itemize}[itemsep=0pt]
    \item {\tt node} represents a feasible multiedge subset $\VxVxT \in \hatcalP(\bVxVxT)$;
    \item {\tt node.links} is a list of one-to-many links such that, given a {\tt link} in this list:\vspace{-1ex}
      \begin{itemize}[itemsep=0pt]
      \item {\tt link} represents a feasible multiedge partition $\calVVT \in \hatfrakC(\{\VxVxT\})$ covered by $\{\VxVxT\}$;
      \item {\tt link.children} is a list of nodes in the tridimensional poset structure representing the feasible multiedge subsets $\{V_1{\times}V^\pm_1{\times}T_1, \ldots, V_k{\times}V^\pm_k{\times}T_k\}$ in the covered partition $\calVVT$;
      \end{itemize}
    \item {\tt node.cost} is a positive float value representing the Lagrange version $q_\lambda(\VxVxT) \in \mathbb{R}^+$ of the information loss induced by $\VxVxT$;
    \item {\tt node.optimal\_link} is one of the links in {\tt node.links} encoding a locally-optimal feasible multiedge partition $\calVVT^* \in \hatfrakP^*(\VxVxT)$ (it is set to {\tt NULL} before the algorithm starts);
    \item {\tt node.optimal\_cost} is a positive float value representing the Lagrange version $q_\lambda(\calVVT^*) \in \mathbb{R}^+$ of the information loss induced by the locally-optimal partition $\calVVT^*$ (it is set to {\tt NULL} before the algorithm starts).
    \end{itemize}
    \textbf{Input:}\vspace{-1ex}
    \begin{itemize}[itemsep=0pt]
    \item A {\tt node} in the tridimensional poset structure.
    \end{itemize}
    \textbf{Output:}\vspace{-1ex}
    \begin{itemize}[itemsep=0pt]
    \item Computes a locally-optimal feasible multiedge subset for {\tt node};
    \item Sets {\tt node.optimal\_link} accordingly;
    \item Returns {\tt node.optimal\_cost}.
    \end{itemize}
    \textbf{Algorithm:}
    {\tt \small
      \begin{itemize}
      \item[] compute\_optimal\_partition (node)
        \begin{itemize}
        \item[] if node.optimal\_cost not NULL, then return node.optimal\_cost
        \item[] node.optimal\_link $\leftarrow$ NULL
        \item[] node.optimal\_cost $\leftarrow$ node.cost
        \item[] for each link in node.links do
          \begin{itemize}
          \item[] cost $\leftarrow$ 0
          \item[] for each child in link.children do
            \begin{itemize}
            \item[] cost $\leftarrow$ cost + compute\_optimal\_partition (child)
            \end{itemize}
          \item[] if cost < node.optimal\_cost, then 
            \begin{itemize}
            \item[] node.optimal\_link $\leftarrow$ link
            \item[] node.optimal\_cost $\leftarrow$ cost
            \end{itemize}
          \end{itemize}
        \item[] return node.optimal\_cost
        \end{itemize}
      \end{itemize}
    }
    \label{algo:specific_algo}
  \end{algorithm}
\end{figure}

\paragraph{Complexity of the resulting optimisation algorithm}

The space complexity of the resulting optimisation algorithm is given in Table~\ref{tab:data_structure} by the number of nodes $|\hatcalP(\bVxVxT)|$ and by the number of links $|\hatfrakC(\bVxVxT)|$ that are encoded in the tridimensional poset structure.
It is exponential in the worst case, that is when all vertex subsets are feasible, polynomial of order~7 in the case of a set of vertex intervals, and polynomial of order~5 in the case of a vertex hierarchy.

The unidimensional poset structures, encoding the set of feasible vertex subsets and the set of time intervals, are considered as inputs of the optimisation problem and their building cost is hence not taken into account in the algorithm's complexity, although it is quite cheap and straightforward in the case of sets of intervals and hierarchies.
Building the corresponding tridimensional poset structure requires as many operations as there are nodes and in-coming links.
Filling it with the values of the cost function requires as many operations as there are in-coming links, thanks to the recursive decomposition of costs.
Hence, the time complexity to build the overall data structure is equivalent to its space complexity.

Finally, when applied to the maximal element, that is to the multiedge set $\bVxVxT$, Algorithm~2 is then recursively applied once to all feasible multiedge subsets $\VxVxT \in \hatcalP(\bVxVxT)$.
The bottleneck is then the summation of costs that are retrieved by the recursive calls, for each multiedge subset of each covered partition.
Hence, here again, there are as many such operations as they are in-coming links in the tridimensional poset structure, so the overall time complexity of the approach is equivalent to the space complexity of the data structure we presented.


\section{Conclusion}
\label{sec:conclusion}

This article presents a formal framework for the compression of temporal graphs.
By summarising homogeneous parts of the graph and replacing them with more general structural patterns, compression allows to reduce its description length while preserving its information content.
This framework first builds on a simple and limited combinatorial problem, that we call the \emph{Lossless Graph Compression Problem}, which exploits the (most classical) structural equivalence relation between vertices for the exact compression of simple graphs.
Among the proposed generalisations to address the more complex \emph{Lossy Multistream Compression Problem}, dealing with the approximated compression of temporal multigraphs, three main contributions are worth mentioning:
\begin{itemize}
\item The definition of an information-theoretic measure, relying on a proper formalisation of a multigraph stochastic model, to quantify and to control the information that is lost during compression, while also taking into account additional information that might be reinjected during the decompression step;
\item The enhancement of the  solution space of the initial problem (i)~by defining a less constrained partitioning of the graph working at the multiedge level instead of the vertex level, and (ii)~by allowing to express and to preserve additional vertex structures during compression;
\item The generalisation from static graphs to temporal graphs by exploiting the \emph{link stream} representation, which consists in the extension of the set of multiedges by a third dimension representing the temporal evolution of these edges, thus allowing the natural extension of all notions that have been previously introduced.
\end{itemize}

Building on a previous algorithmic framework to solve special versions of the \emph{Set Partitioning Problem}, an exact algorithm is finally introduced for the Lossy MSCP.
While it is exponential in the worst case, it is showed to be polynomial when particular vertex structures are assumed.
Yet, in order to achieve the compression of large-scale temporal graphs, future research would need to work on heuristics for the approximate solving of the Lossy MSCP.
This would require the definition of adequate operators on multiedge partitions to browse the solution space, taking into account its particular algebraic structure to efficiently evaluate slight modifications of the considered partitions, and to thus proceed to a greedy search for a local optima.
Another improvement, regarding the current implementation of the optimisation algorithm, would consists in the acknowledgement that most link streams that are considered in empirical research are quite sparse, meaning that the support of the edge function is quite small.
Hence, the data structures proposed in this article would benefit from a sparse representation of the data to decrease computation resources in real-case applications of this framework.


\section*{Acknowledgement}

The author sincerely thanks
Yves Demazeau,
Damien Dosimont,
Matthieu Latapy,
L\'eonard Panichi,
Hindol Rakshit,
Fabrice Rossi,
Mridul Seth,
Catherine Matias,
Lucas Mello Schnorr,
Lionel Tabourier,
Fabien Tarissan,
Tiphaine Viard,
and Jean-Marc Vincent,
for their valuable feedback.
He also gratefully thanks all the members of Prof.\ Grasset's team for their warm welcome within their premises.
This work has been partially funded by the European Commission H2020 FETPROACT 2016-2017 program under grant 732942 (ODYCCEUS) and by the French National Agency of Research (ANR) under grant ANR-15-CE38-0001 (AlgoDiv).


\section*{References}
\bibliographystyle{elsarticle-num} 
\bibliography{ms}

\begin{thebibliography}{10}
\expandafter\ifx\csname url\endcsname\relax
  \def\url#1{\texttt{#1}}\fi
\expandafter\ifx\csname urlprefix\endcsname\relax\def\urlprefix{URL }\fi
\expandafter\ifx\csname href\endcsname\relax
  \def\href#1#2{#2} \def\path#1{#1}\fi

\bibitem{Zhou09}
F.~Zhou, S.~Mahler, H.~Toivonen, {Review of Network Abstraction Techniques},
  in: Workshop on Explorative Analytics of Information Networks at ECML PKDD,
  2009, pp. 50--63.

\bibitem{Toivonen11}
H.~Toivonen, F.~Zhou, A.~Hartikainen, A.~Hinkka, {Compression of Weighted
  Graphs}, in: Proceedings of the 17\textsuperscript{th} ACM SIGKDD
  International Conference on Knowledge Discovery and Data Mining (KDD'11),
  ACM, New York, NY, USA, 2011, pp. 965--973.

\bibitem{Serafino13}
P.~Serafino, {Speeding Up Graph Clustering via Modular Decomposition Based
  Compression}, in: Proceedings of the 28th Annual ACM Symposium on Applied
  Computing (SAC'13), ACM, New York, NY, USA, 2013, pp. 156--163.

\bibitem{Navlakha08}
S.~Navlakha, R.~Rastogi, N.~Shrivastava, {Graph Summarization with Bounded
  Error}, in: Proceedings of the 2008 ACM SIGMOD International Conference on
  Management of Data (SIGMOD'08), ACM, New York, NY, USA, 2008, pp. 419--432.

\bibitem{Zhang10}
N.~Zhang, Y.~Tian, J.~M. Patel, {Discovery-Driven Graph Summarization}, in:
  Proceedings of the 26th International Conference on Data Engineering
  (ICDE'2010), 2010, pp. 880--891.

\bibitem{LeFevre10}
K.~LeFevre, E.~Terzi, {GraSS: Graph Structure Summarization}, in: Proceedings
  of the SIAM International Conference on Data Mining (SDM'2010), 2010, pp.
  454--465.

\bibitem{Pinaud12}
B.~Pinaud, G.~Melan{\c c}on, J.~Dubois, {PORGY: A Visual Graph Rewriting
  Environment for Complex Systems}, {Computer Graphics Forum} 31~(3) (2012)
  1265--1274.

\bibitem{Dwyer13}
T.~Dwyer, N.~H. Riche, K.~Marriott, C.~Mears, {Edge Compression Techniques for
  Visualization of Dense Directed Graphs}, IEEE Transactions on Visualization
  and Computer Graphics 19~(12) (2013) 2596--2605.

\bibitem{Ahnert14}
S.~E. Ahnert, {Generalised power graph compression reveals dominant
  relationship patterns in complex networks}, Scientific Reports 4~(4385).

\bibitem{Holme13}
P.~Holme, J.~Saram{\"a}ki (Eds.), {Temporal Networks}, {Understanding Complex
  Systems}, Springer-Verlag Berlin Heidelberg, 2013.

\bibitem{Viard16}
T.~Viard, M.~Latapy, C.~Magnien, {Computing maximal cliques in link streams},
  Theoretical Computer Science 609 (2016) 245–252.

\bibitem{Latapy17}
M.~Latapy, T.~Viard, C.~Magnien, {Stream Graphs and Link Streams for the
  Modeling of Interactions over Time}, arXiv:1710.04073.

\bibitem{Borgatti92}
S.~P. Borgatti, M.~G. Everett, {Regular blockmodels of multiway, multimode
  matrices}, Social Networks 14~(1) (1992) 91--120.

\bibitem{Narmadha16}
N.~Narmadha, R.~Rathipriya, {Triclustering: An Evolution of Clustering}, in:
  Proceedings of the Online International Conference on Green Engineering and
  Technologies (IC-GET'16), 2016, pp. 1--4.

\bibitem{Guigoures12}
R.~Guigour{\`e}s, M.~Boull{\'e}, F.~Rossi, {A Triclustering Approach for Time
  Evolving Graphs}, in: {Co-clustering and Applications International
  Conference on Data Mining Workshop}, {IEEE}, Brussels, Belgium, 2012, pp.
  115--122.

\bibitem{Lorrain71}
F.~Lorrain, H.~C. White, {Structural equivalence of individuals in social
  networks}, The Journal of Mathematical Sociology 1~(1) (1971) 49--80.

\bibitem{Balas76}
E.~Balas, M.~W. Padberg, {Set Partitioning: A Survey}, SIAM Review 18~(4)
  (1976) 710--760.

\bibitem{RLP_ICTAI14}
R.~Lamarche-Perrin, Y.~Demazeau, J.-M. Vincent, {A Generic Algorithmic
  Framework to Solve Special Versions of the Set Partitioning Problem}, in:
  A.~Andreou, G.~A. Papadopoulos (Eds.), Proceedings of the 2014 IEEE
  26\textsuperscript{th} International Conference on Tools with Artificial
  Intelligence (ICTAI'14), IEEE Computer Society, 2014, pp. 891--897.

\bibitem{RLP_CLUSTER14}
D.~Dosimont, R.~Lamarche-Perrin, L.~M. Schnorr, G.~Huard, J.-M. Vincent, {A
  Spatiotemporal Data Aggregation Technique for Performance Analysis of
  Large-scale Execution Traces}, in: M.~S. P\'erez, G.~Antoniu, K.~Keahey
  (Eds.), Proceedings of the 2014 IEEE International Conference on Cluster
  Computing (CLUSTER'14), IEEE Computer Society, 2014, pp. 149--157.

\bibitem{Batagelj92}
V.~Batagelj, A.~Ferligoj, P.~Doreian, {Direct and indirect methods for
  structural equivalence}, Social Networks 14~(1) (1992) 63--90, special Issue
  on Blockmodels.

\bibitem{Hanneman05}
R.~A. Hanneman, M.~Riddle, {Introduction to social network methods}, University
  of California, Riverside, CA, 2005.

\bibitem{Fortunato10}
S.~Fortunato, {Community detection in graphs}, Physics Reports 486~(3) (2010)
  75--174.

\bibitem{Dhillon03}
I.~S. Dhillon, S.~Mallela, D.~S. Modha, {Information-theoretic Co-clustering},
  in: Proceedings of the Ninth ACM International Conference on Knowledge
  Discovery and Data Mining (SIGKDD'03), ACM, New York, NY, USA, 2003, pp.
  89--98.

\bibitem{RLP_TCCI14}
R.~Lamarche-Perrin, Y.~Demazeau, J.-M. Vincent, {Building Optimal Macroscopic
  Representations of Complex Multi-agent Systems. Application to the Spatial
  and Temporal Analysis of International Relations through News Aggregation},
  in: N.~T. Nguyen, R.~Kowalczyk, J.~M. Corchado, J.~Bajo (Eds.), Transactions
  on Computational Collective Intelligence, Vol.~XV of LNCS 8670,
  Springer-Verlag Berlin, Heidelberg, 2014, pp. 1--27.

\bibitem{Hofstad16}
R.~v.~d. Hofstad, {Configuration Model}, Vol.~1, Cambridge University Press,
  2016, Ch. III.7, p. 216–255.

\bibitem{Kullback51}
S.~Kullback, R.~A. Leibler, {On Information and Sufficiency}, The Annals of
  Mathematical Statistics 22~(1) (1951) 79--86.

\bibitem{Cover91}
T.~M. Cover, J.~A. Thomas, {Elements of Information Theory,
  2\textsuperscript{nd} Edition}, John Wiley \& Sons, Inc., Hoboken, NJ, 2006.

\bibitem{He06}
H.~He, A.~K. Singh, {Closure-Tree: An Index Structure for Graph Queries}, in:
  22\textsuperscript{nd} International Conference on Data Engineering
  (ICDE'06), 2006, pp. 38--38.

\bibitem{Feder95}
T.~Feder, R.~Motwani, {Clique Partitions, Graph Compression and Speeding-Up
  Algorithms}, Journal of Computer and System Sciences 51~(2) (1995) 261--272.

\bibitem{Henandez12}
C.~Hern\'andez, G.~Navarro, {Compressed Representation of Web and Social
  Networks via Dense Subgraphs}, in: L.~Calder\'on-Benavides,
  C.~Gonz\'alez-Caro, E.~Ch\'avez, N.~Ziviani (Eds.), String Processing and
  Information Retrieval, Vol. 7608 of Lecture Notes in Computer Science,
  Springer Berlin Heidelberg, 2012, pp. 264--276.

\bibitem{Dhabu13}
M.~Dhabu, P.~S. Deshpande, S.~Vishwakarma, Partition based graph compression,
  International Journal of Advanced Computer Science and Applications 4~(9)
  (2013) 7--12.

\bibitem{Roy09}
D.~M. Roy, T.~Y. Whye, {The Mondrian Process}, in: {Advances in Neural
  Information Processing Systems}, Vol.~21, Curran Associates, Inc., 2009, pp.
  1377--1384.

\bibitem{RLP_MPI14}
R.~Lamarche-Perrin, Y.~Demazeau, J.-M. Vincent, {A Generic Algorithmic
  Framework to Solve Special Versions of the Set Partiioning Problem}, Tech.
  Rep. MIS-Preprint 105/2014, Max Planck Institute for Mathematics in the
  Sciences, Leipzig, Germany (2014).

\bibitem{Sandholm99}
T.~Sandholm, K.~Larson, M.~Andersson, O.~Shehory, F.~Tohm\'e, Coalition
  structure generation with worst case guarantees, Artificial Intelligence
  111~(1-2) (1999) 209--238.

\bibitem{Rahwan08}
T.~Rahwan, N.~R. Jennings, {Coalition Structure Generation: Dynamic Programming
  Meets Anytime Optimisation}, in: Proceedings of the Twenty-third Conference
  on Artificial Intelligence, AAAI, 2008, pp. 156--161.

\bibitem{Sandholm02}
T.~Sandholm, {Algorithm for optimal winner determination in combinatorial
  auctions}, Artificial Intelligence 135 (2002) 1--54.

\bibitem{Pons11}
P.~Pons, M.~Latapy, {Post-processing hierarchical community structures: Quality
  improvements and multi-scale view}, Theoretical Computer Science 412~(8-10)
  (2011) 892--900.
\newblock \href {http://dx.doi.org/http://dx.doi.org/10.1016/j.tcs.2010.11.041}
  {\path{doi:http://dx.doi.org/10.1016/j.tcs.2010.11.041}}.

\bibitem{Cloitre02}
B.~Cloitre, {Sequence A003095}, in: The On-Line Encyclopedia of Integer
  Sequences, \url{http://oeis.org/A003095}, 2002.

\bibitem{McGarvey07}
G.~McGarvey, {Sequence A135361}, in: The On-Line Encyclopedia of Integer
  Sequences, \url{http://oeis.org/A135361}, 2007.

\bibitem{Anily91}
S.~Anily, A.~Federgruen, {Structured Partitioning Problems}, Operations
  Research 39~(1) (1991) 130--149.

\bibitem{Jackson04}
B.~{Jackson, J.D. Scargle, D. Barnes, S. Arabhi, A. Alt, \emph{et al.}}, An
  algorithm for optimal partitioning of data on an interval, IEEE Signal
  Processing Letters 12~(2) (2005) 105--108.

\bibitem{Pagano13}
G.~Pagano, D.~Dosimont, G.~Huard, V.~Marangozova-Martin, J.-M. Vincent, {Trace
  Management and Analysis for Embedded Systems}, in: {Proceedings of the
  7\textsuperscript{th} International Symposium on Embedded Multicore SoCs
  (MCSoC'13)}, IEEE Computer Society Press, 2013, pp. 119--122.

\bibitem{Rothkopf98}
M.~H. Rothkopf, A.~Peke\v{c}, R.~M. Harstad, {Computationally Manageable
  Combinational Auctions}, Management Science 44~(8) (1998) 1131--1147.

\bibitem{Ahuja93}
R.~K. Ahuja, T.~L. Magnanti, J.~B. Orlin, {Network Flows: Theory, Algorithms,
  and Applications}, Prentice-Hall, Inc., 1993, Ch. {Lagrangian Relaxation and
  Network Optimization}, pp. 598--648.

\bibitem{Chakravarty82}
A.~K. Chakravarty, J.~B. Orlin, U.~G. Rothblum, A partitioning problem with
  additive objective with an application to optimal inventory groupings for
  joint replenishment, Operations Research 30~(5) (1982) 1018--1022.

\bibitem{Vidal93}
R.~V.~V. Vidal, {Optimal Partition of an Interval -- The Discrete Version}, in:
  R.~V.~V. Vidal (Ed.), Applied Simulated Annealing, Vol. 396 of Lecture Notes
  in Economics and Mathematical Systems, Springer Berlin Heidelberg, 1993, pp.
  291--312.

\bibitem{Davey02}
B.~Davey, H.~Priestley, {Introduction to Lattices and Order}, 2nd Edition,
  Cambridge University Press, (2002).

\end{thebibliography}

\clearpage

\appendix
\section{Solving the Set Partitioning Problem}
\label{sec:spp}

This appendix provides a detailed description of the algebraic properties of the \emph{Set Partitioning Problem} (SPP), leading to a principle of optimality that we exploit to derive the algorithm presented in Subsection~\ref{ssec:reduction} of this article.
This part of our work has been previously published in the proceedings of the 2014 IEEE International Conference on Tools with Artificial Intelligence (ICTAI'14)~\cite{RLP_ICTAI14} and detailed in a research report of the Max Planck Institute for Mathematics in the Sciences~\cite{RLP_MPI14}.

\paragraph{Algebraic Structure of the Solution Space}

Given a set of objects $\pop = \{x_1,\ldots,x_n\}$ and a set of feasible subsets $\hatcalP(\pop) \subseteq \calP(\pop)$, the resulting set of feasible partitions $\hatfrakP(\pop) \subseteq \frakP(\pop)$ is structured by an essential algebraic relation, usually referred to as the \emph{refinement relation} $\subseteq$~\cite{Davey02}.

\begin{hidendefinition}[Refinement Relation]~\\
  Partition $\calX \in \hatfrakP(\pop)$ \emph{refines} partition~$\calY \in \hatfrakP(\pop)$, and we mark $\calX \subseteq \calY$,
  if and only if each subset in $\calX$ is a subset of a subset in $\calY$:
  $$\calX \subseteq \calY \quad \Leftrightarrow \quad \forall X \in \calX, \quad \exists Y \in \calY, \quad X \subseteq Y$$
  Given a partition $\calX \in \hatfrakP(\pop)$, we define $\hatfrakR(\calX) \subseteq \hatfrakP(\pop)$ as the \emph{set of feasible partitions refining~$\calX$}:
  $$\hatfrakR(\calX) \; = \; \{\calY \in \hatfrakP(\pop) : \calY \subseteq \calX\}$$
\end{hidendefinition}
As this binary relation is reflexive, antisymmetric, and transitive, it defines a partial order on the partition set $\hatfrakP(\pop)$
that consequently forms a poset and can be represented as a Hasse diagram~\cite{Davey02}. 
  In particular, if the \emph{minimal partition} $\calX_\bot = \{\{x_1\},\ldots,\{x_n\}\}$ is feasible, then it refines all feasible partitions:
  $$\forall x \in \pop, \; \{x\} \in \hatcalP(\pop) \qquad \Rightarrow \qquad \forall \calX \in \hatfrakP(\pop), \quad \calX_\bot \in \hatfrakR(\calX)$$
  and if the \emph{maximal partition} $\calX_\top = \{\{x_1,\ldots,x_n\}\}$ is feasible, then it is refined by all feasible partitions:
  $$\pop \in \hatcalP(\pop) \qquad \Rightarrow \qquad \hatfrakR(\calX_\top) = \hatfrakP(\pop)$$

The \emph{covering relation} $\sqsubset$ is the transitive reduction of the refinement relation, that is the binary relation which holds between immediate ``neighbours'' with respect to $\subseteq$.

\begin{hidendefinition}[Covering relation]~\\
  Partition $\calX \in \hatfrakP(\pop)$ is \emph{covered} by partition $\calY \in \hatfrakP(\pop)$, and we mark $\calX \sqsubset \calY$,
  if and only if $\calX$ and $\calY$ are different, $\calX$ refines $\calY$, and there is no other feasible partition ``in-between'':
  $$\calX \sqsubset \calY \quad \Leftrightarrow \quad \calX \subsetneq \calY \quad {\rm and} \quad \nexists \calZ \in \hatfrakP(\pop), \;\; \calX \subsetneq \calZ \subsetneq \calY$$
  Given a partition $\calX \in \hatfrakP(\pop)$, we define $\hatfrakC(\calX)$ as the \emph{set of feasible partitions covered by~$\calX$}:
  $$\hatfrakC(\calX) \; = \; \{\calY \in \hatfrakP(\pop) : \calY \sqsubset \calX\}$$
\end{hidendefinition}
As it is shown in what follows, these two relations give essential algebraic tools to cleverly search for optimal partitions in $\hatfrakP(\pop)$.

\paragraph{Principle of Optimality}

In dynamic programming, finding a principle of optimality consists in showing that the solution space has an optimal substructure,
that is that the solution to the optimisation problem can be obtained by recursively combining locally-optimal solutions to several subproblems.
Intuitively, in the case of the SPP, one can rely on the fact that \emph{the union of optimal partitions of subsets forms an interesting candidate partition of the union of these subsets}~\cite{RLP_MPI14}.
Hence, by appropriately decomposing the initial set into subsets, one might provide a computationally efficient procedure to recursively build such an optimal solution (see Figure~\ref{fig:algo_decomposition_ordre} for an illustration of the following principle).

\begin{theorem}[Principle of Optimality~\cite{RLP_MPI14}]~\\
  Let $\pop$ be a set of objects, $\hatcalP(\pop)$ a set of feasible subsets, $\hatfrakP(\pop)$ the resulting set of feasible partitions, and $c : \hatcalP(\pop) \rightarrow \mathbb{R}^+$ a cost function defining partition optimality.
  For any feasible partition $\calX \in \hatfrakP(\pop)$, the union of locally-optimal feasible partitions of the subsets in $\calX$ is optimal among the feasible refinements of $\calX$:
  \begin{equation}
    \label{eq:optprinciple}
    \forall X \in \calX, \; \calX_X^* \in \hatfrakP^*(X) \quad \Rightarrow \quad \left( \bigcup_{X \in \calX} \calX_X^* \right) \in \hatfrakR^*(\calX) \text{,}
  \end{equation}
  where $\hatfrakP^*(X)$ is the set of optimal feasible partitions of $X$ and $\hatfrakR^*(\calX)$ is the set of optimal feasible partitions refining $\calX$.
\end{theorem}

\begin{proof}
Let $\calX \in \hatfrakP(\pop)$ be a feasible partition of $\pop$ and, for all $X \in \calX$, let $\calX^*_{X} \in \hatfrakP^*(X)$ be a locally-optimal feasible partition of $X$, meaning that
\begin{equation}
  \label{eq:a}
  \forall \calX^\pm_{X} \in \hatfrakP(X), \quad c(\calX^*_{X}) \; = \; \sum_{X^* \in \calX^*_{X}} c(X^*) \; \leq \; \sum_{X^\pm \in \calX^\pm_{X}} c(X^\pm) \; = \; c(\calX^\pm_{X}) \text{.} \tag{$\alpha$}
\end{equation}
In the following, we mark $\calX^* = \displaystyle\bigcup_{X \in \calX} {\calX^*_{X}}$.
Here is an example of such setting:

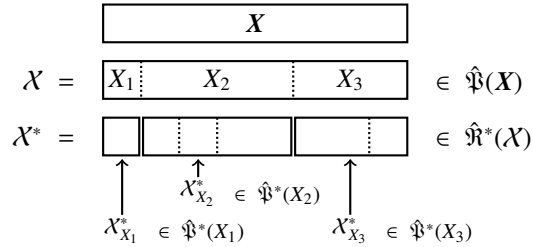
\begin{center}
  \begin{tikzpicture}
    \node at (0,-0.25) {$\pop$};
    \draw [thick] (-2,0) rectangle (2,-0.5);

    \node [anchor=east] at (-2.25,-1) {$\calX \; = $};
    \node at (-1.75,-1) {$X_1$};
    \node at (-0.5,-1) {$X_2$};
    \node at (1.25,-1) {$X_3$};
    \node [anchor=west] at (2.25,-1) {$\in \; \hatfrakP(\pop)$};
    \draw [densely dotted, thick] (-1.5,-0.75) -- (-1.5,-1.25);
    \draw [densely dotted, thick] (0.5,-0.75) -- (0.5,-1.25);
    \draw [thick] (-2,-0.75) rectangle (2,-1.25);

    \node [anchor=east] at (-2.25,-1.75) {$\calX^* \; = $};
    \node [inner sep=1pt, font=\footnotesize] (y1) at (-1.75,-3) {$\calX^*_{X_1}$};
    \node [inner sep=1pt, font=\footnotesize] (z1) at (-0.75,-3) {$\; \in \; \hatfrakP^*(X_1)$};
    \node [inner sep=1pt, font=\footnotesize] (y2) at (-0.75,-2.5) {$\calX^*_{X_2}$};
    \node [inner sep=1pt, font=\footnotesize] (z2) at (0.25,-2.5) {$\; \in \; \hatfrakP^*(X_2)$};
    \node [inner sep=1pt, font=\footnotesize] (y3) at (1.25,-3) {$\calX^*_{X_3}$};
    \node [inner sep=1pt, font=\footnotesize] (z3) at (2.25,-3) {$\; \in \; \hatfrakP^*(X_3)$};
    \node [anchor=west] at (2.25,-1.75) {$\in \; \hatfrakR^*(\calX)$};

    \draw [->, thick] (y1.north) -- (-1.75,-2.05);
    \draw [->, thick] (y2.north) -- (-0.75,-2.05);
    \draw [->, thick] (y3.north) -- (1.25,-2.05);

    \draw [thick] (-2,-1.5) rectangle (-1.525,-2);

    \draw [densely dotted, thick] (-1,-1.5) -- (-1,-2);
    \draw [densely dotted, thick] (-0.5,-1.5) -- (-0.5,-2);
    \draw [thick] (-1.475,-1.5) rectangle (0.475,-2);

    \draw [densely dotted, thick] (1.5,-1.5) -- (1.5,-2);
    \draw [thick] (0.525,-1.5) rectangle (2,-2);
  \end{tikzpicture}
\end{center}
First, since any subset $X^* \in \calX^*_X$ is a feasible subset of $X \in \calX$, then $\calX^*$ is a feasible refinement of $\calX$: $\calX^* \in \hatfrakR(\calX)$.
Second, let $\calX^\pm = \displaystyle\bigcup_{X \in \calX} {\calX^\pm_{X}}$ with $\calX^\pm_X \in \hatfrakP(X)$ be another feasible refinement of $\calX$.
We then have, applying Equation~$\alpha$:
\begin{equation}
  \label{eq:b}
  c(\calX^*) \; = \; \sum_{X \in \calX} c(\calX^*_X) \; \leq \; \sum_{X \in \calX} c(\calX^\pm_X) \; = \; c(\calX^\pm) \text{.} \tag{$\beta$}
\end{equation}
Therefore, $\calX^*$ is optimal among the feasible refinements of $\calX$: $\calX^* \in \hatfrakR^*(\calX)$.
\end{proof}

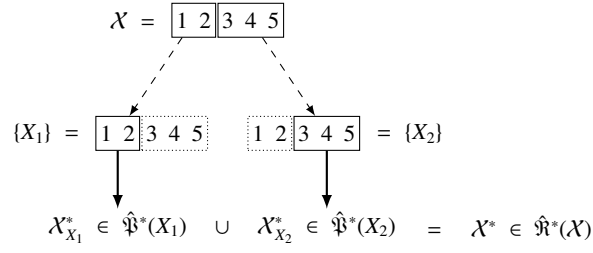
\begin{figure}[t!]
  \centering
  \begin{tikzpicture}[scale=0.5]

  \tikzboldtrue
  \coordinate (pos) at (0,2);
  \tikzpart{12-345+}{p1}
  \node [font=\LARGE, align=right] at ($(p1p1) - (1.8,0)$) {$\{X_1\} \; = \;$};
  \node [font=\LARGE, align=center] at ($(p1p1) - (0,2.5)$) {$\calX^*_{X_1} \; \in \; \hatfrakP^*(X_1)$};

  \coordinate (pos) at (4,2);
  \tikzpart{12+345}{p2}
  \node [font=\LARGE, align=left] at ($(p2p3) + (2.1,0)$) {$\; = \; \{X_2\}$};
  \node [font=\LARGE, align=center] at ($(p2p3) - (0,2.5)$) {$\calX^*_{X_2} \; \in \; \hatfrakP^*(X_2)$};

  \node [font=\LARGE, align=center] at ($0.5*(p1p1) + 0.5*(p2p3) - (0,2.5)$) {$\cup$};

  \node [font=\LARGE, align=left] at ($(p2p3) - (-4.5,2.5)$) {$\quad = \quad \calX^* \; \in \; \hatfrakR^*(\calX)$};

  \coordinate (pos) at (2,5);
  \tikzpart{12-345}{p12}
  \node [font=\LARGE, align=right] at ($(pos) - (2.5,0)$) {$\calX \; =$};

  \draw[toPart] (p12p1) -- (p1p1);
  \draw[toPart] (p12p3) -- (p2p3);

  \draw[toPartition] (p1p1) -- ++(0,-2);
  \draw[toPartition] (p2p3) -- ++(0,-2);

\end{tikzpicture}

  \caption{Recursively solving the SPP by applying the principle of optimality on a partition $\calX$ to find an optimal partition among its refinements (see Equation~\ref{eq:optprinciple})}
  \label{fig:algo_decomposition_ordre}
\end{figure}

\paragraph{Branching the Solution Space}
\label{subsec:decomposition}

Given a feasible subset $X \in \hatcalP(\pop)$ for which one wants to compute a locally-optimal feasible partition $\calX^* \in \hatfrakP^*(X)$, a \emph{branching} first consists in building subspaces $\hatfrakP_i(X) \subseteq \hatfrakP(X)$ that cover the solution space:
$$\hatfrakP_1(X) \; \cup \; \ldots \; \cup \; \hatfrakP_k(X) \quad = \quad \hatfrakP(X) \text{.}$$
Then, once one has found locally-optimal partitions within these subspaces $\calX_1^* \in \hatfrakP_1^*(X), \ldots, \calX_k^* \in \hatfrakP_k^*(X)$, one can easily solve the global optimisation problem by simply choosing among these local solutions one that minimises the objective:
\begin{equation}
  \label{eq:dec_generic}
  \argmin_{\calX \, \in \, \{\calX_1^*,\ldots,\calX_k^*\}} {c(\calX)} \quad \subseteq \quad \hatfrakP^*(X) \text{.}
\end{equation}

For that purpose, the covering relation provides ``elementary steps'' to branch the solution space, each branch corresponding to a direction to go down in the partition poset.
First, assuming that the maximal partition $\{X\}$ is feasible\footnote{
  If this is not the case for the initial set $\pop$, the following approach can easily be generalised by sequentially applying the algorithm to all maximal partitions, that is maximal elements in the poset of partitions induced by the refinement relation.
},
we know that all feasible partitions of $X$ refine the maximal partition $\{X\}$:
$$\hatfrakP(X) \quad = \quad \hatfrakR(\{X\}) \text{.}$$
Second, for any such feasible partition $\calX \in \hatfrakP(X)$, a refining partition of $\calX$ is either \emph{the partition $\calX$ itself}, or \emph{a partition that refines a partition covered by~$\calX$}.
Hence, the solution space can be branched the following way (see Figure~\ref{fig:algo_branching} for an illustration of such branching):
\begin{equation}
  \label{eq:raf_macro_rec}
  \hatfrakP(X) \quad = \quad \{\{X\}\} \; \cup \; \left( \bigcup_{\calX \, \in \, \hatfrakC(\{X\})} {\hatfrakR(\calX)} \right)
\end{equation}

\begin{figure}[h]
  \centering
  \begin{tikzpicture}[scale=0.5]

  \tikzboldtrue
  \coordinate (pos) at (6,5);
  \tikzpart{12345}{p1234}
  \node [font=\LARGE, align=right] at ($(pos) - (3,0)$) {$\{X\} \; =$};

  \coordinate (pos) at (0,2);
  \tikzpart{1-2345}{p1}
  \node [font=\LARGE, align=right] at ($(pos) - (3.25,0)$) {$\hatfrakC(\{X\}) \; = \; \bigg\{ \;$};
  \node [font=\LARGE] at ($(pos) + (0,0.9)$) {$\calX_1$};

  \coordinate (pos) at (4,2);
  \tikzpart{12-345}{p2}
  \node [font=\LARGE] at ($(pos) + (-0.25,0.9)$) {$\calX_2$};

  \coordinate (pos) at (8,2);
  \tikzpart{123-45}{p3}
  \node [font=\LARGE] at ($(pos) + (0.25,0.9)$) {$\calX_3$};

  \coordinate (pos) at (12,2);
  \tikzpart{1234-5}{p4}
  \node [font=\LARGE, align=left] at ($(pos) + (1.75,0)$) {$\; \bigg\}$};
  \node [font=\LARGE] at ($(pos) + (0,0.9)$) {$\calX_4$};

  \coordinate (pos) at (-4,0);
  \node [font=\LARGE, align=right] at ($(pos) + (0,0)$) {$\{\{X\}\}$};
  \node [font=\LARGE, align=center] at ($(pos) + (2,0)$) {$\cup$};

  \coordinate (pos) at (0,0);
  \tikzpart{-}{p1b}
  \node [font=\LARGE, align=center] at ($(pos) + (0,0)$) {$\hatfrakR(\calX_1)$};
  \node [font=\LARGE, align=center] at ($(pos) + (2,0)$) {$\cup$};

  \coordinate (pos) at (4,0);
  \tikzpart{-}{p2b}
  \node [font=\LARGE, align=center] at ($(pos) + (0,0)$) {$\hatfrakR(\calX_2)$};
  \node [font=\LARGE, align=center] at ($(pos) + (2,0)$) {$\cup$};

  \coordinate (pos) at (8,0);
  \tikzpart{-}{p3b}
  \node [font=\LARGE, align=center] at ($(pos) + (0,0)$) {$\hatfrakR(\calX_3)$};
  \node [font=\LARGE, align=center] at ($(pos) + (2,0)$) {$\cup$};

  \coordinate (pos) at (12,0);
  \tikzpart{-}{p4b}
  \node [font=\LARGE, align=center] at ($(pos) + (0,0)$) {$\hatfrakR(\calX_4)$};
  \node [font=\LARGE, align=left] at ($(pos) + (3,0)$) {$\quad = \quad \hatfrakP(\{X\})$};

  \draw[toPartition] (p1234) -- (p1);
  \draw[toPartition] (p1234) -- (p2);
  \draw[toPartition] (p1234) -- (p3);
  \draw[toPartition] (p1234) -- (p4);

  \draw[toPart] (p1) -- (p1b);
  \draw[toPart] (p2) -- (p2b);
  \draw[toPart] (p3) -- (p3b);
  \draw[toPart] (p4) -- (p4b);

\end{tikzpicture}
  \caption{Branching the solution space according to the refinement and the covering relations (see Equation~\ref{eq:raf_macro_rec})}
  \label{fig:algo_branching}
\end{figure}
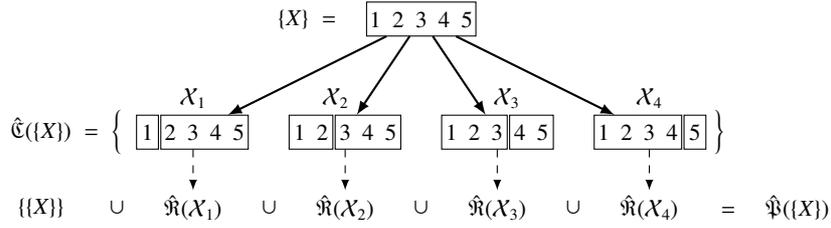

\paragraph{Recursive Algorithm}
\label{subsec:rec_algorithm}

Computing an optimal partition of $\pop$ thus consists in computing locally-optimal partitions among the ones that are refining the partitions covered by $\{\pop\}$.
Thanks to the principle of optimality, such a computation can be recursively performed by applying the same principle to the subsets of these covered partitions.
Hence, the branching and recursion equations (see Equations~\ref{eq:optprinciple}, \ref{eq:dec_generic}, and \ref{eq:raf_macro_rec})
allow to define a divide and conquer algorithm that computes a locally-optimal partition $\calX^* \in \hatfrakP^*(X)$ for any $X \in \hatcalP(\pop)$ by applying the following recursive formula:
\begin{equation}
  \label{eq:final_eq}
  \argmin_{\calX \, \in \, \{\{X\}\} \, \cup \, \left(\displaystyle\bigcup_{\calX \, \in \, \hatfrakC(\{X\})} \; \left\{ \displaystyle\bigcup_{{X^\pm} \in \calX} {\calX_{X^\pm}^*} \right\} \right)} \quad {c(\calX)} \quad \subseteq \quad \hatfrakP^*(X)
\end{equation}
where $\calX_X^*$ designates an optimal feasible partition of $X$, that is $\calX_X^* \in \hatfrakP^*(X)$.

Moreover, in the dynamic programming paradigm, recursive algorithms can be easily improved by \emph{memoization}, that is by recording the results of time-consuming recursive calls~\cite{Rothkopf98,Sandholm02}.
For each subset on which the algorithm is once applied, by keeping trace of the resulting locally-optimal partition, one can immediately return this result when posterior calls occur on the same subset.
This way, the algorithm is applied only once to each feasible subset $X \in \hatcalP(\pop)$.\\

The resulting algorithm is provided in Subsection~\ref{ssec:reduction} of this article.


\end{document}
\endinput